\documentclass[runningheads,orivec,colorlinks=true]{llncs}

\newcommand{\macrospath}{./macros}
\usepackage{ifthen}
\usepackage{xspace}

\newboolean{talk}
\setboolean{talk}{false}
\newboolean{paper}
\setboolean{paper}{false}

\newboolean{IEEEstyle}
\setboolean{IEEEstyle}{false}
\newboolean{lipicsstyle}
\setboolean{lipicsstyle}{false}
\newboolean{eptcsstyle}
\setboolean{eptcsstyle}{false}
\newboolean{entcsstyle}
\setboolean{entcsstyle}{false}
\newboolean{sigplanstyle}
\setboolean{sigplanstyle}{false}
\newboolean{easychairstyle}
\setboolean{easychairstyle}{false}
\newboolean{scrartclstyle}
\setboolean{scrartclstyle}{false}
\newboolean{lncsstyle}
\setboolean{lncsstyle}{false}

\newboolean{acmmart}
\setboolean{acmmart}{false}

\newboolean{needstheorems}
\setboolean{needstheorems}{false}
\newboolean{withimages}
\setboolean{withimages}{false}
\newboolean{withproofs}
\setboolean{withproofs}{true}

\newboolean{french}
\setboolean{french}{false}

\setboolean{lncsstyle}{true}
\setboolean{withproofs}{true}
 \usepackage{amsmath,amsfonts,amssymb,mathtools}
 \usepackage[T1]{fontenc}
\usepackage{graphicx}
\usepackage{hyperref}
\usepackage{color}

\usepackage{booktabs}   
\usepackage{subcaption} 
\usepackage[all]{xy}

\usepackage{ebproof}
\ebproofset{label separation=0.3em}

\usepackage{booktabs}   
\usepackage{subcaption} 
\usepackage[normalem]{ulem}
\usepackage{multirow}

\usepackage[utf8]{inputenc}

\usepackage{multirow}
\usepackage{cmll}
\usepackage{listings}

\usepackage{stmaryrd}
\usepackage{mathtools}
\usepackage{enumerate}
\usepackage{listings}

\usepackage{fancybox}
\usepackage{hhline}
\usepackage{marginnote}
\usepackage{soul}
\usepackage{xcolor}

\usepackage{cleveref}
\usepackage{complexity}

\newcommand{\sem}[1]{\llbracket#1\rrbracket}

\newcommand{\ignore}[1]{}

\newcommand{\colspace}{@{\hspace{.5cm}}}
\newcommand{\myinput}[1]{\ifthenelse{\boolean{withimages}}{\input{#1}}{}}

\newcommand{\reflemma}[1]{Lemma~\ref{l:#1}}
\newcommand{\reflemmap}[2]{Lemma~\ref{l:#1}.\ref{p:#1-#2}}

\newcommand{\refpropp}[2]{Prop.~\ref{prop:#1}.\ref{p:#1-#2}} 

\newcommand{\refsect}[1]{Sect.~\ref{sect:#1}}

\newcommand{\reffig}[1]{Fig.~\ref{fig:#1}}

\newcommand{\refrmk}[1]{Remark~\ref{rmk:#1}} 

\newcommand{\ie}{\textit{i.e.}\xspace}
\newcommand{\eg}{\textit{e.g.}\xspace}
\newcommand{\ih}{\textit{i.h.}\xspace}

\newcommand{\ES}{\text{ES}\xspace}
\newcommand{\ESs}{\text{ESs}\xspace}
\newcommand{\leftsh}{\text{left}\xspace}
\newcommand{\rightsh}{\text{right}\xspace}
\newcommand{\full}{\text{full}\xspace}
\newcommand{\Full}{\text{Full}\xspace}
\renewcommand{\full}{\text{strong}\xspace}
\renewcommand{\Full}{\text{Strong}\xspace}

\newcommand{\giulio}[1]{{#1}}
\newcommand{\cgiulio}[2]{\giulio{#2}}

\newcommand{\defeq}{\coloneqq} 
\newcommand{\grameq}{\Coloneqq} 
\newcommand{\set}[1]{\{#1\}}

\newcommand{\nat}{\mathbb{N}}
\newcommand{\size}[1]{|#1|}

\renewcommand{\l}{\lambda}
\newcommand{\isub}[2]{\{#1/#2\}}
\newcommand{\replace}[2]{#1{\shortleftarrow}#2}
\renewcommand{\isub}[2]{\{\replace{#1}{#2}\}}
\newcommand{\esub}[2]{[\replace{#1}{#2}]}
\renewcommand{\esub}[2]{[#1{\shortleftarrow}#2]}

\newcommand{\fv}[1]{{\sf fv}(#1)}

\newcommand{\rootRew}[1]{\mapsto_{#1}}

\newcommand{\Rew}[1]{\rightarrow_{#1}}

\newcommand{\lRew}[1]{\; \mbox{}_{#1}{\leftarrow}\ }

\newcommand{\lto}{\lRew{}}

\newcommand{\rtom}{\rootRew{\msym}}
\newcommand{\rtoe}{\rootRew{\esym}}

 \newcommand{\rtobv}{\rootRew{\betav}}
 
\newcommand{\rtoin}{\rootRew{\inert}}

\newcommand{\betav}{\beta_v} 
 \newcommand{\tobv}{\Rew{\betav}}

\newcommand{\betain}{\beta_{\isym}} 
\newcommand{\inert}{\betain} 
\newcommand{\toin}{\Rew{\inert}}

\newcommand{\esym}{{\mathsf e}}
\newcommand{\isym}{i}

\newcommand{\msym}{\mathsf{m}}

\newcommand{\wsym}{{\mathsf{o}}} 
\newcommand{\wmsym}{{\wsym\msym}} 
\newcommand{\omsym}{{\wmsym}} 

\newcommand{\shufeqext}{\shufeqext} 
\newcommand{\tsym}{{\mathsf t}} 

 \newcommand{\tom}{\Rew{\msym}}
 \newcommand{\toe}{\Rew{\esym}}

\newcommand{\tovsubo}{\Rew{\wsym}}

\newcommand{\tomo}{\Rew{\omsym}}

\newcommand{\toeo}{\Rew{\wsym{\esym}}}

\newcommand{\tovsubs}{\Rew{\esssym}}

\newcommand{\toms}{\Rew{\esssym{\msym}}}

\newcommand{\toes}{\Rew{\esssym{\esym}}}

\newcommand{\tm}{t}
\newcommand{\tmtwo}{u}
\newcommand{\tmthree}{r}
\newcommand{\tmfour}{q}

\newcommand{\tmp}{\tm'}
\newcommand{\tmtwop}{\tmtwo'}

\newcommand{\var}{x}
\newcommand{\vartwo}{y}
\newcommand{\varthree}{z}
\newcommand{\varfour}{w}

\newcommand{\val}{v}
\newcommand{\valtwo}{\val'}

\newcommand{\tval}{\val_\tsym}

\newcommand{\ctxholep}[1]{\langle #1\rangle}
\newcommand{\ctxhole}{\ctxholep{\cdot}}

\newcommand{\ctx}{C}
\newcommand{\ctxtwo}{\ctx'}

\newcommand{\ctxp}[1]{\ctx\ctxholep{#1}}
\newcommand{\ctxtwop}[1]{\ctxtwo\ctxholep{#1}}

\newcommand{\sctx}{L}

\newcommand{\sctxp}[1]{\sctx\ctxholep{#1}}

\newcommand{\arbctxp}[1]{\arbctxp{#1}}
\newcommand{\arbctxtwop}[1]{\arbctxtwop{#1}}

\newcommand{\genevctx}{E}

\newcommand{\evctx}{\genevctx}

\newcommand{\evctxp}[1]{\evctx\ctxholep{#1}}

\ifthenelse{\boolean{acmmart}
}{
  
  }{
  
  }

\newcommand{\deriv}{d}

\newcommand{\sizehole}[2]{|#2|_{#1}}

\newcommand{\sizem}[1]{\sizehole{\msym}{#1}} 

\ifthenelse{\boolean{IEEEstyle}}{
\usepackage{amssymb}
\usepackage{amsthm}
    \newtheorem{theorem}{Theorem}[section]
    \newtheorem{lemma}[theorem]{Lemma}

    \newtheorem{proposition}[theorem]{Proposition}
    \newtheorem{definition}[theorem]{Definition}
}{}

\ifthenelse{\boolean{lncsstyle}}{}{\newtheorem{remark}{Remark}[section]}

\newcommand{\itm}{i}
\newcommand{\itmtwo}{\itm'} 
\newcommand{\itmthree}{\itm''}

 \newcommand{\gconst}{i} 
\newcommand{\gconsttwo}{\gconst'}
\newcommand{\gconstthree}{\gconst''}

\newcommand{\fire}{f}
\newcommand{\firetwo}{\fire'}
\newcommand{\firethree}{\fire''}

\newcommand{\sfire}{\fire_\fullsym}
\newcommand{\sfiretwo}{\firetwo_\fullsym}

\newcommand{\vsub}{\mathsf{vsc}} 
\newcommand{\VSC}{\textnormal{VSC}\xspace}

\newcommand{\tovsub}{\Rew{\vsub}}

\newcommand{\tow}{\Rew{\wsym}} 

\newcommand{\osym}{{\mathsf o}}

\newcommand{\la}[1]{\lambda #1.}

\newcommand{\ictx}{R}

\newcommand{\ictxp}[1]{\ictx\ctxholep{#1}}

\newcommand{\myproof}[1]{
\ifthenelse{\boolean{omitproofs}}{\begin{IEEEproof} Proof available but omitted for readability. \end{IEEEproof}}{#1}}

\newcommand{\valES}{\text{answer}\xspace}

\newcommand{\gregoire}{Gr{\'{e}}goire\xspace}

\newcommand{\withproofs}[1]{\ifthenelse{\boolean{withproofs}}{#1}{}}
\newcommand{\withoutproofs}[1]{\ifthenelse{\boolean{withproofs}}{}{#1}}
\newcommand{\shortlong}[2]{\withoutproofs{#1}\withproofs{#2}}

\newcommand{\NoteProof}[1]{
	\marginnote{{\normalfont\scriptsize{Proof\,p.\,{\pageref{#1}}\,}}}}
\newcommand{\NoteState}[1]{
	\marginnote{{\normalfont\scriptsize{See p.\,{\pageref{#1}}}}}}
\renewcommand{\NoteProof}[1]{\marginnote{{Proof\,p.\,{\pageref{#1}}}}}
\renewcommand{\NoteState}[1]{\marginnote{{See\,p.\,{\pageref{#1}}\\\cref{#1}}}}

\crefname{proposition}{Prop.}{Props.}
\crefname{theorem}{Thm.}{Thms.}
\crefname{lemma}{Lemma}{Lemmas}
\crefname{corollary}{Cor.}{Cors.}
\crefname{section}{Sect.}{Sects.}
\Crefname{section}{Section}{Sections}

\newcommand{\firecalc}{\lambda_\mathsf{fire}}

\newcommand{\doubt}[1]{}

\newcommand{\letexp}{\mathsf{let}}
\newcommand{\letin}[3]{{\sf let}\ #1=#2\ {\sf in}\ #3}

\newcommand{\Rule}{\mathsf{r}}

\newcounter{numberone}
\newcounter{numberoneroman}
\newcounter{numberonealph}

\newcommand{\cbn}{CbN\xspace}

\newcommand{\cbv}{CbV\xspace}

\newcommand{\ocbv}{Open \cbv}

\newcommand{\ccbv}{Closed \cbv}
\newcommand{\scbv}{Strong \cbv}

\newcommand{\mset}[1]{[#1]}
\newcommand{\emptymset}{\mset{\,}}
\renewcommand{\emptymset}{\zero}
\newcommand{\zero}{\mathbf{0}}

\newcommand{\ltype}{\typefont{A}}
\newcommand{\ltypetwo}{\typefont{B}}

\newcommand{\typctx}{\Gamma}
\newcommand{\typctxtwo}{\Delta}
\newcommand{\typctxthree}{\Sigma}

\newcommand{\hastype}{\!:\!}

\newcommand{\domain}[1]{\mathsf{dom}(#1)}

\newcommand{\mytr}[1]{\underline{#1}}
\newcommand{\auxtr}[1]{\overline{#1}}

\makeatletter
\newcommand\Copy[2]{
        \marginpar{\scriptsize \ \ \hyperlink{hl-appendix-#1}{Proof p.\,{\pageref*{appendix-#1}}}}
	\immediate\write\@auxout{\unexpanded{\global\long\@namedef{mytext@#1}{#2}
  }}%
	#2%
}

\newcommand\Paste[1]{%
        \hypertarget{hl-appendix-#1}{}\label{appendix-#1}
	\renewcommand{\inappendix}[1]{}
	\ifcsname mytext@#1\endcsname
	\@nameuse{mytext@#1}%
	\else
	``??''
	\fi
	\renewcommand{\inappendix}[1]{#1}
}
\makeatother

\newcommand{\inappendix}[1]{#1}

\newcommand{\weakctx}{O}
\newcommand{\weakctxtwo}{\weakctx'}
\newcommand{\weakctxp}[1]{\weakctx\ctxholep{#1}}
\newcommand{\weakctxtwop}[1]{\weakctxtwo\ctxholep{#1}}

\newcommand{\subctx}{\sctx}

\newcommand{\strongctx}{\extctx}
\newcommand{\strongctxtwo}{\strongctx'}

\newcommand{\strongctxp}[1]{\strongctx\ctxholep{#1}}
\newcommand{\strongctxtwop}[1]{\strongctxtwo\!\ctxholep{#1}}

\newcommand{\extctx}{S}

\newcommand{\extctxp}[1]{\extctx\ctxholep{#1}}

\newcommand{\larrow}[2]{#1 \multimap #2}
\newcommand{\ground}{\typefont{G}}

\newcommand{\Ax}{\mathsf{ax}}
\newcommand{\Es}{\mathsf{es}}

\newcommand{\derive}[2]{#1 \vartriangleright #2}
\newcommand{\concl}[4]{\derive{#1}{#2 \vdash #3 \hastype #4}}

\newcommand{\subctxp}[1]{\subctx\ctxholep{#1}}

\newcommand{\Pointed}{Rigid\xspace}
\newcommand{\pointed}{rigid\xspace}
\newcommand{\ptm}{r}
\newcommand{\ptmtwo}{\ptm'}
\newcommand{\ptmthree}{\ptm''}

\newcommand{\esssym}{\mathsf{x}}

\newcommand{\toess}{\Rew{\esssym}}

\newcommand{\fullsym}{{\mathsf{f}}}
\renewcommand{\fullsym}{{\mathsf{s}}}
\newcommand{\sitm}{\itm_\fullsym}
\newcommand{\sitmtwo}{\itmtwo_\fullsym}
\newcommand{\sitmthree}{\itmthree_\fullsym}

\newcommand{\sval}{\val_\fullsym}

\newcommand\Crumb\mytr
\newcommand\CrumbAux\auxtr

\makeatletter
\newcommand{\exder}{%
  \def\exderW[##1]{\triangleright_{##1}\ }%
  \def\exderWO{\triangleright\ }%
  \@ifnextchar[\exderW\exderWO%
  }
\makeatother

\newcommand{\tderiv}{\Phi}
\newcommand{\tderivtwo}{\Psi}
\newcommand{\tderivtwop}{\tderivtwo'} 
\newcommand{\tderivthree}{\Theta}

\newcommand{\typefont}[1]{{\mathsf{#1}}}
\newcommand{\mtype}{\typefont{M}}
\newcommand{\mtypetwo}{\typefont{N}}
\newcommand{\mtypethree}{\typefont{O}}

\newcommand\emptytype{\mathbf{0}}

\newcommand\mplus{\uplus}

\newcommand{\tyjp}[4]{{#3} \vdash^{#1} #2 \hastype #4}
\newcommand{\namedtyjp}[5]{#1 \vartriangleright \tyjp{#2}{#3}{#4}{#5}}

\newcommand{\dom}[1]{\mathsf{dom}(#1)}

\newcommand{\ruleApp}{@}
\newcommand{\ruleFun}{\lambda}
\newcommand{\ruleES}{\mathsf{es}}
\newcommand{\ruleMany}{\mathsf{many}}
\newcommand{\ruleManyVar}{\ruleMany}
\newcommand{\ruleManyVal}{\ruleMany}
\newcommand{\ruleAp}{@}
\newcommand{\ruleAx}{\mathsf{ax}}

\newcommand{\I}{I}
\newcommand{\J}{J}
\renewcommand{\K}{K}
\newcommand{\iI}{{i \in \I}}
\newcommand{\jJ}{{j \in \J}}
\newcommand{\kK}{{k \in \K}}

\newcommand{\tarrow}[2]{#1 \multimap #2}
\newcommand{\ty}[2]{\tarrow{#1}{#2}}

\newcommand{\mult}[1]{[ #1 ] }

\newcommand{\bigmplus}{\biguplus}

\newcommand{\Id}{{\mathsf{I}}}

\newcommand{\rightsym}{{\textsc{r}}}
\newcommand{\rltype}{\ltype^{\rightsym}}

\newcommand{\rmtype}{\mtype^{\rightsym}}

\newcommand{\leftsym}{{\textsc{l}}}
\newcommand{\lltype}{\ltype^\leftsym}

\newcommand{\lmtype}{\mtype^\leftsym}

\newcommand\mydots{\hbox to .6em{.\hss.}}

\newcommand{\qedhere}{\qed}

\newcommand{\ctxeq}{\simeq_{\mathtt{C}}}
\newcommand{\ctxleq}{\precsim_{\mathtt{C}}}

\newcommand{\correction}[1]{\blue{#1}}
\renewcommand{\correction}[1]{#1}

\renewcommand{\tmthree}{s}
\renewcommand{\extctx}{X}

\begin{document}
\title{Strong Call-by-Value and Multi Types}
\author{Beniamino Accattoli\inst{1}\orcidID{0000-0003-4944-9944} \and Giulio Guerrieri\inst{2}\orcidID{0000-0002-0469-4279} \and Maico Leberle}
\authorrunning{B. Accattoli, G. Guerrieri, and M. Leberle}
%
\institute{Inria \& LIX, \'Ecole Polytechnique, UMR 7161, Palaiseau, France
\email{\href{mailto:beniamino.accattoli@inria.fr}{beniamino.accattoli@inria.fr}}
\and
Aix Marseille Univ, CNRS, LIS UMR 7020, Marseille, France
\email{\href{mailto:giulio.guerrieri@lis-lab.fr}{giulio.guerrieri@lis-lab.fr}} 
}
\maketitle

\begin{abstract}
This paper provides foundations for strong (that is, possibly under abstraction) call-by-value evaluation for the $\l$-calculus. Recently, Accattoli et al. proposed a form of call-by-value strong evaluation for the $\l$-calculus, the \emph{external strategy}, and proved it reasonable for time. Here, we study the external strategy using a semantical tool, namely Ehrhard's call-by-value multi types, a variant of intersection types. We show that the external strategy terminates exactly when a term is typable with so-called \emph{shrinking multi types}, mimicking similar results for strong call-by-\emph{name}. 
Additionally, the external strategy is normalizing in the untyped setting, that is, it reaches the normal form whenever it exists.

We also consider the \emph{call-by-extended-value} approach to strong evaluation shown reasonable for time by Biernacka et al. The two approaches turn out to \emph{not} be equivalent: terms may be externally divergent but terminating for call-by-extended-value.

%
\end{abstract}

\section{Introduction}
\label{sect:intro}

Plotkin's call-by-value $\l$-calculus $\lambda_v$ \cite{DBLP:journals/tcs/Plotkin75} is at 
the heart  of programming languages such as OCaml and proof assistants such as Coq. In the study of programming 
languages, call-by-value (shortened to \cbv) evaluation is usually \emph{weak}, that is, it does not reduce under 
abstractions, and terms are assumed to be \emph{closed}, \ie , without free variables.
These constraints give rise to an elegant framework---we call it \emph{Closed \cbv},  following 
Accattoli and Guerrieri 
\cite{DBLP:conf/aplas/AccattoliG16}.

Plotkin did not present the \cbv $\l$-calculus $\lambda_v$ with these restrictions, and properties such as confluence hold also without the restrictions. 
As soon as open terms are allowed, however, or evaluation is \emph{strong} (that is, it can reduce under abstractions), the calculus behaves badly at the semantical level. 
There are at least two issues, first pointed out by Paolini and Ronchi Della Rocca~\cite{DBLP:journals/ita/PaoliniR99,DBLP:conf/ictcs/Paolini01,parametricBook}.
\begin{enumerate}
\item \emph{False normal forms}: some terms are contextually equivalent to the looping term $\Omega \defeq (\la\var\var\var)(\la\var\var\var)$ and yet they are normal in Plotkin's setting. 
\item \emph{Failing of denotational soundness/adequacy beyond the closed case}: denotational models are usually both \emph{sound} (that is, denotations are stable by reduction: if $\tm\to\tmtwo$ then $\sem\tm=\sem\tmtwo$) and \emph{adequate} (that is, the denotation  $\sem\tm$ is non-empty if and only if the evaluation of $\tm$ terminates) only for \ccbv; at least one of the two properties fails in the open/strong case.
\end{enumerate}

\paragraph{Extensions of Plotkin's Call-by-Value.} 
A number of calculi extending Plotkin's $\lambda_v$ have been proposed. A first line of work studies a related and yet different issue of $\lambda_v$, namely the equational incompleteness with respect to continuation-passing translations, pointed out by Plotkin himself in \cite{DBLP:journals/tcs/Plotkin75}. 
This issue was solved with categorical tools by Moggi \cite{DBLP:conf/lics/Moggi89}, which led to a number of studies, among others
\cite{DBLP:journals/lisp/SabryF93,DBLP:journals/toplas/SabryW97,DBLP:journals/tcs/MaraistOTW99,DBLP:conf/icfp/CurienH00,DBLP:journals/logcom/DyckhoffL07,DBLP:conf/tlca/HerbelinZ09}, that 
introduced many proposals of improved calculi~for~\cbv.

A second and more recent line of work, due to Accattoli, Guerrieri, and coauthors, addresses the problem of open terms and strong evaluation directly \cite{AccattoliPaolini12,DBLP:conf/fossacs/CarraroG14,DBLP:journals/tcs/Accattoli15,DBLP:conf/aplas/AccattoliG16,DBLP:journals/lmcs/GuerrieriPR17,DBLP:conf/aplas/AccattoliG18,DBLP:journals/scp/AccattoliG19,DBLP:conf/lics/AccattoliCC21,DBLP:journals/pacmpl/AccattoliG22}. It builds on the work of Paolini and Ronchi Della Rocca and on tools and techniques coming from the theory of Girard's linear logic \cite{DBLP:journals/tcs/Girard87}. 

In \cite{DBLP:conf/aplas/AccattoliG16}, they compare four different extensions of Plotkin's calculus in the framework of weak evaluation with possibly open terms. Their result is that the four calculi are all \emph{termination equivalent}: $\tm$ terminates in one of these extensions if and only if terminates in the other ones. In particular, in these extensions the issue of \emph{false normal forms} is solved because all terms contextually equivalent to $\Omega$ do diverge, in contrast to what happens in Plotkin's calculus $\lambda_v$. 
The notion of termination shared by the four extension is then referred to as \emph{\ocbv}. 

One of the aims of this paper is identifying an analogous notion of termination for strong \cbv evaluation. Perhaps surprisingly, indeed, the termination equivalent calculi of \ocbv do not agree on what such a notion should be.

\paragraph{Two Relevant Extensions.} Two\cgiulio{ of the four}{} \ocbv calculi in \cite{DBLP:conf/aplas/AccattoliG16} are relevant \cgiulio{for this paper}{here}. The first one is a \emph{call-by-extended-values} $\l$-calculus where the restriction on $\beta$-redexes \emph{by value} is weakened to $\beta$-redexes having as argument an extended, more general notion of \emph{value}. 
\cgiulio{It was first}{First} used as a nameless technical tool by Paolini and Ronchi Della Rocca~\cite{DBLP:journals/ita/PaoliniR99,parametricBook}, then rediscovered by Accattoli and Sacerdoti Coen \cite{fireballs} to study cost models\correction{, it \cgiulio{brings}{has} some similarities with a calculus introduced by \gregoire and Leroy \cite{DBLP:conf/icfp/GregoireL02} to study a \cbv abstract machine for Coq}.
In \cite{fireballs}, extended values are called \emph{fireballs} (a pun on \emph{fire-able}) and the calculus is called \emph{fireball calculus}.

The second extension is the \emph{value substitution calculus} (shortened to VSC) due to Accattoli and Paolini and related to linear logic proof nets \cite{AccattoliPaolini12,DBLP:journals/tcs/Accattoli15}. It was introduced to overcome some of the semantical problems of Plotkin's setting, and it is a flexible tool, used in particular to relate the four extensions in \cite{DBLP:conf/aplas/AccattoliG16}.

\paragraph{Beyond False Normal Forms.} In later works \cite{DBLP:conf/aplas/AccattoliG18,DBLP:journals/pacmpl/AccattoliG22}, Accattoli and Guerrieri show that the termination equivalence of \ocbv does not necessarily solve the other semantical issue of \ocbv, namely \emph{the failing of denotational soundness/adequacy beyond the closed case}. On the one hand, they show that the fireball calculus is adequate but \emph{not sound} with respect to Ehrhard's \cbv relational model \cite{DBLP:conf/csl/Ehrhard12}, a paradigmatic model arising from the theory of linear logic and handily presented as a \emph{multi types system} (a variant of intersection types). On the other hand, they show that the open VSC is \emph{both} adequate and sound with respect to that model, suggesting that it is a better setting for \ocbv.

\paragraph{Strong Call-by-Value.} The strong case has received less attention. In particular, it is not even clear what is the right notion of termination. The recent literature contains two proposals of strong \cbv evaluation, which have been carefully studied from the point of views of abstract machines and reasonable cost models, but not from a semantical point of view. 
The first one is the \emph{strong fireball calculus}, for which abstract machines have been recently designed in 2020 \cite{DBLP:conf/aplas/BiernackaBCD20} and 2021 \cite{DBLP:conf/ppdp/BiernackaCD21} by Biernacka et al., the latter being reasonable for time (defined as the number of $\beta$-steps). The second proposal is the strong VSC, and more precisely the \emph{external strategy} of the strong VSC, introduced in 2021 by Accattoli et al. \cite{DBLP:conf/lics/AccattoliCC21}, together with a reasonable machine implementing it.

The works \cite{DBLP:conf/ppdp/BiernackaCD21} and \cite{DBLP:conf/lics/AccattoliCC21} have been developed independently and at the same time, by two distinct groups, who cite each other. 
They state that both implement \emph{Strong \cbv}, but they fail to notice that they implement \emph{different notions of termination}, raising the question of what exactly should be considered as Strong \cbv. 
As we point out here, indeed, some terms are normalizing in the strong fireball calculus but have no normal form with respect to the external strategy. 

\paragraph{Strong Call-by-Value and Multi Types.} To \cgiulio{help clarifying}{clarify} the situation, we here explore the semantic perspective provided by \cbv multi types. For such types, typability coincides with termination of \emph{open} \cbv evaluation, as shown by Accattoli and Guerrieri \cite{Guerrieri18,DBLP:conf/aplas/AccattoliG18,DBLP:journals/pacmpl/AccattoliG22}, so they do not directly model strong evaluation. A similar mismatch happens\cgiulio{ also}{} in call-by-name (\cgiulio{shortened to \cbn}{\cbn for short}), where terms typable with multi types coincide with the \emph{head} (rather than strong) terminating ones. It is well known, however, that the restriction to \emph{shrinking types} (that have no negative occurrences of the empty multiset) does model strong evaluation: in \cbn, terms typable with shrinking types coincide with the leftmost(-outermost) terminating ones, and the leftmost strategy is \cgiulio{the}{a} normalizing strategy of Strong \cbn. 
Such a use of \cgiulio{shrinking types}{shrinkingness} is standard in the theory of intersection and multi types, \cgiulio{see for instance Krivine's book 
\cite{DBLP:books/daglib/0071545} or, more recently, de Carvalho \cite{deCarvalho18}, Kesner and Ventura \cite{DBLP:conf/ifipTCS/KesnerV14}, or Bucciarelli et al. \cite{DBLP:journals/igpl/BucciarelliKV17}}{see Krivine
\cite{DBLP:books/daglib/0071545}, de Carvalho \cite{deCarvalho18}, Kesner and Ventura \cite{DBLP:conf/ifipTCS/KesnerV14}, Bucciarelli et al. \cite{DBLP:journals/igpl/BucciarelliKV17}}. 

Here, we adapt to \cbv the shrinking technique as presented by Accattoli et al. for \cbn multi types in \cite{DBLP:journals/pacmpl/AccattoliGK18}, where the \emph{shrinking} terminology is also introduced. Our main result is the characterization of external termination via shrinking types: a term $\tm$ is typable with shrinking \cbv multi types \emph{if and only if} the external strategy terminates on $\tm$. Technically, the result is a smooth adaptation of the technique in \cite{DBLP:journals/pacmpl/AccattoliGK18}. Smoothness is here a \emph{plus}, as it shows that the external strategy is the notion of \scbv termination \emph{naturally} validated by \cbv multi types, without ad-hoc stretchings of the technique.

\paragraph{Untyped Normalization Theorem.} In an untyped setting, not every term normalizes (think of $\Omega$) and in the strong case some terms have both reductions that normalize and reductions that diverge, for instance $(\la\var\vartwo) (\la\varthree\Omega)$. Thus, it is important to have a strategy that reaches a normal form whenever possible, i.e., that is \emph{normalizing} in an untyped setting. The canonical evaluation strategy in Strong \cbn is  leftmost reduction and its key property is precisely that it is normalizing. A further contribution of the paper is an \emph{untyped normalization theorem} for the external strategy in the Strong VSC, obtained as an easy corollary of the study via multi types. Such a result gives to the external strategy the same solid status of the leftmost strategy in \cbn, and completes the picture.

\paragraph{No Tight Bounds.} Multi types can be used to extract \emph{tight bounds} on the length of evaluations and the size of normal forms. Here, we only study termination, not tight bounds, even if in the technical report \cite{accattoli2021semantic} we  also developed the enriched results with tight bounds. A first reason is that the characterization of external termination and the untyped normalization theorem we focus on here do not need the bounds. 
A second reason is that the enriched results are considerably more technical, while here we aim at a \mbox{slightly weaker but more accessible treatment}. 

\paragraph{Proofs.} \shortlong{
	\giulio{Omitted proofs are in \cite{ICTAClong}, the long version of this paper.}
}{Most proofs are \giulio{not in the body of the paper, but} in the Appendix. 
}


\section{Technical Preliminaries}
\label{sect:app-preliminaries}

\paragraph{Basic Rewriting Notions.}
For a relation $R$ on a set of terms, $R^*$ is its reflexive-transitive closure. 
Given a relation $\Rew{\Rule}$, an $\Rule$-\emph{evaluation}
(or simply evaluation if unambiguous) $\deriv$ is a finite sequence of terms $(\tm_i)_{0 \leq i \leq n}$ (for some $n \geq 0$) such that $\tm_i \Rew{\Rule} \tm_{i+1}$ for all $1 \leq i < n$, and we write $\deriv \colon \tm \Rew{\Rule}^* \tmtwo$ if $\tm_0 = \tm$ and $\tm_n = \tmtwo$. The
\emph{length} $n$ of $\deriv$ is noted $\size{\deriv}$, and $\size{\deriv}_a$ is the number of $a$-\emph{steps} (\ie the number of $\tm_i \Rew{a} \tm_{i+1}$ for some $1 \leq i \leq n$) in $\deriv$, for a given subrelation $\Rew{a}$ of $\Rew{\Rule}$.

A term $\tm$ is $\Rule$-\emph{normal} if there is no $\tmtwo$ such that $\tm \Rew{\Rule} \tmtwo$.
An evaluation $\deriv \colon \tm \Rew{\Rule}^* \tmtwo$ is \emph{$\Rule$-normalizing} if $\tmtwo$ is $\Rule$-normal.
A term $\tm$ is \emph{weakly $\Rule$-normalizing} if there is a $\Rule$-normalizing evaluation $\deriv \colon \tm \Rew{\Rule}^* \tmtwo$; and $\tm$ is \emph{strongly $\Rule$-normalizing} if there no infinite sequence $(\tm_i)_{i \in \nat}$  such that $\tm_0  = \tm$ and $\tm_i \Rew{\Rule} \tm_{i+1}$ for all $i \in \nat$.
Clearly, strong $\Rule$-normalization implies weak $\Rule$-normalization.

\paragraph{The Diamond Property.} \correction{Following Dal Lago and Martini \cite{DBLP:journals/tcs/LagoM08}, we say that a relation $\Rew{\Rule}$ is \emph{diamond} if $\tmtwo_1 \,{}_\Rule\!\!\lto \tm \Rew{\Rule} \tmtwo_2$ and $\tmtwo_1 \neq \tmtwo_2$ imply $\tmtwo_1 \Rew{\Rule} \tmthree \, {}_\Rule\!\!\lto \tmtwo_2$ for some $\tmthree$. 
	\cgiulio{Note that the t}{T}erminology in the literature is inconsistent: Terese \cite[Exercise 1.3.18]{Terese} dubs this property $\textup{CR}^1$, and defines the diamond more restrictively, without requiring $u_1 \neq u_2$ in the hypothesis: $u_1$ and $u_2$ have to join even if $u_1 = u_2$.

Dal Lago and Martini show that if $\Rew{\Rule}$ is diamond then:}
\begin{enumerate}
	\item \cgiulio{\emph{Confluence}:}{}
	$\Rew{\Rule}$ is confluent, that is, $\tmtwo_1 \,{}_\Rule^*\!\!\lto \tm \Rew{\Rule}^* \tmtwo_2$  implies $\tmtwo_1 \Rew{\Rule}^* \tmthree \, {}_\Rule^*\!\!\lto \tmtwo_2$ for some $\tmthree$; 
	\item \cgiulio{\emph{\correction{Length invariance}}:}{} all $\Rule$-evaluations with the same start and end terms have the same length (\ie if $\deriv \colon \tm \Rew{\Rule}^* \tmtwo$  and $\deriv' \colon \tm \Rew{\Rule}^* \tmtwo$ then $\size{\deriv} = \size{\deriv'}$);
	\item \cgiulio{\emph{Uniform normalization}:}{} $\tm$ is weakly $\Rule$-normalizing if and only if it is strongly $\Rule$-normalizing.
\end{enumerate}
\giulio{Properties 2 and 3 are called \emph{length invariance} and \emph{uniform normalization}, respectively.}
\cgiulio{Essentially}{Basically}, the diamond\cgiulio{ property}{} captures a more liberal form of determinism.

\paragraph{Contextual Equivalence.} The standard of reference for program equivalences is contextual equivalence, that can be defined abstractly as follows.
\begin{definition}[Contextual Preorder and Equivalence] Given a language of terms $\mathcal{T}$ with its associated notion of contexts $\ctx$ and an operational semantics given as a rewriting relation $\to$, we define the associated \emph{contextual preorder} $\ctxleq$ and \emph{contextual equivalence} $\ctxeq$ as follows:
	\begin{itemize}
		\item $\tm \ctxleq \tmp$ if $\ctxp\tm$  weakly $\to$-normalizing implies $\ctxp\tmp$ weakly $\to$-normalizing for all contexts $\ctx$ such that $\ctxp{\tm}$ and $\ctxp\tmp$ are closed terms. 
		\item $\ctxeq$ is the equivalence relation induced by $\ctxleq$: $\tm \ctxeq \tmp \iff \tm \ctxleq \tmp$ and $\tmp \ctxleq \tm$.
	\end{itemize}
\end{definition}

\section{Call-by-Value and Call-by-Fireball}
\label{sect:cbv}
The call-by-value $\l$-calculus \giulio{$\lambda_v$} was introduced by Plotkin \cite{DBLP:journals/tcs/Plotkin75} in 1975 as the restriction of the $\l$-calculus where $\beta$-redexes can be fired only when their argument is a \emph{value}, \cgiulio{where}{and} values are defined as variables and abstractions. 
\begin{figure}[t]
  \centering
\begin{tabular}{c}
  
   $
    \begin{array}{r@{\hspace{.5cm}}rll}
	    	    \textsc{Terms} & \tm,\tmtwo,\tmthree,\tmfour &\grameq& \var \mid  \la\var\tm \mid \tm\tmtwo\\
    
    \textsc{Values} & \val,\valtwo & \grameq &  \la\var\tm\\
        
	    \textsc{Evaluation contexts} & \evctx  & \grameq & \ctxhole\mid \tm\evctx \mid \evctx\tm \\\\	
\end{array}$
\\
	    $\begin{array}{@{\hspace{.8cm}}r@{\hspace{.3cm}}c@{\hspace{.3cm}}l@{\hspace{.8cm}}rll}
      \multicolumn{3}{c}{\textsc{Rule at top level}} & \multicolumn{3}{c}{\textsc{Contextual closure}} \\
	    (\la\var\tm)\val & \rtobv & \tm\isub\var{\val} &
	    \evctxp \tm & \tobv & \evctxp \tmtwo \textrm{~~~if } \tm \rtobv     \tmtwo \\
    \end{array}
    $
    \end{tabular}
  \caption{\label{fig:plotkin} Our presentation of Plotkin's calculus \giulio{$\lambda_v$}.}
\end{figure}

\paragraph{Our Presentation of \cbv.} \correction{In \reffig{plotkin}, we present \cgiulio{Plotkin's calculus}{$\lambda_v$} adopting three specific choices, departing from  Plotkin's presentation \cite{DBLP:journals/tcs/Plotkin75} in inessential details. Firstly, Plotkin also considers constants in the term syntax, with their own reduction rules, which are \cgiulio{however}{} left unspecified. 
	For simplicity, in our presentation there are no constants. 

Secondly, for Plotkin values are variables and abstractions, while here \emph{values are only abstractions}, as it is the case in the papers about Strong \cbv motivating our study \cite{fireballs,DBLP:conf/aplas/BiernackaBCD20,DBLP:conf/ppdp/BiernackaCD21}. As stressed by Accattoli and Guerrieri \cite{DBLP:journals/pacmpl/AccattoliG22}, removing variables from values provides a better inductive description of normal forms in the open/strong case, without affecting properties such as termination or confluence. For further quantitative benefits, see Accattoli and Sacerdoti Coen \cite{DBLP:journals/iandc/AccattoliC17}. 
In our paper, variables are not values, but all our results could be restated by considering them~as~values.

Thirdly, Plotkin defines both a multi-step non-deterministic (and confluent) evaluation relation reducing redexes everywhere in the term and a single-step deterministic reduction (proceeding left-to-right) that is \emph{weak}, that is, it does not reduce under abstractions. 
We instead adopt the somewhat halfway approach by Dal Lago and Martini \cite{DBLP:journals/tcs/LagoM08}, which is single-step and weak but \emph{non-deterministic}. The idea is that it can evaluate the left and right sub-term of an application in any order (that is, the left sub-term is not forced to be evaluated first as Plotkin does), because (in the weak case) the obtained reduction relation has the \emph{diamond property} (definition in \Cref{sect:app-preliminaries}), a relaxed form of determinism. Thus, the obtained notion of reduction slightly generalizes Plotkin's single-step reduction without changing whether a term terminates or not. The non-deterministic version is obtained via the notion of \emph{evaluation context} $\evctx$ defined in \reffig{plotkin}.}

\paragraph{Problem with Open Terms.} It is well known that Plotkin's framework works well only as long as terms are closed. The problem with open terms is that there are \emph{false normal forms} such as $\Omega_l \defeq (\la\var\delta)(\vartwo\vartwo)\delta$ (where $\delta\defeq \la\var\var\var$ is the duplicator) which are $\tobv$-normal, because $\vartwo\vartwo$ is not a value (and cannot become one), but are \emph{semantically divergent}. Such a divergence can be formalized in various ways, perhaps the simplest of which is that $\Omega_l$ is contextually equivalent (definition in \Cref{sect:app-preliminaries}) to the diverging term $\Omega \defeq \delta\delta$. This problem was first pointed out by Paolini and Ronchi Della Rocca
\cite{DBLP:journals/ita/PaoliniR99,DBLP:conf/ictcs/Paolini01,parametricBook}. It is discussed at length by Accattoli and Guerrieri \cite{DBLP:conf/aplas/AccattoliG16}, who analyze various ways of extending Plotkin's calculus to solve the issue, that is, as to make false normal forms such as $\Omega_l$ diverge without simply switching to \cbn (which, intuitively, amounts to diverge on $(\la\var\vartwo) \Omega$ as done by \cbv but not by \cbn), and show their equivalence with respect to termination, referring to them collectively as \emph{\ocbv}.

\begin{figure}[t]
  \centering
\begin{tabular}{c}
  
   $
    \begin{array}{r@{\hspace{.5cm}}rll}
	    	    \textsc{Terms} \mbox{, } \textsc{Values}\mbox{, }\textsc{Evaluation ctxs}\mbox{, and}\tobv& \multicolumn{3}{l}{\mbox{As for Plotkin's calculus}}
		    \\
    
    \textsc{Fireballs} & \fire, \firetwo, \firethree & \grameq & \val \mid \gconst\\
	    
	    \textsc{Inert terms} & \gconst,\gconsttwo, \gconstthree & \grameq &  \var \fire_1\ldots \fire_n\ \ \ \ n\geq 0\\\\
\end{array}$
\\
	    $\begin{array}{@{\hspace{.8cm}}r@{\hspace{.3cm}}c@{\hspace{.3cm}}l@{\hspace{.8cm}}rll}
      \multicolumn{3}{c}{\textsc{Rule at top level}} & \multicolumn{3}{c}{\textsc{Contextual closure}} \\
	    (\l\var.\tm)\gconst & \rtoin & \tm\isub\var\gconst &
	    \evctxp \tm & \toin & \evctxp \tmtwo \quad\textup{if } \tm \rtoin \tmtwo
%
    \end{array}
    $
    \end{tabular}
  \caption{\label{fig:fireball} The fireball calculus $\firecalc$.}
\end{figure}

\paragraph{The Fireball Calculus.} The simplest presentation of \ocbv is probably the fireball calculus, defined in \reffig{fireball}.
The idea is that the values of the \cbv $\l$-calculus are generalised to \emph{fireballs} (a pun on \emph{fire-able} terms), by adding \emph{inert terms}, which contain in particular variables. 
Actually fireballs and inert terms  are defined by mutual induction (in \reffig{fireball}). 
For instance, \correction{$\la\var\tm$ is a fireball as a value}, while $x$, $\vartwo(\la\var\var)$, $\var\vartwo$ and $\varthree(\la\var\var)(\varthree\varthree) (\la\vartwo\tm)$ are fireballs as inert terms. 

The main feature of inert terms is that they are open, normal, and that when plugged in a context they cannot create a redex, hence the name. Essentially, they are the \emph{neutral terms} of \ocbv. 


Dynamically, $\beta$-redexes can also be fired if their argument is an inert term, via $\toin$. 
Evaluation is weak, as evaluation contexts do not go under abstractions. 
Note that\cgiulio{ the term}{} $\Omega_l$ given above now diverges:
$
\Omega_l = (\la\var\delta)(\vartwo\vartwo)\delta \toin \delta\delta \tobv \ldots$

Two relevant properties of the fireball calculus are that its reduction is non-deterministic but diamond and that a term is a normal form if and only if it is a fireball, a property called \emph{harmony} by Accattoli and Guerrieri \cite{DBLP:conf/aplas/AccattoliG16}.

\paragraph{Issues with the Fireball Calculus.} In later works \cite{DBLP:conf/aplas/AccattoliG18,DBLP:journals/pacmpl/AccattoliG22},  Accattoli and Guerrieri show that, despite the termination equivalence of the various formalisms for \ocbv, they can behave quite differently with respect to semantical notions. In particular, they show that the fireball calculus does not behave well with respect to Ehrhard's multi types \cite{DBLP:conf/csl/Ehrhard12}, which are a notion of type inducing a paradigmatic denotational model of the $\l$-calculus, the \emph{relational model}, linked to linear logic. Technically, they show that multi types do not satisfy subject reduction with respect to the fireball calculus. The issue can be illustrated without multi types, and amounts to the fact that the fireball calculus allows one to erase and duplicate inert terms. The erasure is particularly problematic. We reformulate here the problem in terms of contextual equivalence, to stress that the issue concerns the fireball calculus independently of multi types. 

An expected  property of any calculus is what we call here \emph{contextual stability}, that is, the fact that its operational semantics $\to$ is included in its contextual equivalence $\ctxeq$---in symbols: if $\tm \to^* \tmtwo$ then $\tm \ctxeq \tmtwo$. In the fireball calculus, contextual stability fails. 
Consider $\tm \defeq (\la\var\Id)(\vartwo\vartwo)$, where $\Id\defeq \la\varthree\varthree$ is the identity combinator, and note that $\tm \toin \Id$ because $\vartwo\vartwo$ is an inert term. 
Now, consider the closing context $\ctx \defeq (\la\vartwo\ctxhole)\delta$ and note that $\ctxp\Id  =  (\la\vartwo\Id)\delta  \tobv  \Id$,~while
\begin{equation*}
	\ctxp\tm =  (\la\vartwo((\la\var\Id)(\vartwo\vartwo)))\delta  \tobv  (\la\var\Id)(\delta\delta) \tobv  (\la\var\Id)(\delta\delta) \tobv\ldots 
\end{equation*}
That is, $\tm \toin \Id$ but $\tm\not\ctxeq \Id$ because $\ctxp\tm$ diverges while $\ctxp\Id$ terminates.

To overcome false normal forms such as $\Omega_l$, one \emph{has to} work around redexes having inert terms has arguments. But substituting inert terms (thus sometimes erasing them) as done by the fireball calculus is a \emph{brute force} solution. 
At the open level, it does not alter termination but it deteriorates semantics properties of the calculus such as contextual stability or the relationship with multi types. 
The alternative presentation of \ocbv given by the value substitution calculus \cite{AccattoliPaolini12}, discussed in the next section, does not suffer of these shortcomings.

\paragraph{Issues with the Strong Fireball Calculus.} The strong evaluation strategy 
by Biernacka et al. \cite{DBLP:conf/aplas/BiernackaBCD20,DBLP:conf/ppdp/BiernackaCD21} is a deterministic version of the fireball calculus (namely proceeding right-to-left) extended to evaluate under abstractions. 
We omit the actual definition of its extension under abstractions because it is non-trivial and the issue we want to point out can be explained without detailing them.

At the open level, the issue with the fireball calculus is about the semantics, but not about termination. The semantical issue of the open level, however, induces a termination issue at the strong level. Indeed, the strong fireball calculus suffers of a phenomenon similar to the one of false normal forms, even if no redexes are stuck. Consider the term $\tmtwo \defeq (\la\var\Id) (\vartwo (\la\varthree\Omega))$. In the fireball calculus, one has $\tmtwo \toin \Id$ because $\vartwo (\la\varthree\Omega)$ is an inert term. That is, $\tmtwo$ terminates on a strong normal form, namely $\Id$. Instead, $\tmtwo$ is \emph{semantically divergent} at the strong level, as it was the case for $\Omega_l$ in Plotkin's calculus. Intuitively, the inert sub-term $\vartwo (\la\varthree\Omega)$ should be somehow kept, instead of being erased, and evaluated strongly, which would lead to divergence because of the $\Omega$ under abstraction. 

The non-trivial point is how to bring evidence that this is what should happen in a good definition of strong \cbv evaluation.
One way to see that something is wrong with the strong fireball calculus is to observe that the step $\tmtwo \toin\Id$ provides another breaking of contextual stability. Such a breaking is detailed in \Cref{sect:app-cbv-cbfireball}, because although interesting it does not involve strong evaluation.  

Better evidence is developed along the paper. We adopt the value substitution calculus (VSC) and the external strategy by Accattoli et al. \cite{DBLP:conf/lics/AccattoliCC21} for which $\tm$ is divergent. Then, we show that, when refining Ehrhard's \cbv multi types \cite{DBLP:conf/csl/Ehrhard12} by adding the machinery for characterizing termination of strong evaluation, one obtains that $\tm$ above is not typable. That is, $\tm$ is also semantically diverging. 

\section{Value Substitution Calculus}
\label{sect:vsc}

Here we present the \emph{value substitution calculus} (\VSC for short) introduced by Accattoli and Paolini \cite{AccattoliPaolini12}, and recall some properties. The operational semantics shall be different from the fireball calculus of the previous section, but we shall nonetheless exploit the concept of fireball to describe its normal forms.

\paragraph{Terms.} The \VSC is a \cbv $\lambda$-calculus extended with $\letexp$-expressions, similarly to Moggi's \cbv calculus \cite{Moggi88tech,DBLP:conf/lics/Moggi89}. 
We do however write a $\letexp$-expression $\letin\var\tmtwo\tm$ as a more compact  \emph{explicit substitution} $\tm\esub{\var}{\tmtwo}$ (\ES for short), which binds $\var$ in $\tm$. Moreover, our $\letexp$/ES does not fix an order of evaluation between $\tm$ and $\tmtwo$, in contrast to many papers in the literature (\eg Sabry and Wadler \cite{DBLP:journals/toplas/SabryW97} or Levy et al. \cite{DBLP:journals/iandc/LevyPT03}) where $\tmtwo$ is evaluated first. The grammars follow:
\begin{center}
$\arraycolsep=3pt\begin{array}{rrl\colspace \colspace\colspace rrl}
\textsc{Values } & \val & \grameq \la\var\tm 
&
\textsc{Terms } & \tm,\tmtwo, \tmthree & \grameq \var \mid \val \mid \tm\tmtwo 
\mid \tm \esub\var\tmtwo 
\end{array}$
\end{center}
The set of free variables of term $\tm$ is denoted by $\fv{\tm}$ and terms are identified up to $\alpha$-renaming of bound variables. We use $\tm \isub{\var}{\tmtwo}$ for the capture-avoiding substitution of
$\tm$ for each free occurrence of $\var$ in $\tm$.

\paragraph{Contexts.} All along the
paper, we use (many notions of) \emph{contexts}, \ie terms with a hole, noted $\ctxhole$. For now, we need general contexts and the notion of \emph{substitution contexts} $\sctx$, which are simply lists of \ES. The grammars are:
\begin{center}
$\arraycolsep=3pt\begin{array}{r\colspace rl}
\textsc{(General) Contexts} & \ctx & \grameq \ctxhole \mid \ctx \tm \mid \tm \ctx \mid \la{\var}{\ctx} \mid %
\ctx \esub\var\tm \mid \tm \esub \var \ctx 
\\
\textsc{Substitution Contexts } &\subctx & \grameq \ctxhole \mid \subctx \esub\var\tm
\end{array}$
\end{center}

Plugging a term $\tm$ in a context $\ctx$ is noted $\ctxp{\tm}$, possibly~capturing variables, for instance $(\la\var\la\vartwo\ctxhole)\ctxholep{\var\vartwo} = \la\var\la\vartwo\var\vartwo$ (while $(\la\var\la\vartwo\varthree)\isub\varthree{\var\vartwo} = \la{\var'}\la{\vartwo'}\var\vartwo$).
An \emph{answer} is a term of shape $\sctxp\val$.

\paragraph{Reduction Rules.}The reduction rules of \VSC are slightly unusual as they use \emph{contexts} both to allow one to reduce redexes located in sub-terms, which is standard, \emph{and} to define the redexes themselves, which is less standard---these kind of rules is 
called \emph{at a distance}. The rewriting rules in fact mimic exactly cut-elimination on proof nets, via Girard's \cbv translation $(A \Rightarrow B)^v = \oc (A^v \multimap B^v)$ \cite{DBLP:journals/tcs/Girard87} of intuitionistic logic into linear logic, see Accattoli \cite{DBLP:journals/tcs/Accattoli15}.  

There are two rewrite rules. Their root cases (that is, before context closure) follow (the terminology is inherited from linear logic): 
\begin{center}
$\begin{array}{r\colspace r@{\ }l@{\ }l}
    \textsc{Multiplicative root rule} & \subctxp{\la\var\tm}\tmtwo &  \rtom  & \subctxp{\tm\esub{\var}{\tmtwo}} \\
    \textsc{Exponential root rule}  & \tm\esub\var{\subctxp{\val}} &  \rtoe  & \subctxp{\tm\isub{\var}{\val}} 
\end{array}$    
\end{center}
Both root rules are \emph{at a distance} in that they involve a substitution context $\subctx$.
We shall consider two variants of the  \VSC, the open and the strong version. They differ only in the choice of evaluation contexts for the root rewrite rules.

\smallskip
\paragraph{The Open VSC.} We first focus on the open fragment of the \VSC, where rewriting is forbidden under abstraction and terms are \emph{possibly} open (but not necessarily). This fragment has a nice inductive description of its normal forms, called \emph{fireballs} and \emph{inert terms} as they are the lifting to the VSC of the concepts of the fireball calculus. Open contexts and rules are defined as follows.
\begin{center}
		$\begin{aligned}
		\textsc{Open contexts} && \weakctx & \grameq \ctxhole \mid \weakctx \tm \mid \tm \weakctx \mid \weakctx \esub\var\tm \mid 
\tm\esub\var \weakctx 
		\end{aligned}$
		\smallskip
		
				\begin{tabular}{c\colspace\colspace\colspace cc}
		\begin{tabular}{ccc}
\textsc{Open rewrite rules}:
			&
			\multirow{2}{*}{\begin{prooftree}
					\hypo{\tm \rootRew{a} \tm'}		
					\infer1{\weakctxp{\tm} \Rew{\wsym a} \weakctxp{\tm'}}
			\end{prooftree}}
		\\
		($a \in \set{\msym,\esym}$)
	\end{tabular}
	&

		$\begin{array}{cccc}
\multicolumn{3}{c}{\textsc{Open reduction}:}
\\ \tovsubo  \, \defeq \, \tomo \cup \toeo
	\end{array}$
		\end{tabular}
\end{center}

\begin{proposition}[Properties of open reduction \cite{DBLP:journals/pacmpl/AccattoliG22}]
	\label{prop:ovsc-diamond}\label{prop:properties-open-reduction}
	\begin{enumerate}
		\item \label{p:properties-open-reduction-diamond} Open reduction $\tovsubo$ is diamond.
		\item \label{p:properties-open-redction-harmony} A term is $\osym$-normal if and only if it is a fireball, where \emph{fireballs} (and \emph{inert terms}) are defined by:
	\end{enumerate}
\begin{center}
	$\begin{array}{r\colspace r@{\ } l\colspace\colspace\colspace r\colspace r@{\ } l}
	\textsc{Inert terms} & \itm &\grameq \var \mid \itm \fire \mid \itm \esub{\var}{\itmtwo}
	&
	\textsc{Fireballs} & \fire &\grameq \val \mid \itm \mid \fire \esub{\var}{\itm}
	\end{array}$
\end{center}
\end{proposition}

\paragraph*{Plotkin \textit{vs} VSC.}
The open fragment of the VSC is enough to discuss the relationship with Plotkin's \cbv $\lambda$-calculus as defined in the previous section. Plotkin's calculus can be easily simulated in the \VSC. Indeed, if $(\la\var\tm)\val \rtobv \tm\isub\var\val$ then $(\la\var\tm)\val \rtom \tm\esub\var\val \rtoe \tm\isub\var\val$.

There is no sensible way, instead, to simulate \VSC into Plotkin's calculus. 
Indeed, \VSC is a proper extension of Plotkin's: false normal forms of Plotkin's calculus such as $\Omega_l =(\la{\var}\delta)(\vartwo\vartwo)\delta$ and $\Omega_r\defeq \delta ((\la{\var}\delta) (\vartwo\vartwo))$ are divergent in \VSC:
\begin{center}
$\begin{array}{lllllllllllllll}
\Omega_l  & \tomo & \delta\esub\var{\vartwo\vartwo}\delta & \tomo & (\var\var)\esub\var{\delta}\esub\var{\vartwo\vartwo} 
&\toeo & (\delta\delta)\esub\var{\vartwo\vartwo}  &\tomo & \ldots 
\\[5pt]
\Omega_r & \tomo & \delta \delta\esub\var{\vartwo\vartwo} & \tomo & (\var\var)\esub\var{\delta\esub\var{\vartwo\vartwo}} & \toeo& (\delta\delta)\esub\var{\vartwo\vartwo}&\tomo & \ldots 
\end{array}$
\end{center}
Note that divergence of $\Omega_l$ crucially uses distance on $\tom$ (in the second step), while divergence of $\Omega_r$ crucially uses distance on $\toe$ (in the third step).

\paragraph{VSC and Contextual Equivalence.} Pleasantly, Plotkin's calculus and the VSC induce the same notion of contextual equivalence on $\l$-terms without ES, since contextual equivalence is defined with respect to contexts that \emph{close} terms, see Accattoli and Guerrieri \cite{DBLP:journals/pacmpl/AccattoliG22}. Moreover, the Open VSC is \emph{contextually stable}.

\begin{proposition}[Contextual stability \cite{DBLP:journals/corr/abs-2303-08161}]
	\NoteProof{propappendix:ctx-stability}
The Open VSC is contextually stable, that is, if $\tm\tovsubo^* \tmtwo$ then $\tm \ctxeq\tmtwo$.
\label{prop:ctx-stability}
\end{proposition}

\paragraph*{The Strong VSC.} The \Full \VSC is obtained by allowing rewriting rules  everywhere, including under abstractions, via a closure by general contexts. 
\begin{center}
%
		\begin{tabular}{c\colspace\colspace\colspace cc}
		\begin{tabular}{ccc}
\textsc{\Full rewrite rules:}
			&
			\multirow{2}{*}
			{\begin{prooftree}
    \hypo{\tm \rootRew{a} \tm'}		
    \infer1{\ctxp\tm \Rew{a} \ctxp{\tm'}}
			\end{prooftree}}
		\\
		($a \in \set{\msym,\esym}$)
	\end{tabular}
	&

		$\begin{array}{r l@{\ } l@{\ } lllllll}
  \multicolumn{3}{c}{\textsc{\Full reduction:} }
  \\\tovsub  & \defeq  &\tom \cup \toe
\end{array}$
	\end{tabular}
\end{center}
Unlike the previous cases, $\tovsub$ is not diamond: consider all the $\vsub$-evaluations of $(\var\var) \esub{\var}{\la{\vartwo} \Id \Id}$, with $\Id \defeq \la{\varthree}\varthree$.  

\begin{proposition}[Properties of \full reduction \cite{AccattoliPaolini12,DBLP:journals/pacmpl/AccattoliG22}] 
\label{thm:vsc-confluence}\label{prop:properties-full-reduction} 
\hfill
\begin{enumerate}
	\item \label{p:properties-full-reduction-confluence} The reduction $\tovsub$ is confluent.
	
	\item \label{p:properties-full-reduction-harmony} A term is $\vsub$-normal if and only if it is a \full fireball, where \emph{\full fireballs} (and \emph{\full inert terms}, \emph{\full values}) are:
\end{enumerate}
\begin{center}
	$\begin{array}{rl\colspace rl}
	\textsc{\Full inert terms } 
	&
	\sitm \grameq \var \mid \sitm \sfire \mid \sitm \esub\var{\sitmtwo}
	&
	\textsc{\Full values } 
	&
	\sval  \grameq  \la\var\sfire
	\\
	\textsc{\Full fireballs } 
	&
	\sfire \grameq \sitm \mid \sval \mid \sfire \esub\var\sitm
	\end{array}$
\end{center}
\end{proposition}
The notions of \full inert terms and \full fireballs are a generalization of inert terms and fireballs, respectively, by simply iterating the construction under all abstractions.
%
%
Note that they are similar to normal forms of the (\cbn) $\l$-calculus, but they can have \ESs containing \full inert terms. 

\section{The External Strategy}
\label{sect:external}

In this section, we define Accattoli et al.'s (strong) \emph{external} strategy \cite{DBLP:conf/lics/AccattoliCC21}, that shall be studied via multi types in \refsect{shrinking}. 
Its role is analogous to the leftmost-outermost 
strategy of the $\l$-calculus. A notable difference, however, is that the external strategy is itself non-deterministic, but in a 
harmless way, because it is diamond. The idea is the same used for Plotkin's calculus, that is, allowing one to reduce sub-terms of applications (and ES) in any order. In a strong setting, however, it is a bit trickier to enforce it.

We need a few notions. Firstly, \emph{rigid terms}, 
\ie the variation over inert terms where the arguments of the head variable can be whatever term: 
\begin{center}
	$\textsc{\Pointed terms} \quad \ptm, \ptmtwo  \grameq  \var \mid \ptm\tm \mid \ptm\esub\var \ptmtwo$
\end{center}
Every (\full) inert term is a \pointed term, but the converse does not hold---consider $\vartwo({\delta\Id})$, which is rigid but not inert.

Secondly, we need evaluation contexts for the external strategy $\toess$, which is defined on top of open evaluation. The base case is given by the open rewriting rules (themselves defined via a closure by open contexts, see \Cref{sect:vsc}), which are then closed by \emph{external contexts}, 
defined mutually with \emph{rigid contexts}: 
\begin{center}
		$\begin{array}{r@{\quad} r@{\ } llcc}
		\textsc{External contexts} & \extctx & \grameq \ctxhole \mid \la\var \extctx \mid \tm\esub\var \ictx \mid \extctx \esub\var \ptm \mid \ictx 
		\\
		\textsc{Rigid contexts} & \ictx & \grameq \ptm \extctx \mid \ictx \tm \mid \ictx\esub\var \ptm \mid \ptm \esub\var\ictx
		\end{array}$
\end{center}
\begin{center}
		\begin{tabular}{c\colspace\colspace cc}
		\begin{tabular}{ccc}
\textsc{External rewrite rules:}
			&
			\multirow{2}{*}
			{\begin{prooftree}
					\hypo{\tm \Rew{\wsym a} \tm'}		
					\infer1{\extctxp\tm \Rew{\esssym a} \extctxp{\tm'}}
			\end{prooftree}}
		\\
		($a \in \set{\msym,\esym}$)
	\end{tabular}
&
		$\begin{array}{cccccc}
\multicolumn{3}{c}{\textsc{External reduction:}}
 \\
		\tovsubs \, \defeq \, \toms \cup \toes
	\end{array}$
			\end{tabular}
\end{center}
\giulio{Clearly, $\tovsubs \,\subsetneq\, \tovsub$.}
The strategy 
diverges on $\vartwo (\la\varthree\Omega)$ (\cgiulio{because}{as} $\vartwo\la\varthree\ctxhole$ is a rigid---thus external---context) and normalizes the potentially diverging term $(\la\var\vartwo) (\la\varthree\Omega) \allowbreak\tovsubs^* \vartwo$, because values can be erased even if they diverge under abstraction. 

\emph{Key example}: the external strategy diverges on the term $\tm = (\la\var\Id) (\vartwo (\la\varthree\Omega))$ of \refsect{cbv} on which the strong fireball calculus terminates, showing that \emph{the two strong settings have different notions of termination}. Indeed, $\tm \toms \Id\esub\var{\vartwo (\la\varthree\Omega)}$ and then the external strategy diverges because $\vartwo (\la\varthree\Omega)$ cannot be erased and $\Omega$ occurs in the external evaluation context $\Id\esub\var{\vartwo (\la\varthree\ctxhole)}$. 

The grammars of external and rigid contexts allow evaluation
to enter only inside non-applied abstractions, \eg $(\la\var(\Id\Id)) \val \not\tovsubs 
(\la\var(\vartwo\esub\vartwo\Id)) \val$. 
This is a sort of outside-in order
 which 
is neither left-to-right nor right-to-left---we have both $(\Id\Id) (\Id\Id)\toms (\vartwo\esub\vartwo\Id) (\Id\Id)$ 
and $(\Id\Id) (\Id\Id)\toms (\Id\Id)(\vartwo\esub\vartwo\Id)$---since open contexts do not impose an order on applications. As just showed, the strategy is non-deterministic: another example is given by $\tm = \var 
(\la\vartwo(\Id \Id)) \esub\var{\var 
(\Id\Id)} \toms \var (\la\vartwo \varthree\esub\varthree \Id) \esub\var{\var (\Id\Id)}$, and $\tm \toms \var 
(\la\vartwo(\Id \Id)) \esub\var{\var(\varthree\esub\varthree{\Id})}$. 
Such a behavior however is only a relaxed form of determinism, as it satisfies the \emph{diamond property}.

\begin{proposition}[Properties of external reduction $\tovsubs$ \cite{DBLP:conf/lics/AccattoliCC21}]
	\label{prop:external-properties}
	\label{prop:vsc-diamond}
	\label{l:fullness}
	\begin{enumerate}
		\item\label{p:external-properties-diamond} 
		External reduction $\tovsubs$ is diamond.
		\item\label{p:external-properties-fullness} \emph{Fullness}: 
		Let $\tm$ be a VSC term: $\tm$ is $\esssym$-normal if and only if $\tm$ is $\vsub$-normal.
\end{enumerate}
\end{proposition}

%
%

\section{Multi Types by Value}
\label{sect:types}
We present here the system of \cbv \emph{multi types} that we shall use to characterize the termination of the external strategy in \refsect{shrinking}. The system was introduced by Ehrhard \cite{DBLP:conf/csl/Ehrhard12} for Plotkin's \cbv $\l$-calculus, as the \cbv 
version of de Carvalho's multi types system for \cbn \cite{Carvalho07,deCarvalho18}. Both systems can be seen as presentations of the relational semantics of linear logic restricted to the \cbv/\cbn interpretation of the $\l$-calculus. 
The \cbv~multi type system is also used in \cite{DBLP:conf/lfcs/Diaz-CaroMP13,DBLP:conf/fossacs/CarraroG14,Guerrieri18,DBLP:conf/aplas/AccattoliG18,DBLP:journals/pacmpl/AccattoliG22,DBLP:journals/corr/abs-2303-08161}.

\paragraph{Multi Types.} There are two layers of types, \emph{linear} and \emph{multi types}: 
\begin{center}
$\begin{array}{r\colspace r@{\ } ll}
\textsc{Linear types} & \ltype, \ltypetwo &\grameq \ground \mid \larrow{\mtype}{\mtypetwo}
\\
\textsc{Multi types} & \mtype, \mtypetwo &\grameq \mset{\ltype_1, \dots, \ltype_n} \quad n \in \nat
\end{array}$
\end{center}
where $\ground$ is an unspecified ground type and $\mset{\ltype_1, \dots, \ltype_n}$ is our notation for finite 
multisets. 
The \emph{empty} multi type $\mset{\,}$ obtained by taking $n = 0$ is also denoted by $\emptymset$. 
When \cbv multi types are used to study weak evaluation, they are usually presented without the ground type $\ground$, as one can use $\emptymset$ as base case for types. For studying strong evaluation, however, $\ground$ is mandatory, as we shall see.

%
A multi type $\mset{\ltype_1, \dots, \ltype_n}$ has to be intended as a conjunction $\ltype_1 \land \dots \land 
\ltype_n$ of linear types $\ltype_1, \dots, \ltype_n$, for a commutative, associative, non-idempotent conjunction 
$\land$ (morally a tensor $\otimes$), of neutral element $\emptymset$.
\correction{Note however~that~$\mset{\ltype} \neq \ltype$.}

The intuition is that a linear type corresponds to a single use of a term $\tm$, and that $\tm$ is typed with a 
multiset 
$\mtype$ of $n$ linear types if it is going to be used (at most) $n$ times. The  meaning of \emph{using a term} (once) is not 
easy to define precisely. Roughly, it means that if $\tm$ is part of a larger term $\tmtwo$, then a copy 
of $\tm$ shall end up in evaluation position during the evaluation of $\tmtwo$. More precisely, the copy shall 
end 
up in evaluation position where \mbox{it is applied to some terms.}

\begin{figure*}[t!]
\centering
\begin{tabular}{c\colspace\colspace c\colspace\colspace cccc}
	$\begin{prooftree}[label separation = .2em]
					\infer0
					[\scriptsize$\ruleAx$]
					{\var \hastype [\ltype] \vdash \var \hastype \ltype}
			\end{prooftree}$
			&
			$\begin{prooftree}[separation=1em, label separation = .2em]
					\hypo{\typctx \vdash \tm \hastype \mset{ \larrow{\mtype\!}{\!\mtypetwo} }}
					\hypo{\typctxtwo \vdash \tmtwo \hastype \mtype}
					\infer2[\scriptsize$\ruleAp$]
					{\typctx \uplus \typctxtwo \vdash \tm\tmtwo \hastype \mtypetwo}
			\end{prooftree}$
			&
			$\begin{prooftree}[label separation = .2em]
					\hypo{\tyjp{}{\tm}{\typctx, \var \hastype \mtype}{\mtypetwo}}
					\infer1[\scriptsize$\ruleFun$]
					{\tyjp{}{\la{\var}{\tm}}{\typctx}{\ty{\mtype}{\mtypetwo}}}
			\end{prooftree}$
\\\\
			
				&
							$\begin{prooftree}[separation=1em, label separation = .2em]
					\hypo{\typctx, \var \hastype \mtype \vdash \tm \hastype \mtypetwo}
					\hypo{\typctxtwo \vdash \tmtwo \hastype \mtype}
					\infer2
					[\scriptsize$\ruleES$]
					{\typctx \uplus \typctxtwo \vdash \tm \esub\var\tmtwo \hastype \mtypetwo}
			\end{prooftree}$
&
$\begin{prooftree}[separation=1em, label separation = .2em]
				\hypo{\left[\tyjp{}{\tval}{\typctx_{i}}{\ltype_{i}}\right]_{\iI}}
				\infer1
				[\scriptsize$\ruleMany$]
				{\tyjp{}{\tval}{\biguplus_{\iI}\typctx_{i} }{\mult{\ltype_{i}}_{\iI}}}
				\end{prooftree}$
		\end{tabular}
	\caption{Call-by-value multi type system. In rule $\ruleMany$, $\tval$ is a \emph{theoretical value}, \ie a variable or an abstraction.}
	\label{fig:cbvtypes}
\end{figure*}		
The derivation rules for the multi types system are in \Cref{fig:cbvtypes} (explanation follows).  The rules are the same as in Ehrhard \cite{DBLP:conf/csl/Ehrhard12}, up to their extension to \ESs. 

\emph{Judgments} have shape $\typctx \vdash \tm \hastype \mtype$ or $\typctx \vdash \tm \hastype \ltype$ where $\tm$ is a term, $\mtype$ is a multi type, $\ltype$ is a 
linear type, and $\typctx$ is a \emph{type context}, that is, a total function from variables to multi types 
such that  $\domain{\typctx} \defeq \{\var \mid \typctx(\var) \neq \emptymset\}$ is finite.

\paragraph{Technicalities about Types.} The type context $\typctx$ is \emph{empty} if $\dom{\typctx} = \emptyset$.  
\emph{Multi-set sum} $\mplus$ is extended to type contexts point-wise,
\ie\  $(\typctx \mplus \typctxtwo)(\var) \defeq \typctx(\var) \mplus \typctxtwo(\var)$ for each variable $\var$.
This notion is extended to a finite family of type contexts as expected, 
in particular $\bigmplus_{i \in J\!} \typctx_i$ is the empty context  when $J = \emptyset$.
A type context $\typctx$ is denoted by $\var_1 \hastype \mtype_1, \dots, \var_n \hastype \mtype_n$ (for some $n \in 
\nat$) if $\dom{\typctx} \subseteq \{\var_1, \dots, \var_n\}$ and $\typctx(\var_i) = \mtype_{i}$ for all $1 \leq i \leq 
n$.
Given two type contexts $\typctx$ and $\typctxtwo$ such that $\dom{\typctx} \cap \dom{\typctxtwo} = \emptyset$, the 
type 
context $\typctx, \typctxtwo$ is defined by $(\typctx, \typctxtwo)(\var) \defeq \typctx(\var)$ if $\var \in 
\dom{\typctx}$, $(\typctx, \typctxtwo)(\var) \defeq \typctxtwo(\var)$ if $\var \in \dom{\typctxtwo}$, and $(\typctx, 
\typctxtwo)(\var) \defeq \emptymset$ otherwise.
Note that $\typctx, \var \hastype \emptymset = \typctx$, where we implicitly assume $\var \notin \dom{\typctx}$. 

We write $\concl{\tderiv}{\typctx}{\tm}{\mtype}$ if $\tderiv$ is a (\emph{type}) \emph{derivation} (\ie a tree constructed using the rules in \Cref{fig:cbvtypes}) with conclusion the multi judgment $\typctx \vdash \tm \hastype \mtype$.
In particular,  we write $\concl{\tderiv}{\,}{\tm}{\mtype}$ when $\typctx$ is empty.
We write $\derive{\tderiv}{\tm}$ if $\concl{\tderiv}{\typctx}{\tm}{\mtype}$ for some type context $\typctx$ and some multi type $\mtype$. 

We need a notion of size of type derivations, which shall be used as the termination measure for typable terms.

\begin{definition}[Derivation size]
	Let $\tderiv$ be a type derivation. 
	The \emph{size} $\size{\tderiv}$ of $\tderiv$ is the number of 
	rule occurrences in $\tderiv$ except for rule $\ruleMany$.
\label{def:two-sizes}
\end{definition}

\paragraph{Multisets and Rule $\ruleMany$.}
Rule $\ruleMany$ plays a crucial role, as it is the only rule introducing multi-sets on the right-hand side of judgments: it takes as premises a multi-set of derivations of linear types for a term $\tval$, and glues them together giving a judgment with the multi-set of linear types to $\tval$. The term $\tval$ is a  \emph{theoretical value} $\tval$, that is, a variable or an abstraction---the terminology is taken from Accattoli and Sacerdoti Coen \cite{DBLP:journals/iandc/AccattoliC17}.
Rule $\ruleMany$ is the multi types analogous of the promotion rule of linear logic,
which, in the \cbv representation of the $\lambda$-calculus, is indeed used for typing
abstractions and variables. Note that in particular all abstractions are typable with $\emptymset$ via a $\ruleMany$ rule with no premises.

\paragraph{Subject Reduction and Expansion.} The first properties of the type system that we show are subject reduction and expansion, which hold for \emph{every} VSC step, not only external ones. They rely on a substitution lemma (and its inverse for subject expansion) in the \shortlong{Appendix of \Cref{?}}{Appendix}.

\begin{proposition}[Qualitative subjects]
	\NoteProof{propappendix:qual-subject}
	Let $\tm \tovsub \tm'$.
	\label{prop:qual-subject} 
	\begin{enumerate}
		\item \label{p:qual-subject-reduction}
		\emph{Reduction}: 	if $\namedtyjp{\tderiv}{}{\tm}{\typctx}{\mtype}$ then there is  $\namedtyjp{\tderiv'}{}{\tm'}{\typctx}{\mtype}$ such that $\size{\tderiv} \geq \size{\tderiv'}$.
		
		\item \label{p:qual-subject-expansion}
		\emph{Expansion}: if $\namedtyjp{\tderiv'}{}{\tm'}{\typctx}{\mtype}$ then there is a derivation $\namedtyjp{\tderiv}{}{\tm}{\typctx}{\mtype}$.
	\end{enumerate}
\end{proposition}

Note that subject reduction (\refpropp{qual-subject}{reduction}) says also that the derivation size \emph{cannot increase} after a reduction step. It does not say that it \emph{decreases} at every step because, for instance, if $\la\var\tm \tovsub \la\var\tm'$ and $\la\var\tm$ is typed using a empty $\ruleMany$ rule (\ie with 0 premises), which is a derivation of size 0, then also $\la\var\tm'$ is typed using a empty $\ruleMany$ rule, of size 0. 
Hence, not all typable terms terminate for strong/external evaluation, as for instance $\la\var\Omega$ is typable (with $\emptymset$).

There are two ways to strengthen subject reduction without changing the type system and recover termination: restricting either the reductions to not take place under abstraction, which is what we shall do in the next section for the open case, or the kind of types taken into account (roughly, as to limit the use of $\emptymset$), which is what shall guarantee termination for the external strategy.

\section{Multi Types for Open \cbv}
\label{sect:open}

Here we recall the qualitative part of the relationship between \cbv multi types and \ocbv\ studied by Accattoli and Guerrieri in 
\cite{DBLP:journals/pacmpl/AccattoliG22}, where they develop also a quantitative study not used here. 
The reason to recall their result is twofold. 
Firstly, the  external case relies on the open one. 
Secondly, the open case provides the blueprint for the strong case, allowing us to stress similarities and differences. 

The result is that the open evaluation $\tovsubo$ of $\tm$ terminates if and only if $\tm$ is typable.
Since $\tovsubo$ does not reduce under abstractions, every abstraction is $\osym$-normal and indeed typable, for instance with $\emptymset$. As an example, note that $\la{\var}\Omega$ is typable with $\emptytype$ (rule $\ruleManyVal$ with $0$ premises), though $\Omega$ is not.

\paragraph*{Correctness.}
Open correctness establishes that all typable terms are $\osym$-normalizing and it is proved by showing that the size of type derivation decreases with every $\osym$-step. Open correctness is proved following a standard scheme, namely proving a quantitative version of  open subject reduction, stating that every $\tovsubo$ step preserves types and decreases the general size of a derivation.

%

\begin{proposition}[Open quantitative subject reduction]
	\NoteProof{propappendix:weak-subject-reduction}
	Let $\namedtyjp{\tderiv}{}{\tm}{\typctx}{\mtype}$ be a derivation. If $\tm \tovsubo \tm'$ then there is 
	$\namedtyjp{\tderiv'}{}{\tm'}{\typctx}{\mtype}$ with $\size{\tderiv} > \size{\tderiv'}$.
\label{prop:weak-subject-reduction}
\end{proposition}

The size of derivations decreases after any $\tovsubo$ step, 
thus proving $\tovsubo$-termination for typable terms. Clearly, the size provides a bound to the number of steps.

\begin{theorem}[Open correctness]
	\NoteProof{thmappendix:open-correctness}
	Let
	$\derive{\tderiv}{\tm}$ be a derivation.
	Then there is a $\osym$-normalizing evaluation $\deriv \colon \tm \tovsubo^* \fire$ with $\size{\deriv} \leq \size{\tderiv}$.
	\label{thm:open-correctness}
\end{theorem}

\paragraph*{Completeness.}
Open completeness establishes that every $\osym$-normalizing term is typable. 
Again, the proof technique is standard: a lemma states that every $\osym$-normal form is typable, and subject expansion (\refpropp{qual-subject}{expansion}) allows us to pull back typability along $\tovsubo$ steps. The lemma about normal forms says that they are all typable with $\emptymset$ and relies on a stronger statement about inert terms: they can be assigned whatever multi type $\mtype$, by tuning the type context $\typctx$ accordingly.

\newcounter{prop:precise-open-typability-nf}
\addtocounter{prop:precise-open-typability-nf}{\value{theorem}}
\begin{lemma}[Typability of open normal forms]
	\label{prop:precise-open-typability-nf} 
	\NoteProof{propappendix:precise-open-typability-nf}
	\begin{enumerate}
		\item \label{p:precise-open-typability-nf-inert}
		\emph{Inert:}
		for every inert term $\itm$ and multi type $\mtype$, there exist a type context $\typctx$ 	and a derivation $\concl{\tderiv}{\typctx}{\itm}{\mtype}$.

		\item\label{p:precise-open-typability-nf-fireball}
		\emph{Fireball:}
		for every fireball $\fire$ there exists a type context $\typctx$ 	and a derivation $\concl{\tderiv}{\typctx}{\fire}{\zero}$.
	\end{enumerate}
\end{lemma}

%

\begin{theorem}[Open completeness]
	\NoteProof{thmappendix:open-completeness}
	Let $\deriv \colon \tm \tovsubo^* \fire$ be an $\osym$-normalizing evaluation.
	Then there is a derivation $\concl{\tderiv}{\typctx}{\tm}{\emptytype}$.
	\label{thm:open-completeness}
\end{theorem}
\section{Shrinking Multi Types for the External Strategy}
\label{sect:shrinking}

In this section, we restrict the set of judgments as to characterize the typable terms that terminate with respect to the external strategy. The restriction is obtained by adapting to \cbv the \emph{shrinking technique} for \cbn multi types in 
Accattoli et al. \cite{DBLP:journals/pacmpl/AccattoliGK18}. At the end of the section, we also obtain the untyped normalization theorem for the external strategy.

The definition of shrinking judgments is standard and not due to \cite{DBLP:journals/pacmpl/AccattoliGK18}, see for instance Krivine's book 
\cite{DBLP:books/daglib/0071545}, but the proof technique that we shall use is due to \cite{DBLP:journals/pacmpl/AccattoliGK18} and it is different from others in the literature \cite{DBLP:books/daglib/0071545,deCarvalho18,DBLP:conf/ifipTCS/KesnerV14,DBLP:journals/igpl/BucciarelliKV17}. Its key ingredient is the isolation of a key property of rigid terms (\reflemma{spread-shrinking} below). 

\paragraph{The Need for Shrinking.} As already pointed out at the end of \refsect{types}, some terms that diverge with strong evaluation are typable. We have that $\Omega$ itself is not typable, but $\la\vartwo\Omega$ is typable with $\emptymset$ (via a $\ruleManyVal$ rule with 0 premises). It might seem that the problem is being typable with $\emptymset$, but also $\var(\la\vartwo\Omega)$ is externally divergent and can be typed by assigning $\mset{\larrow\emptymset\mtype}$ to $\var$ (any $\mtype$ works). 
In this case the problem is that, since $\emptymset$ is on the left of $\multimap$, the argument is meant to be erased, but $\var$ cannot actually 
erase it. This is a problem typical of strong settings, as it occurs also in Strong \cbn. The solution is to restrict 
to type derivations satisfying a predicate that forbids types where $\emptymset$ plays these dangerous tricks; in particular a ground type $\ground\neq\emptymset$ is needed. This is 
unavoidable and standard, see 
\cite{DBLP:books/daglib/0071545,DBLP:journals/pacmpl/AccattoliGK18}. Following \cite{DBLP:journals/pacmpl/AccattoliGK18}, the predicate is here called \emph{shrinkingness} 
because it ensures that the size of type derivations shrinks at each $\tovsubs$ step (see \Cref{prop:shrinking-subject-reduction} below).

\begin{figure*}[t!]
\centering
		$\begin{array}{rrcl \colspace rrcl}
		\textsc{Right  multi type } & \rmtype &\grameq &	\mset{\rltype_1, \dots, \rltype_n} \ \ n \geq 1
		&
		\textsc{Right  linear type } & \rltype &\grameq & \ground \mid  \larrow{\lmtype}{\rmtype}
		\\
		\textsc{Left  multi type } & \lmtype &\grameq &	\mset{\lltype_1, \dots, \lltype_n} \ \ n \geq 0
		&
		\textsc{Left  linear type } & \lltype &\grameq & \ground \mid  \larrow{\rmtype}{\lmtype}
		\end{array}$
	\caption{Right and left (shrinking) types.}
	\label{fig:shrinking}
\end{figure*}
\paragraph{Defining Shrinking.} The definition of shrinking forbids the empty multi-set $\emptymset$ on the left of some type arrows $\multimap$. We actually need two notions of shrinking types, \emph{left} and \emph{right}. Intuitively, it is because the typing rule $\ruleFun$ shifts a type from the left-hand side of a judgment to the left of $\multimap$ on the right-hand side of a judgment. 
Formally, \emph{left} and \emph{right shrinking} (multi or linear) \emph{types} are defined in \Cref{fig:shrinking} (we omit \emph{shrinking} when referring to left or right types, for brevity), by mutual induction. The key point is that right multi types \emph{cannot be empty} (note $n \geq 1$), thus $\emptymset$ is forbidden on the left of top $\multimap$ for left linear types.

The notions extend to type contexts and to derivations as follows:
	\begin{itemize}

	\item A type context $\var_1 \hastype \mtype_1, \dots, \var_n \hastype \mtype_n$ is \emph{\leftsh}  if each $\mtype_i$ is \leftsh;
	\item A derivation $\concl{\tderiv}{\typctx}{\tm}{\mtype}$ is \emph{shrinking} if $\typctx$ is \leftsh and $\mtype$ is \rightsh.
\end{itemize}
Examples: $\mset{\ground}$ is both \leftsh and \rightsh (this fact shall play a role below), while $\zero$ is \leftsh but not \rightsh, and $\mset{\larrow{\emptymset}{\mset{\ground}}}$ is \rightsh but not \leftsh.

\paragraph{Key Property of Left Shrinking.} Shrinkingness is a predicate of derivations 
depending only on their \emph{final judgment}. 
For proving  
properties of shrinking derivations, we have to analyze how shrinking propagates to sub-derivations, to apply the \ih in proofs. 
The following lemma is specific to \emph{\leftsh} shrinkingness, on which the propagation of shrinkingness 
then builds. 
%
It says that for specific terms---typable \emph{rigid} terms---\leftsh shrinkingess spreads from the type context to the right-hand multi type in~a~judgment. 
It is the key property of the proof technique.

\begin{lemma}
	[Spreading of \leftsh shrinkingness on judgments]
	\NoteProof{lappendix:spread-shrinking}
	Let $\concl{\tderiv}{\typctx}{\ptm}{\mtype}$ be a derivation and $\ptm$ a \pointed term. 
	If $\typctx$ is \leftsh then $\mtype$ is \leftsh.
	\label{l:spread-shrinking}
\end{lemma}

\paragraph{Correctness.}
Shrinking correctness establishes that all typable terms with a shrinking derivation are externally normalizing.  We follow the same pattern as for the open case, but the proof of subject reduction is trickier---this is the delicate  point of the proof technique by Accattoli et al. \cite{DBLP:journals/pacmpl/AccattoliGK18}. It crucially uses the key property of left shrinking for rigid terms above (\reflemma{spread-shrinking}), and it also requires an auxiliary statement with a weaker hypothesis for the induction to go through.

\newcounter{prop:shrinking-subject-reduction}
\addtocounter{prop:shrinking-subject-reduction}{\value{theorem}}
\begin{proposition}[Shrinking quantitative subject reduction for $\toess$]
	\label{prop:shrinking-subject-reduction}
	\NoteProof{propappendix:shrinking-subject-reduction}
	\begin{enumerate}
	\item \emph{Auxiliary statement}: 	Let $\typctx$ be a \leftsh   context.
	Suppose that $\concl{\tderiv}{\typctx}{\tm}{\mtype}$ and that if $\tm$ is a \valES then $\mtype$ is \rightsh. If $\tm \toess \tm'$ then there is a derivation $\namedtyjp{\tderiv'}{}{\tm'}{\typctx}{\mtype}$ with $\size{\tderiv} > \size{\tderiv'}$.
\item \emph{Actual statement}:
Let $\namedtyjp{\tderiv}{}{\tm}{\typctx}{\mtype}$ be a shrinking derivation. If $\tm \toess \tm'$ then there is a derivation 
$\namedtyjp{\tderiv'}{}{\tm'}{\typctx}{\mtype}$ with $\size{\tderiv} > \size{\tderiv'}$.
\end{enumerate}
\end{proposition}

\begin{proof}
Note that the auxiliary statement is stronger, because every shrinking derivation (defined as having $\typctx$ left and $\mtype$ right) satisfies it: if $\tm$ is not an answer then $\mtype$ can be whatever, and thus in particular it can be right. For the auxiliary statement, we give two cases, the one motivating the use of the auxiliary statement and one showing the use of the key property for rigid terms. The other cases are in the Appendix. The proof is by induction on the evaluation strong context $\strongctx$ such that $\tm = \strongctxp{\tmtwo} \toess 
	\strongctxp{\tmtwo'} = \tm'$ with $\tmtwo \tovsubo \tmtwo'$. The two cases:
	\begin{itemize}
		\item \emph{Rigid context applied to term}, \ie $\strongctx = \ictx\tmthree$. 
		Then, $\tm = \strongctxp{\tmtwo} = \ictxp{\tmtwo} \tmthree \toess \ictxp{\tmtwo'} \tmthree = \strongctxp{\tmtwo'} = \tm'$ with $\tmtwo \tovsubo \tmtwo'$.
		The derivation $\tderiv$ has the following shape:
		\begin{center}$
		\tderiv = 
		\begin{prooftree}
		\hypo{\concl\tderivtwo\typctxtwo{\ictxp\tmtwo}{ \mset{\larrow\mtypetwo\mtype} }}
		\hypo{\concl\tderivthree\typctxthree \tmthree \mtypetwo}
		\infer2[\footnotesize$\ruleAp$]{\typctxtwo \uplus \typctxthree \vdash \ictxp{\tmtwo} \tmthree \hastype \mtype}
		\end{prooftree}$
		\end{center}
		where $\typctx = \typctxtwo \mplus \typctxthree$ is   \leftsh by hypothesis, and then so is $\typctxtwo$.
		By \ih (as $\ictxp{\tmtwo}$ is not an \valES), there is a derivation 
		$\concl{\tderivtwo'}{\typctxtwo}{\ictxp{\tmtwo'}}{\mset{\larrow{\mtypetwo}{\mtype}}}$ with $\size{\tderivtwo'} < \size{\tderivtwo}$.
		We can then build the following derivation:
		\begin{center}$
		\tderiv' = 
		\begin{prooftree}
		\hypo{ \concl{\tderivtwo'}\typctxtwo {\ictxp{\tmtwo'}}{ \mset{\larrow\mtypetwo\mtype} } }
		\hypo{ \concl\tderivthree\typctxthree \tmthree \mtypetwo}
		\infer2[\footnotesize$\ruleAp$]{\typctxtwo \uplus \typctxthree \vdash \ictxp{\tmtwo'} \tmthree \hastype \mtype}
		\end{prooftree}$
		\end{center}
		where $\typctx = \typctxtwo \uplus \typctxthree$ and $\size{\tderiv'} = \size{\tderivtwo'} + \size{\tderivthree} +1 < 
			\size{\tderivtwo} + \size{\tderivthree} +1 = \size{\tderiv}$. Note that in this case we have no hypothesis on $\mtypetwo$, thus on $\mset{\larrow{\mtypetwo}{\mtype}}$, which is why we need a weaker statement in order to use the \ih

		\item \emph{Rigid term applied to strong context}, \ie $\strongctx = \ptm \strongctxtwo$. 
		Then, $\tm = \strongctxp{\tmtwo} = \ptm \strongctxtwop{\tmtwo} \toess \ptm \strongctxtwop{\tmtwo'} = 
		\strongctxp{\tmtwo'} = \tm'$ with $\tmtwo \tovsubo \tmtwo'$.
		The derivation $\tderiv$ is:
		\begin{center}$
		\tderiv = 
		\begin{prooftree}
		\hypo{ \concl\tderivtwo\typctxtwo \ptm { \mset{\larrow\mtypetwo\mtype} } }
		\hypo{ \concl\tderivthree\typctxthree {\strongctxtwop\tmtwo} \mtypetwo }
		\infer2[\footnotesize$\ruleAp$]{\typctxtwo \uplus \typctxthree \vdash \ptm \strongctxtwop{\tmtwo} \hastype \mtype}
		\end{prooftree}$
		\end{center}
		where $\typctx = \typctxtwo \mplus \typctxthree$ is   \leftsh by hypothesis, and then so are $\typctxtwo$ and $\typctxthree$.
		According to spreading of \leftsh shrinkingness applied to $\tderivtwo$ (\Cref{l:spread-shrinking}, which can be applied because $\ptm$ is a rigid term), $\mset{\larrow{\mtypetwo}{\mtype}}$ is a   \leftsh multi type and hence $\mtypetwo$ is a   \rightsh multi type.
		Thus, the \ih applied to $\tderivthree$ gives a derivation 
		$\concl{\tderivthree'}{\typctxthree}{\strongctxtwop{\tmtwo'}}{\mtypetwo}$ with $\size{\tderivthree'} < \size{\tderivthree}$.
		We then build the following derivation: 
		\begin{center}$
		\tderiv' = 
		\begin{prooftree}
		\hypo{ \concl\tderivtwo\typctxtwo \ptm  { \mset{\larrow\mtypetwo\mtype} } }
		\hypo{ \concl{\tderivthree'}\typctxthree {\strongctxtwop{\tmtwo'}} \mtypetwo }
		\infer2[\footnotesize$\ruleAp$]{\typctxtwo \uplus \typctxthree \vdash \ptm \strongctxtwop{\tmtwo'}  \hastype \mtype}
		\end{prooftree}$
		\end{center}
		where $\typctx = \typctxtwo \uplus \typctxthree$ and $\size{\tderiv'} = \size{\tderivtwo} + \size{\tderivthree'} +1 < 
			\size{\tderivtwo} + \size{\tderivthree} +1 = \size{\tderiv}$.			\qedhere

			\end{itemize}

\end{proof}

\begin{theorem}[Shrinking correctness for $\toess$]
	\NoteProof{thmappendix:correctness}
	Let
	$\derive{\tderiv}{\tm}$ be a shrinking derivation.
	Then there is a $\esssym$-normalizing evaluation $\deriv \colon \tm \toess^* \sfire$ with $\size{\deriv} \leq \size{\tderiv}$.
		\label{thm:correctness}
\end{theorem}
Shrinking correctness for $\toess$ shows that the term $\tm = (\la\var\Id) (\vartwo (\la\varthree\Omega))$ diverging for $\toess$ but normalizing for the strong fireball calculus is not typable, otherwise it would $\toess$-terminate. That is, it shows that $\tm$ is semantically diverging.

\paragraph{Completeness.} 
Shrinking completeness is proven as in the open case, using a lemma about the shrinking typability of strong fireballs. Note that the lemma now has an existential quantification on the type $\mtype$ of strong fireballs, while in the open case the type was simply $\emptymset$; here $\emptymset$ would not work, because it is not right. Note also the part about inert terms, stating that $\mtype$ is \emph{left}: it is not a mistake, it can be seen as an instance of the key properties of rigid terms (\reflemma{spread-shrinking}, inert terms are rigid), and it gives shrinking derivations when $\mtype$ is instantiated with, say, $\mset\ground$, which is both left an right.

\begin{lemma}[Shrinking typability of normal forms]
	\label{prop:typability-normal}
	\NoteProof{propappendix:typability-normal}
	\begin{enumerate}
		\item\label{p:typability-normal-inert}
		\emph{Inert:}
		for every \full inert term $\sitm$ and \leftsh multi type $\mtype$,	there exist a \leftsh type context $\typctx$ 	and a derivation $\concl{\tderiv}{\typctx}{\sitm}{\mtype}$.

		\item\label{p:typability-normal-fireball}
		\emph{Fireball:}
		for every \full fireball $\sfire$ 	there is a shrinking derivation $\derive{\tderiv}{\sfire}$.
	\end{enumerate}
\end{lemma}


\begin{theorem}[Shrinking completeness for $\toess$]
	\NoteProof{thmappendix:completeness}
	Let $\deriv \colon \tm \toess^* \sfire$ be a $\esssym$-normalizing evaluation.
	Then there is a shrinking derivation $\derive{\tderiv}{\tm}$.
		\label{thm:completeness}
\end{theorem}

\paragraph{Untyped Normalization Theorem.} By exploiting an elegant proof technique used by de Carvalho et al. \cite{DBLP:journals/tcs/CarvalhoPF11} and Mazza et al. \cite{DBLP:journals/pacmpl/MazzaPV18}, we obtain an \emph{untyped normalization} theorem for the external strategy of the \VSC, as a corollary of our study of multi types. The key points are subject expansion  (\refpropp{qual-subject}{expansion}) for the \emph{whole} reduction $\tovsub$, instead that just for the external strategy, \correction{and the fact that shrinkingness is a predicate of derivations depending only on their conclusion}.

\begin{theorem}[Untyped normalization for $\toess$]
	Let $\tm$ be a term such that there is a strong fireball $\sfire$ and a reduction sequence $\tm \tovsub^* \sfire$. Then $\tm \toess^*\sfire$. 	\label{thm:normalization}
\end{theorem}

\begin{proof}
Shrinking typability of normal forms (\Cref{prop:typability-normal}) gives a shrinking derivation $\concl{\tderiv}{\typctx}{\sfire}{\mtype}$.
		Subject expansion (\refpropp{qual-subject}{expansion}) iterated along $\tm \tovsub^* 
		\sfire$ gives a shrinking derivation $\concl{\tderivtwo}{\typctx}{\tm}{\mtype}$. By shrinking correctness (\Cref{thm:correctness}), $\tm \toess^* \sfiretwo$ for a strong fireball $\sfiretwo$. By confluence of the VSC (\refpropp{properties-full-reduction}{confluence}), $\sfire=\sfiretwo$.
		\qedhere
\end{proof}

\paragraph{Relational Semantics.}
\giulio{Multi types induce a denotational model, the \emph{relational semantics}, interpreting a term
as the set of its type judgments \cite{Carvalho07,deCarvalho18,DBLP:conf/csl/Ehrhard12,Guerrieri18,DBLP:conf/aplas/AccattoliG18,DBLP:journals/pacmpl/AccattoliG22}.
Here we focus on the semantics induced by shrinking derivations.
Let $\tm$ be a term and let $\vec{\var} = (\var_1, \dots, \var_n)$ be a list of pairwise distinct variables with $n \geq 0$ and $\fv{\tm} \subseteq \{\var_1, \dots , \var_n\}$:
the  \emph{shrinking semantics} $\sem{\tm}_{\vec{\var}}$ of $\tm$ for $\vec{\var}$ is defined by $\sem{\tm}_{\vec{\var}} \defeq \{((\mtypetwo_1, \dots , \mtypetwo_n ), \mtype) \mid \exists \ \text{shrinking } \concl{\tderiv}{\var_1 \hastype \mtypetwo_1, \dots, \var_n \hastype \mtypetwo_n}{\tm}{\mtype} \}$.

Subject reduction and expansion (\Cref{prop:qual-subject}) guarantee that $\sem{\tm}_{\vec{\var}}$ is \emph{invariant} under $\tovsub$: if $\tm \tovsub \tmtwo$ then $\sem{\tm}_{\vec{\var}} = \sem{\tmtwo}_{\vec{\var}}$.
Shrinking correctness (\Cref{thm:correctness}) and completeness (\Cref{thm:completeness}), along with untyped normalization (\Cref{thm:normalization}), 
guarantee \emph{adequacy} for this semantics, \ie a semantic characterization of normalization in Strong \VSC:
$\tm$ is weakly $\vsub$-normalizing if and only if $\sem{\tm}_{\vec{\var}} \neq \emptyset$.}

\section{Conclusions}
This paper studies call-by-value strong evaluation defined as the external strategy $\toess$ of the value substitution calculus (VSC). Such a strategy is analyzed using the semantical tool of Ehrhard's multi types, declined in their shrinking variant as it is standard for studying strong evaluation. The main contributions are that $\toess$-normalizing terms are exactly those typable with shrinking call-by-value multi types, plus an untyped normalization theorem for $\toess$ in the strong VSC. These results mimic faithfully similar results for strong call-by-name.

These contributions are developed to validate the external strategy as the good notion of termination for strong call-by-value, in contrast to the other \emph{non-equivalent} proposal in the literature given by the strong fireball calculus, which we show to terminate on some terms on which the external strategy diverges.

We conjecture that all $\toess$-normalizing terms are also normalizing in the strong fireball calculus, but a proof is likely to be very technical.

\bibliographystyle{splncs04}
\bibliography{\macrospath/biblio}

\shortlong{}{
\clearpage
\appendix
\section*{Appendix}
The first section of the Appendix gives an example of breaking of contextual stability in the strong fireball calculus.
Then there is a section for every section of the paper with omitted proofs.

\section{Example of Contextual Instability Relative to the Strong Fireball Calculus}
\label{sect:app-cbv-cbfireball}
We show here that the term $\tmtwo \defeq (\la\var\Id) (\vartwo (\la\varthree\Omega))$ that is normalizing in the strong fireball calculus but diverging for the external strategy gives another example of breaking of the contextual stability property. We recall that $\tmtwo \toin \Id$ because $\vartwo (\la\varthree\Omega)$ is an inert term. That is, $\tmtwo$ terminates on a strong normal form in the strong fireball calculus. 

Now, set $\ctx \defeq (\la\vartwo\ctxhole)(\la\varfour\varfour\Id)$ and note that:
\begin{center}
$\begin{array}{lllllllllllllll}
\ctxp\tmtwo &= & (\la\vartwo((\la\var\Id) (\vartwo (\la\varthree\Omega))))(\la\varfour\varfour\Id) & \tobv & (\la\var\Id) ((\la\varfour\varfour\Id)(\la\varthree\Omega))
\\
&\tobv & (\la\var\Id)((\la\varthree\Omega)\Id)  &\tobv & (\la\var\Id)\Omega &\tobv\ldots 
\\[4pt]
\ctxp\Id & = & (\la\vartwo\Id)(\la\varfour\varfour\Id) & \tobv & \Id
\end{array}$
\end{center}
That is, $\tmtwo \toin \Id$ but $\tmtwo\not\ctxeq \Id$ because $\ctxp\tmtwo$ diverges while $\ctxp\Id$ terminates. 

Note that the breaking of contextual stability is of a slightly different nature with respect to the other example given in the paper, as the context here \emph{unlocks} a divergence in $\tmtwo$, instead of \emph{creating} it.

\section{Proofs of \Cref{sect:vsc}}
\begin{proposition}[Contextual stability \cite{DBLP:journals/corr/abs-2303-08161}]
\label{propappendix:ctx-stability}
\NoteState{prop:ctx-stability}
The Open VSC is contextually stable, that is, if $\tm\tovsubo^* \tmtwo$ then $\tm \ctxeq\tmtwo$.
\end{proposition}
\begin{proof}
In the technical report by Accattoli et al. \cite{DBLP:journals/corr/abs-2303-08161} it is shown (Thm 11.8.3 in \cite{DBLP:journals/corr/abs-2303-08161}) the soundness of \emph{net bisimilarity} with respect to contextually equivalence, that is, the fact that if $\tm$ is net bisimilar to $\tmtwo$ then $\tm\ctxeq\tmtwo$. The statement follows from the basic fact that if $\tovsubo$ is included in net bisimilarity.\qedhere
\end{proof}

\section{Proofs of \Cref{sect:types}}

First, we observe the following property: given a derivation for a term $\tm$, all variables associated with a non-empty multi type in the type context are free variables~of~ $\tm$.
\begin{remark}
	\label{rmk:free-variables}
	If $\namedtyjp{\tderiv}{}{\tm}{\typctx}{\mtype}$ then $\dom{\typctx} \subseteq \fv{\tm}$.
\end{remark}

\begin{lemma}[Typing of values: splitting]
	\label{l:typing-value-splitting}
	Let $\namedtyjp{\tderiv}{}{\tval}{\typctx}{\mtype}$ (for $\tval$ theoretical value).
	\begin{enumerate}
		\item \label{p:typing-value-splitting-one} If $\mtype = \emptytype$, then $\dom{\typctx} = \emptyset$ and 
		$\size{\tderiv} = 0$. 
		
		\item \label{p:typing-value-splitting-two} For every splitting $\mtype = \mtype_{1} \mplus \mtype_{2}$, there exist 
		type derivations $\namedtyjp{\tderiv_{1}}{}{\tval}{\typctx_{1}}{\mtype_{1}}$ and 
		$\namedtyjp{\tderiv_{2}}{}{\tval}{\typctx_{2}}{\mtype_{2}}$ such that $\typctx = \typctx_{1} \mplus \typctx_{2}$ and $\size{\tderiv} = \size{\tderiv_{1}} + 
		\size{\tderiv_{2}}$.
		
	\end{enumerate}
\end{lemma}

\begin{proof}\hfill
	\begin{enumerate}
		\item By a simple inspection of the typing rules, $\mtype = \emptytype$ and the fact that $\tval$ is a value imply 
		that 
		$\tderiv$ is of the form
		\begin{prooftree}
			\hypo{}
			\infer1[\footnotesize$\ruleMany$]{\tyjp{}{\tval}{}{\emptytype}}
		\end{prooftree}
		where $\dom{\typctx} = \emptyset$ and $\size{\tderiv} = 0$.
		
		\item Let
		\begin{equation*}
		\tderiv = 
		\begin{prooftree}
		\hypo{}
		\ellipsis{$\tderiv_{i}$}{\tyjp{}{\tval}{\typctx_{i}}{\ltype_{i}}}
		\delims{\left(}{\right)_{\iI}}
		\infer1[\footnotesize$\ruleMany$]{\tyjp{}{\tval}{\bigmplus_{\iI} \typctx_{i}}{\mult{\ltype}_{\iI}}}
		\end{prooftree}
		\end{equation*}
		with $\bigmplus_{\iI} \typctx_{i} = \typctx$ and $\mult{\ltype}_{\iI} = \mtype = \mtype_{1} \mplus \mtype_{2}$. Let 
		$I_{1}$ and $I_{2}$ be sets of indices such that $I = I_{1} \cup I_{2}$, $\mtype_{1} = \mult{\ltype_{i}}_{i \in I_{1}}$ 
		and $\mtype_{2} = \mult{\ltype_{i}}_{i \in I_{2}}$. 
		As $\tval$ is a value, 	we can then derive
		\begin{equation*}
		\tderiv_{1} = 
		\begin{prooftree}
		\hypo{}
		\ellipsis{$\tderiv_{i}$}{\tyjp{}{\tval}{\typctx_{i}}{\ltype_{i}}}
		\delims{\left(}{\right)_{i \in I_1}}
		\infer1[\footnotesize$\ruleMany$]{\tyjp{}{\tval}{\bigmplus_{i \in I_{1}} \typctx_{i}}{\mult{\ltype}_{i \in I_{1}}}}
		\end{prooftree}
		\end{equation*}
		and 
		\begin{equation*}
		\tderiv_{2} = 
		\begin{prooftree}
		\hypo{}
		\ellipsis{$\tderiv_{i}$}{\tyjp{}{\tval}{\typctx_{i}}{\ltype_{i}}}
		\delims{\left(}{\right)_{i \in I_2}}
		\infer1[\footnotesize$\ruleMany$]{\tyjp{}{\tval}{\bigmplus_{i \in I_{2}} \typctx_{i}}{\mult{\ltype}_{i \in I_{2}}}}
		\end{prooftree}
		\end{equation*}
		noting that 
		$$
		\typctx = \bigmplus_{\iI} \typctx_{i} = \left( \bigmplus_{i \in I_{1}} \typctx_{i} \right) \mplus 
		\left(\bigmplus_{i \in I_{2}} \typctx_{i} \right)
		$$
		with 
		$$
		\size{\tderiv} = \sum_{\iI} \size{\tderiv_{i}} = \left( \sum_{i \in I_{1}} \size{\tderiv_{i}} \right) + \left( 
		\sum_{i \in I_{2}} \size{\tderiv_{i}} \right) = \size{\tderiv_{1}} + \size{\tderiv_{2}}
		$$
		\qed
		
	\end{enumerate}
\end{proof}

\begin{lemma}[Substitution]
	\label{l:substitution}	
	Let $\tm$ be a term, $\val$ be a value, and $\namedtyjp{\tderiv}{}{\tm}{\typctx, \var \hastype 
		\mtypetwo}{\mtype}$ and $\namedtyjp{\tderivtwo}{}{\val}{\typctxtwo}{\mtypetwo}$ be derivations.
	Then there is a derivation $\namedtyjp{\tderivthree}{}{\tm \isub{\var}{\val}}{\typctx \mplus \typctxtwo}{\mtype}$ 
	with 
	$\size{\tderivthree} \leq \size{\tderiv} + 
	\size{\tderivtwo}$. 
\end{lemma}

\begin{proof}
	By induction on the term $\tm$.
	Cases:
	\begin{itemize}
		\item \emph{Variable}, then are two sub-cases:
		\begin{enumerate}
			\item $\tm = \var$, then $\tm \isub{\var}{\tval} = \var \isub{\var}{\tval} = \tval$ and 			
			$\sizem{\tderiv} = 0$ and $\size{\tderiv} = 1$.
			
			the derivation $\tderiv$ has necessarily the form (for some $n \in \nat$)
			\begin{equation*}
			\tderiv = 
			\begin{prooftree}
			\infer0[\footnotesize$\Ax$]{\tyjp{}{\var}{\var \hastype \mset{\ltype_1}}{\ltype_1}}
			\hypo{\overset{n \in \nat}{\ldots}}
			\infer0[\footnotesize$\Ax$]{\tyjp{}{\var}{\var \hastype \mset{\ltype_n}}{\ltype_n}}
			\infer3[\footnotesize$\ruleManyVar$]{\tyjp{}{\var}{\var \hastype 
					\mset{\ltype_1,\dots,\ltype_n}}{\mset{\ltype_1,\dots,\ltype_n}}}
			\end{prooftree}
			\end{equation*}
			with $\mtype = \mset{\ltype_1, \dots, \ltype_n} = \mtypetwo$ and $\dom{\typctx} = \emptyset$.
			Thus, $\sizem{\tderiv} = 0$ and $\size{\tderiv} = n$.
			Let $\tderivthree = \tderivtwo$: so, $\namedtyjp{\tderivthree}{}{\tm \isub{\var}{\tval}}{\typctx \mplus 
				\typctxtwo}{\mtype}$ (since $\typctx \mplus \typctxtwo = \typctxtwo$) with $\sizem{\tderivthree} = \sizem{\tderivtwo} = 
			\sizem{\tderivtwo} + \sizem{\tderiv}$ and $\size{\tderivthree} = \size{\tderivtwo} \leq \size{\tderivtwo} + 
			\size{\tderiv}$ (note that $\size{\tderivthree} = \size{\tderiv} + \size{\tderivtwo}$ if and only if $n=0$).
			
			\item $\tm = \varthree \neq \var$, then $\tm \isub{\var}{\tval} = \varthree$ and 
			$\sizem{\tderiv} = 0$, $\size{\tderiv} = 1$, 
			the derivation $\tderiv$ has necessarily the form (for some $n~\in~\nat$)
			\begin{equation*}
			\tderiv = 
			\begin{prooftree}
			\infer0[\footnotesize$\Ax$]{\tyjp{}{\varthree}{\varthree \hastype \mset{\ltype_1}}{\ltype_1}}
			\hypo{\overset{n \in \nat}{\ldots}}
			\infer0[\footnotesize$\Ax$]{\tyjp{}{\varthree}{\varthree \hastype \mset{\ltype_n}}{\ltype_n}}
			\infer3[\footnotesize$\ruleManyVar$]{\tyjp{}{\varthree}{\varthree \hastype 
					\mset{\ltype_1,\dots,\ltype_n}}{\mset{\ltype_1,\dots,\ltype_n}}}
			\end{prooftree}
			\end{equation*}
			where $\mtype = \mset{\ltype_1, \dots, \ltype_n}$ and  $\mtypetwo = \emptymset$ and $\typctx = \varthree \hastype 
			\mtype$ (while $\typctx(\var) = \emptymset$).
			Thus, $\sizem{\tderiv} = 0$ and $\size{\tderiv} = n$.
			By \reflemmap{typing-value-splitting}{one}, from $\namedtyjp{\tderivtwo}{}{\tval}{\typctxtwo}{\emptytype}$ it 
			follows that $\sizem{\tderivtwo} = 0 = \size{\tderivtwo}$ and $\dom{\typctxtwo} = \emptyset$.  
			Therefore, $\typctx \mplus \typctxtwo = \typctx$.
			Let $\tderivthree = \tderiv$: so, $\namedtyjp{\tderivthree}{}{\tm \isub{\var}{\tval}}{\typctx \mplus 
				\typctxtwo}{\mtype}$  with $\sizem{\tderivthree} = \sizem{\tderiv} = \sizem{\tderiv} + \sizem{\tderivtwo}$ and 
			$\size{\tderivthree} = \size{\tderiv} = \size{\tderiv} + \size{\tderivtwo}$.
		\end{enumerate}
		
		\item \emph{Application}, \ie $\tm = \tmtwo\tmthree$. 
		Then $\tm \isub{\var}{\tval} = \tmtwo \isub{\var}{\tval} \tmthree \isub{\var}{\tval}$ and necessarily
		\begin{equation*}
		\tderiv = 
		\begin{prooftree}
		\hypo{}
		\ellipsis{$\tderiv_{1}$}{\typctx_1, \var \hastype \mtypetwo_1 \vdash \tmtwo \hastype 
			\mset{\larrow{\mtypethree}{\mtype}}}
		\hypo{}
		\ellipsis{$\tderiv_{2}$}{\typctx_2, \var \hastype \mtypetwo_2 \vdash \tmthree \hastype \mtypethree}
		\infer2[\footnotesize$\ruleAp$]{\typctx, \var \hastype \mtypetwo \vdash \tmtwo \tmthree \hastype \mtype}
		\end{prooftree}
		\end{equation*}
		with $\sizem{\tderiv} = \sizem{\tderiv_{1}} + \sizem{\tderiv_{2}} + 1$, $\size{\tderiv} = \size{\tderiv_{1}} + 
		\size{\tderiv_{2}} + 1$, $\typctx = \typctx_1 \uplus \typctx_2$ and $\mtypetwo = \mtypetwo_2 \uplus \mtypetwo_2$. 
		According to \reflemmap{typing-value-splitting}{two} applied to $\tderivtwo$ and to the decomposition $\mtypetwo = 
		\mtypetwo_1 \uplus \mtypetwo_2$, there are contexts $\typctxtwo_1, \typctxtwo_2$ and derivations 
		$\namedtyjp{\tderivtwo_{1}}{}{\tval}{\typctxtwo_1}{\mtypetwo_1}$ and 
		$\namedtyjp{\tderivtwo_{2}}{}{\tval}{\typctxtwo_2}{\mtypetwo_2}$ such that $\typctxtwo = \typctxtwo_{1} \mplus 
		\typctxtwo_2$, $\sizem{\tderivtwo} = \sizem{\tderivtwo_1} + \sizem{\tderivtwo_2}$ and $\size{\tderivtwo} = 
		\size{\tderivtwo_{1}} + \size{\tderivtwo_{2}}$.
		
		By \ih, there are derivations $\namedtyjp{\tderivthree_1}{}{\tmtwo \isub{\var}{\tval}}{\typctx_1 \uplus 
			\typctxtwo_1}{\mult{\larrow{\mtypethree}{\mtype}}}$ and $\namedtyjp{\tderivthree_2}{}{\tmthree 
			\isub{\var}{\tval}}{\typctx_2 \uplus \typctxtwo_2}{\mtypethree}$ such that $\sizem{\tderivthree_{i}} = 
		\sizem{\tderiv_{i}} + \sizem{\tderivtwo_{i}}$ and $\size{\tderivthree_{i}} \leq \size{\tderiv_{i}} + 
		\size{\tderivtwo_{i}}$ for all $i \in \{1,2\}$.
		Since $\typctx \uplus \typctxtwo = \typctx_1 \uplus \typctxtwo_1 \uplus \typctx_2 \uplus \typctxtwo_2$, we can 
		build the derivation
		\begin{equation*}
		\tderivthree = 
		\begin{prooftree}
		\hypo{}
		\ellipsis{$\tderivthree_1$}{\typctx_1 \uplus \typctxtwo_1 \vdash \tmtwo\isub{\var}{\tval} \hastype 
			\mset{\larrow{\mtypethree}{\mtype}}}
		\hypo{}
		\ellipsis{$\tderivthree_2$}{\typctx_2 \uplus \typctxtwo_2 \vdash \tmthree\isub{\var}{\tval} \hastype \mtypethree}
		\infer2[\footnotesize$\ruleAp$]{\typctx \uplus \typctxtwo \vdash \tmtwo \isub{\var}{\tval} \tmthree 
			\isub{\var}{\tval} \hastype \mtype}
		\end{prooftree}
		\end{equation*}
		where $\sizem{\tderivthree} = \sizem{\tderivthree_{1}} + \size{\tderivthree_{2}} + 1 = \sizem{\tderiv_{1}} + 
		\sizem{\tderivtwo_{1}} + \sizem{\tderiv_{2}} + \sizem{\tderivtwo_{2}} + 1 = \sizem{\tderiv} + \sizem{\tderivtwo}$ and 
		$\size{\tderivthree} = \size{\tderivthree_{1}} + \size{\tderivthree_{2}} + 1 \leq \sizem{\tderiv_{1}} + 
		\sizem{\tderivtwo_{1}} + \sizem{\tderiv_{2}} + \size{\tderivtwo_{2}} + 1 = \size{\tderiv} + \size{\tderivtwo}$.
		
		\item \emph{Abstraction}, \ie $\tm = \la{\vartwo}{\tmtwo}$.
		We can suppose without loss of generality that $\vartwo \notin \fv{\tval} \cup \{\var \}$, hence $\tm 
		\isub{\var}{\tval} = \la{\vartwo}{\tmtwo\isub{\var}\tval}$ and  $\tderiv$ is necessarily of the form (for some $n \in 
		\nat$) 
		\begin{equation*}
		\begin{prooftree}[separation=1em]
		\hypo{}
		\ellipsis{$\tderiv_{i}$}{\typctx_{i}, \vartwo \hastype \mtypethree_{i}, \var \hastype \mtypetwo_{i} \vdash \tmtwo 
			\hastype \mtype_{i}}
		\infer1[\footnotesize$\ruleFun$]{\tyjp{}{\la{\vartwo}{\tmtwo}}{\typctx_{i}, \var \hastype 
				\mtypetwo_{i}}{\ty{\mtypethree_{i}\!}{\!\mtype_{i}}}}
		\delims{\left(}{\right)_{1\leq i \leq n}}
		\infer1[\footnotesize$\ruleManyVal$]{\tyjp{}{\la{\vartwo}{\tmtwo}}{\bigmplus_{i=1}^{n} \typctx_{i} ; \var \hastype 
				\mplus_{i=1}^{n} \mtypetwo_i}{\mplus_{i=1}^{n} \mset{\larrow{\mtypethree_i}{\mtype_i}}}}
		\end{prooftree}
		\end{equation*}
		with $\sizem{\tderiv} = \sum_{i=1}^{n} (\sizem{\tderiv_{i}} + 1)$ and $\size{\tderiv} = \sum_{i=1}^n 
		(\size{\tderiv_i} + 1)$.
		Since $\vartwo \notin \fv{\tval}$, then $\vartwo \notin \domain{\typctxtwo}$ (\refrmk{free-variables}), and so 
		$\namedtyjp{\tderivtwo}{}{\tval}{\typctxtwo, \vartwo \hastype \emptymset}{\mtypetwo}$. 
		Now, there are two subcases:
		\begin{itemize}
			\item \emph{Empty multi type}: If $n = 0$,  then $\mtypetwo = \emptymset = \mtype$ and $\dom{\typctx} = 
			\emptyset$, with $\sizem{\tderiv} = 0 = \size{\tderiv}$. 
			According to \reflemmap{typing-value-splitting}{one} applied to $\tderivtwo$, $\dom{\typctxtwo} = \emptyset$ with 
			$\sizem{\tderivtwo} = 0 = \size{\tderivtwo}$.
			We can then build the derivation 
			\begin{equation*}
			\tderivthree = 
			\begin{prooftree}
			\infer0[\footnotesize$\ruleManyVal$]{\tyjp{}{\la{\vartwo}(\tmtwo\isub{\var}{\tval})}{}{\emptymset}}
			\end{prooftree}
			\end{equation*}
			where $\sizem{\tderivthree} = 0 = \sizem{\tderiv} + \sizem{\tderivtwo}$ and $\size{\tderivthree} = 0 = 
			\size{\tderiv} + \size{\tderivtwo}$, and $\concl{\tderivthree}{\typctx \mplus 
				\typctxtwo}{\tm\isub{\var}{\tval}}{\mtype}$ since $\dom{\typctx \mplus \typctxtwo} = \emptyset$.
			
			\item\emph{Non-empty multi type}: If $n > 0$, then we can decompose $\tderivtwo$ according to the partitioning 
			$\mtypetwo = \biguplus_{i=1}^n \mtypetwo_i$ by repeatedly applying \reflemmap{typing-value-splitting}{two}, and hence 
			for all $1 \leq i \leq n$ there are context $\typctxtwo_{i}$ and a derivation 
			$\namedtyjp{\tderivtwo_i}{}{\tval}{\typctxtwo_{i} ; \vartwo \hastype \emptytype}{\mtypetwo_i}$ such that 
			$\sizem{\tderivtwo} = \sum_{i=1}^n \sizem{\tderivtwo_{i}}$ and $\size{\tderivtwo} = \sum_{i=1}^{n} 
			\size{\tderivtwo_{i}}$.
			By \ih, for all $1 \leq i \leq n$, there is a derivation 
			$\namedtyjp{\tderivthree_{i}}{}{\tmtwo\isub{\var}{\tval}}{\typctx_i \uplus \typctxtwo_i, \vartwo \hastype 
				\mtypethree_i}{\mtype_i}$ such that $\sizem{\tderivthree_{i}} = \sizem{\tderiv_{i}} + \sizem{\tderivtwo_{i}}$ and 
			$\size{\tderivthree_{i}} \leq \size{\tderiv_{i}} + \size{\tderivtwo_{i}}$.
			Since $\typctx \uplus \typctxtwo = \bigmplus_{i=1}^n (\typctx_i \uplus \typctxtwo_i)$, we can build 
			$\tderivthree$ as
			\begin{equation*}
			\begin{prooftree}[separation=1em]
			\hypo{}
			\ellipsis{$\tderivthree_i$}{\typctx_i \uplus \typctxtwo_i, \vartwo \hastype \mtypethree_i \vdash \tmtwo 
				\isub{\var}{\tval} \hastype \mtype_i}
			\infer1[\footnotesize$\ruleFun$]{\tyjp{}{\la{\vartwo}{(\tmtwo \isub{\var}{\tval})}}{\typctx_{i} \mplus 
					\typctxtwo_{i}}{\ty{\mtypethree_{i}\!}{\!\mtype_{i}}}}
			\delims{\left(}{\right)_{1\leq i \leq n}}
			\hypo{}
			\infer2[\footnotesize$\ruleManyVal$]{\typctx \uplus \typctxtwo \vdash \la{\vartwo}{(\tmtwo \isub{\var}{\tval})} 
				\hastype \mtype}
			\end{prooftree}
			\end{equation*}
			noting that $\sizem{\tderivthree} = \sum_{i=1}^n (\sizem{\tderivthree_{i}} + 1) = \sum_{i=1}^n 
			(\sizem{\tderiv_{i}} + \sizem{\tderivtwo_{i}} + 1) = \sum_{i=1}^{n} (\sizem{\tderiv_{i}} + 1) + \sum_{i=1}^{n} 
			\sizem{\tderivtwo_{i}} = \sizem{\tderiv} + \sizem{\tderivtwo} $ and that $\size{\tderivthree} = \sum_{i=1}^{n} 
			(\size{\tderivthree_{i}} + 1) \leq \sum_{i=1}^{n} (\size{\tderiv_{i}} + \size{\tderivtwo_{i}} + 1) = \sum_{i=1}^{n} 
			(\size{\tderiv_{i}} + 1) + \sum_{i=1}^{n} \size{\tderivtwo_{i}} = \size{\tderiv} + \size{\tderivtwo}$.
		\end{itemize}
		
		\item \emph{Explicit substitution}, \ie $\tm = \tmtwo \esub{\vartwo}{\tmthree}$. 
		We can suppose without loss of generality that $\vartwo \notin \fv{\tval} \cup \{\var \}$, hence $\tm 
		\isub{\var}{\tval} = \tmtwo\isub{\var}{\tval} \esub{\vartwo}{\tmthree\isub{\var}\tval}$ and necessarily
		\begin{equation*}
		\tderiv = 
		\begin{prooftree}
		\hypo{}
		\ellipsis{$\tderiv_{1}$}{\typctx_1 , \var \hastype \mtypetwo_1 , \vartwo \hastype \mtypethree \vdash \tmtwo 
			\hastype \mtype}
		\hypo{}
		\ellipsis{$\tderiv_{2}$}{\typctx_2, \var \hastype \mtypetwo_2 \vdash \tmthree \hastype \mtypethree}
		\infer2[\footnotesize$\Es$]{\typctx, \var \hastype \mtypetwo \vdash \tmtwo \esub{\vartwo}{\tmthree} \hastype \mtype}
		\end{prooftree}
		\end{equation*}
		with $\sizem{\tderiv} = \sizem{\tderiv_{1}} + \sizem{\tderiv_{2}}$, $\size{\tderiv} = \size{\tderiv_{1}} + 
		\size{\tderiv_{2}} + 1$, $\typctx = \typctx_1 \uplus \typctx_2$ and $\mtypetwo = \mtypetwo_2 \uplus \mtypetwo_2$. 
		According to \reflemmap{typing-value-splitting}{two} applied to $\tderivtwo$ and to the decomposition $\mtypetwo = 
		\mtypetwo_1 \uplus \mtypetwo_2$, there are contexts $\typctxtwo_1, \typctxtwo_2$ and derivations 
		$\namedtyjp{\tderivtwo_{1}}{}{\tval}{\typctxtwo_{1}}{\mtypetwo_{1}}$ and 
		$\namedtyjp{\tderivtwo_{2}}{}{\tval}{\typctxtwo_{2}}{\mtypetwo_{2}}$ such that $\typctxtwo = \typctxtwo_{1} \mplus 
		\typctxtwo_{2}$, $\sizem{\tderivtwo} = \sizem{\tderivtwo_{1}} + \sizem{\tderivtwo_{2}}$ and $\size{\tderivtwo} = 
		\size{\tderivtwo_{1}} + \size{\tderivtwo_{2}}$.
		
		By \ih, there are derivations $\namedtyjp{\tderivthree_{1}}{}{\tmtwo\isub{\var}{\tval}}{\typctx_{1} \mplus 
			\typctxtwo_{1}, \vartwo \hastype \mtypethree}{\mtype}$ and $\namedtyjp{\tderivthree_{2}}{}{\tmthree 
			\isub{\var}{\tval}}{\typctx_{2} \mplus \typctxtwo_{2}}{\mtypethree}$ such that $\sizem{\tderivthree_{i}} = 
		\sizem{\tderiv_{i}} + \sizem{\tderivtwo_{i}}$ and $\size{\tderivthree_{i}} \leq \size{\tderiv_{i}} + 
		\size{\tderivtwo_{i}}$ for all $i \in \{1,2\}$.
		Since $\typctx \uplus \typctxtwo = \typctx_1 \uplus \typctxtwo_1 \uplus \typctx_2 \uplus \typctxtwo_2$, we can 
		build the derivation
		\begin{equation*}
		\tderivthree = 
		\begin{prooftree}
		\hypo{}
		\ellipsis{$\tderivthree_1$}{\tyjp{}{\tmtwo\isub{\var}{\tval}}{\typctx_{1} \mplus \typctxtwo_{1}, \vartwo \hastype 
				\mtypethree}{\mtype}}
		\hypo{}
		\ellipsis{$\tderivthree_2$}{\typctx_2 \mplus \typctxtwo_2 \vdash \tmthree\isub{\var}{\tval} \hastype \mtypethree}
		\infer2[\footnotesize$\Es$]{\typctx \uplus \typctxtwo \vdash \tmtwo \isub{\var}{\tval} \esub{\vartwo} {\tmthree 
				\isub{\var}{\tval}} \hastype \mtype}
		\end{prooftree}
		\end{equation*}
		verifying that $\sizem{\tderivthree} = \sizem{\tderivthree_{1}} + \sizem{\tderivthree_{2}} = \sizem{\tderiv_{1}} + 
		\sizem{\tderivtwo_{1}} + \sizem{\tderiv_{2}} + \sizem{\tderivtwo_{2}} = \sizem{\tderiv} + \sizem{\tderivtwo}$ and 
		$\size{\tderivthree} = 1 + \size{\tderivthree_{1}} + \size{\tderivthree_{2}} \leq 1 + (\size{\tderiv_{1}} + 
		\size{\tderivtwo_{1}}) + (\size{\tderiv_{2}} + \size{\tderivtwo_{2}}) = \size{\tderiv} + \size{\tderivtwo}$.
		\qedhere
	\end{itemize}	
\end{proof}

\begin{lemma}[Typing of values: merging]
	\label{l:typing-value-complete} 
	Let $\tval$ be a theoretical value.
	\begin{enumerate}
		\item \label{p:typing-value-complete-empty} There is a derivation $\namedtyjp{\tderiv}{}{\tval}{}{\zero}$.
		
		\item \label{p:typing-value-complete-merge} For every derivations 
		$\namedtyjp{\tderiv_{1}}{}{\tval}{\typctx_{1}}{\mtype_{1}}$ and 
		$\namedtyjp{\tderiv_{2}}{}{\tval}{\typctx_{2}}{\mtype_{2}}$, there exists a  derivation 
		$\namedtyjp{\tderiv}{}{\tval}{\typctx_{1} \mplus \typctx_2}{\mtype_{1} \mplus \mtype_{2}}$.
	\end{enumerate}
\end{lemma}

\begin{proof}
	\begin{enumerate}
		\item 
		Let $\tderiv$ be the following type derivation (applying the rule $\ruleMany$ with $n = 0$):
		\begin{equation*}
		\tderiv =
		\begin{prooftree}
		\hypo{}
		\infer1[\footnotesize$\ruleMany$]{\tyjp{}{\tval}{}{\emptytype}}
		\end{prooftree}\ .
		\end{equation*}
		
		\item Let 
		\begin{equation*}
		\tderiv_{1} =
		\begin{prooftree}
		\hypo{}
		\ellipsis{$\tderiv_{i}$}{\tyjp{}{\tval}{\typctx_{i}}{\ltype_{i}}}
		\delims{ \left( }{ \right)_{\iI} }
		\infer1[\footnotesize$\ruleMany$]{\tyjp{}{\tval}{\bigmplus_{\iI} \typctx_{i}}{\bigmplus_{\iI}\mult{\ltype_{i}}}}
		\end{prooftree}
		\end{equation*}
		with $\typctx_{1} = \bigmplus_{\iI} \typctx_{i}$ and $\mtype_{1} = \bigmplus_{\iI}\mult{\ltype_{i}}$, and let
		\begin{equation*}
		\tderiv_{2} =
		\begin{prooftree}
		\hypo{}
		\ellipsis{$\tderiv_{j}$}{\tyjp{}{\tval}{\typctx_{j}}{\ltype_{j}}}
		\delims{ \left( }{ \right)_{\jJ} }
		\infer1[\footnotesize$\ruleMany$]{\tyjp{}{\tval}{\bigmplus_{\jJ} \typctx_{j}}{\bigmplus_{\jJ}\mult{\ltype_{j}}}}
		\end{prooftree}
		\end{equation*}
		with $\typctx_{2} = \bigmplus_{\jJ} \typctx_{j}$ and $\mtype_{2} = \bigmplus_{\jJ}\mult{\ltype_{j}}$.
		
		We can derive $\tderiv$ by setting $K = I \cup J$ and then
		\begin{equation*}
		\tderiv =
		\begin{prooftree}
		\hypo{}
		\ellipsis{$\tderiv_{k}$}{\tyjp{}{\tval}{\typctx_{k}}{\ltype_{k}}}
		\delims{ \left( }{ \right)_{\kK} }
		\infer1[\footnotesize$\ruleMany$]{\tyjp{}{\tval}{\bigmplus_{\kK} \typctx_{k}}{\bigmplus_{\kK}\mult{\ltype_{k}}}}
		\end{prooftree}
		\end{equation*}
		trivially verifying the statement.
		\qedhere
	\end{enumerate}
\end{proof}

\begin{lemma}[Anti-substitution]
	\label{l:anti-substitution}
	Let $\tm$ be a term, $\val$ be a value, and 
	$\namedtyjp{\tderiv}{}{\tm\isub{\var}{\val}}{\typctx}{\mtype}$
	be a derivation. 
	Then there are two  derivations $\namedtyjp{\tderivtwo}{}{\tm}{\typctxtwo, \var \hastype 
		\mtypetwo}{\mtype}$ 
	and $\namedtyjp{\tderivthree}{}{\val}{\typctxthree}{\mtypetwo}$ such that $\typctx = \typctxtwo \mplus \typctxthree$.
\end{lemma}

\begin{proof}
	By induction on the term $\tm$.
	Cases:
	\begin{itemize}
		\item \emph{Variable}, then are two sub-cases (let $\mtype = \mset{\ltype_1, \dots, \ltype_n}$ for some $n \in 
		\nat$):
		\begin{enumerate}
			\item $\tm = \var$, then $\tm \isub{\var}{\tval} = \tval$. 
			Let $\typctxthree = \typctx$, let $\typctxtwo$ be the empty context (\ie $\dom{\typctxtwo} = \emptyset$), let 
			$\mtypetwo = \mtype$ and let $\tderivtwo$ be the derivation 
			\begin{equation*}
			\tderivtwo = 
			\begin{prooftree}
			\infer0[\footnotesize$\Ax$]{\tyjp{}{\var}{\var \hastype \mset{\ltype_1}}{\ltype_1}}
			\hypo{\overset{n \in \nat}{\ldots}}
			\infer0[\footnotesize$\Ax$]{\tyjp{}{\var}{\var \hastype \mset{\ltype_n}}{\ltype_n}}
			\infer3[\footnotesize$\ruleManyVar$]{\tyjp{}{\var}{\var \hastype 
					\mset{\ltype_1,\dots,\ltype_n}}{\mset{\ltype_1,\dots,\ltype_n}}}
			\end{prooftree}
			\end{equation*}
			Thus, $\concl{\tderivtwo}{\typctxtwo, \var \hastype \mtypetwo}{\tm}{\mtype}$.
			Let $\tderivthree = \tderiv$: so, $\namedtyjp{\tderivthree}{}{\tval}{\typctxthree}{\mtypetwo}$ and $\typctxthree 
			\mplus \typctxtwo = \typctx$.
			
			\item $\tm = \varthree \neq \var$, then $\tm \isub{\var}{\tval} = \varthree$ and 
			the derivation $\tderiv$ has necessarily the form (for some $n~\in~\nat$)
			\begin{equation*}
			\tderiv = 
			\begin{prooftree}
			\infer0[\footnotesize$\Ax$]{\tyjp{}{\varthree}{\varthree \hastype \mset{\ltype_1}}{\ltype_1}}
			\hypo{\overset{n \in \nat}{\ldots}}
			\infer0[\footnotesize$\Ax$]{\tyjp{}{\varthree}{\varthree \hastype \mset{\ltype_n}}{\ltype_n}}
			\infer3[\footnotesize$\ruleManyVar$]{\tyjp{}{\varthree}{\varthree \hastype 
					\mset{\ltype_1,\dots,\ltype_n}}{\mset{\ltype_1,\dots,\ltype_n}}}
			\end{prooftree}
			\end{equation*}
			where $\mtype = \mset{\ltype_1, \dots, \ltype_n}$ and $\typctx = \varthree \hastype \mtype$ (while $\typctx(\var) 
			= \emptymset$).
			Let $\typctxthree$ be the empty context (\ie $\dom{\typctxthree} = 0$) and  $\mtypetwo = \emptymset$.
			By \reflemmap{typing-value-complete}{empty}, there is a derivation 
			$\namedtyjp{\tderivthree}{}{\tval}{}{\emptytype}$ (and hence 
			$\namedtyjp{\tderivthree}{}{\tval}{\typctxthree}{\mtypetwo}$).  
			Let $\typctxtwo = \typctx$  and $\tderivtwo = \tderiv$:
			therefore, $\typctxthree \mplus \typctxtwo = \typctx$
			and $\namedtyjp{\tderivtwo}{}{\tm }{\typctxtwo, \var \hastype \mtypetwo}{\mtype}$.
		\end{enumerate}
		
		\item \emph{Application}, \ie $\tm = \tm_1\tm_2$. 
		Then $\tm \isub{\var}{\tval} = \tm_1 \isub{\var}{\tval} \tm_2 \isub{\var}{\tval}$ and necessarily
		\begin{equation*}
		\tderiv = 
		\begin{prooftree}
		\hypo{}
		\ellipsis{$\tderiv_{1}$}{\typctx_1 \vdash \tm_1\isub{\var}{\tval} \hastype \mset{\larrow{\mtypethree}{\mtype}}}
		\hypo{}
		\ellipsis{$\tderiv_{2}$}{\typctx_2 \vdash \tm_2\isub{\var}{\tval} \hastype \mtypethree}
		\infer2[\footnotesize$\ruleAp$]{\typctx \vdash \tm_1\isub{\var}{\tval} \tm_2\isub{\var}{\tval} \hastype \mtype}
		\end{prooftree}
		\end{equation*}
		with  $\typctx = \typctx_1 \mplus \typctx_2$. 		
		By \ih, for all $i \in \{1,2\}$, there are derivations $\namedtyjp{\tderivtwo_i}{}{\tm_i}{\typctxtwo_i, \var 
			\hastype \mtypetwo_i}{\mult{\larrow{\mtypethree}{\mtype}}}$ and $\namedtyjp{\tderivthree_i}{}{\tval 
		}{\typctxthree_i}{\mtypetwo_i}$ such that $\typctx_i = \typctxtwo_i \mplus \typctxthree_i$.
		According to \reflemmap{typing-value-complete}{merge} applied to $\tderivthree_1$ and $\tderivthree_2$, there is a 
		derivation $\namedtyjp{\tderivthree}{}{\tval}{\typctxthree}{\mtypetwo}$ where $\mtypetwo = \mtypetwo_1 \mplus 
		\mtypetwo_2$ and $\typctxthree = \typctxthree_1 \mplus \typctxthree_2$.
		We can build the derivation (where $\typctxtwo = \typctxtwo_1 \mplus \typctxtwo_2$)
		\begin{equation*}
		\tderivtwo = 
		\begin{prooftree}
		\hypo{}
		\ellipsis{$\tderivtwo_1$}{\typctxtwo_1, \var \hastype \mtypetwo_1 \vdash \tm_1 \hastype 
			\mset{\larrow{\mtypethree}{\mtype}}}
		\hypo{}
		\ellipsis{$\tderivtwo_2$}{\typctxtwo_2, \var \hastype \mtypetwo_2 \vdash \tm_2 \hastype \mtypethree}
		\infer2[\footnotesize$\ruleAp$]{\typctxtwo, \var \hastype \mtypetwo \vdash \tm \hastype \mtype}
		\end{prooftree}
		\end{equation*}

		\item \emph{Abstraction}, \ie $\tm = \la{\vartwo}{\tmtwo}$.
		We can suppose without loss of generality that $\vartwo \notin \fv{\tval} \cup \{\var \}$, hence $\tm 
		\isub{\var}{\tval} = \la{\vartwo}{\tmtwo\isub{\var}\tval}$ and necessarily, for some $n \in \nat$,
		\begin{equation*}
		\tderiv =
		\begin{prooftree}[separation=1em]
		\hypo{}
		\ellipsis{$\tderiv_i$}{\typctx_i, \vartwo \hastype \mtypethree_i \vdash \tmtwo\isub{\var}{\tval} \hastype \mtype_i}
		
		\infer1[\footnotesize$\ruleFun$]{\tyjp{}{\la{\vartwo}{\tmtwo\isub{\var}{\tval}}}{\typctx_{i}}{\ty{\mtypethree_{i}\!}{
					\!\mtype_{1}}}}
		\delims{\left(}{\right)_{1 \leq i \leq n}}
		\hypo{}
		\infer2[\footnotesize$\ruleManyVal$]{\tyjp{}{\la{\vartwo}{\tmtwo\isub{\var}{\tval}}}{\typctx}{\mtype}}
		\end{prooftree}
		\end{equation*}
		with $\typctx = \bigmplus_{i=1}^{n} \typctx_i$, $\mtype = \bigmplus_{i=1}^{n} 
		\mset{\larrow{\mtypethree_i\!}{\!\mtype_i}}$.
		
		There are two sub-cases:
		\begin{itemize}
			\item \emph{Empty multi type}: If $n = 0$,  then $\mtype = \emptymset$ and $\dom{\typctx} = \emptyset$. 
			We can  build the derivation 
			\begin{equation*}
			\tderivtwo = 
			\begin{prooftree}
			\infer0[\footnotesize$\ruleManyVal$]{\tyjp{}{\la{\vartwo}\tmtwo}{}{\emptymset}}
			\end{prooftree}
			\end{equation*}
			Let $\mtypetwo = \emptymset$ and $\typctxtwo$ be the empty context (\ie $\dom{\typctxtwo} = \emptyset$): then 
			$\concl{\tderivtwo}{\typctxtwo, \var \hastype \mtypetwo}{\tm}{\mtype}$.			
			According to \reflemmap{typing-value-complete}{empty}, there is a derivation 
			$\concl{\tderivthree}{}{\tval}{\emptymset}$. 
			Let $\typctxthree$ be the empty context (\ie $\dom{\typctxthree} = \emptyset$): so, 
			$\concl{\tderivthree}{\typctxthree}{\tval}{\mtypetwo}$ with $\typctx = \typctxtwo \mplus \typctxthree$.
			
			\item\emph{Non-empty multi type}: If $n > 0$ then by \ih, for all $1 \leq i \leq n$, there are derivations 
			$\concl{\tderivtwo_i}{\typctxtwo_i, \vartwo \hastype \mtypethree_i, \var \hastype \mtypetwo_i}{\tmtwo}{\mtype_i}$ and 
			$\concl{\tderivthree_i}{\typctxthree_i}{\tval}{\mtypetwo_i}$ such that $\typctx_i = \typctxtwo_i \mplus \typctxthree_i$.
			We can build the derivation
			\begin{equation*}
			\tderivtwo = 
			\begin{prooftree}[separation=1em]
			\hypo{}
			\ellipsis{$\tderivtwo_i$}{\typctxtwo_i ; \vartwo \hastype \mtypethree_i ; \var \hastype \mtypetwo_i \vdash \tmtwo 
				\hastype \mtype_i}
			\infer1[\footnotesize$\ruleFun$]{\tyjp{}{\la{\vartwo}{\tmtwo}}{\typctxtwo_{i} ; \var \hastype 
					\mtypetwo_{i}}{\ty{\mtypethree_{i}\!}{\!\mtype_{i}}}}
			\delims{ \left( }{ \right)_{1 \leq i \leq n} }
			\infer1[\footnotesize$\ruleManyVal$]{\tyjp{}{\la{\vartwo}{\tmtwo}}{\bigmplus_{i=1}^n \typctxtwo_i ; \var \hastype 
					\mplus_{i=1}^n \mtypetwo_i}{\bigmplus_{i=1}^{n} \mult{\ty{\mtypethree_{i}\!}{\!\mtype_{i}}}}}
			\end{prooftree}
			\end{equation*}
			By repeatedly applying \reflemmap{typing-value-complete}{merge}, there is a derivation 
			$\concl{\tderivthree}{\typctxthree}{\tval}{\mtypetwo}$ with $\typctxthree = \bigmplus_{i=1}^n \typctxthree_i$.
			So, $\typctx = \bigmplus_{i=1}^n \typctx_i = \bigmplus_{i=1}^n (\typctxtwo_i \mplus \typctxthree_i) = \typctxtwo 
			\mplus \typctxthree$.
		\end{itemize}
		
		\item \emph{Explicit substitution}, \ie $\tm = \tmtwo \esub{\vartwo}{\tmthree}$. 
		We can suppose without loss of generality that $\vartwo \notin \fv{\tval} \cup \{\var \}$, hence $\tm 
		\isub{\var}{\tval} = \tmtwo\isub{\var}{\tval} \esub{\vartwo}{\tmthree\isub{\var}\tval}$ and necessarily
		\begin{equation*}
		\tderiv = 
		\begin{prooftree}
		\hypo{}
		\ellipsis{$\tderiv_1$}{\tyjp{}{\tmtwo\isub{\var}{\tval}}{\typctx_{1}, \vartwo \hastype \mtypethree}{\mtype}}
		\hypo{}
		\ellipsis{$\tderiv_2$}{\typctx_2 \vdash \tmthree\isub{\var}{\tval} \hastype \mtypethree}
		\infer2[\footnotesize$\Es$]{\typctx \vdash \tmtwo \isub{\var}{\tval} \esub{\vartwo} {\tmthree \isub{\var}{\tval}} 
			\hastype \mtype}
		\end{prooftree}
		\end{equation*}
		By \ih applied to $\tderiv_1$ and \refrmk{free-variables}, there are derivations 
		$\concl{\tderivtwo_1}{\typctxtwo_1, \vartwo \hastype \mtypethree, \var \hastype \mtypetwo_1}{\tmtwo}{\mtype}$ and
		$\concl{\tderivthree_1}{\typctxthree_1}{\tval}{\mtypetwo_1}$ with $\typctx_1 = \typctxtwo_1 \mplus \typctxthree_1$.
		By \ih applied to $\tderiv_2$ , there are derivations $\concl{\tderivtwo_2}{\typctxtwo_2, \var \hastype 
			\mtypetwo_2}{\tmthree}{\mtype}$ and
		$\concl{\tderivthree_2}{\typctxthree_2}{\tval}{\mtypetwo_2}$ with $\typctx_2 = \typctxtwo_2 \mplus \typctxthree_2$.
		According to \reflemmap{typing-value-complete}{merge}, there is a derivation 
		$\namedtyjp{\tderivthree}{}{\tval}{\typctxthree}{\mtypetwo}$ with $\typctxthree = \typctxthree_{1} \mplus 
		\typctxthree_{2}$ and $\mtypetwo = \mtypetwo_1 \mplus \mtypetwo_2$.
		We can build the derivation (where $\typctxtwo = \typctxtwo_1 \mplus \typctxtwo_2$)
		\begin{equation*}
		\tderivtwo = 
		\begin{prooftree}
		\hypo{}
		\ellipsis{$\tderivtwo_{1}$}{\typctxtwo_1 , \var \hastype \mtypetwo_1 , \vartwo \hastype \mtypethree \vdash \tmtwo 
			\hastype \mtype}
		\hypo{}
		\ellipsis{$\tderivtwo_{2}$}{\typctxtwo_2, \var \hastype \mtypetwo_2 \vdash \tmthree \hastype \mtypethree}
		\infer2[\footnotesize$\Es$]{\typctxtwo, \var \hastype \mtypetwo \vdash \tmtwo \esub{\vartwo}{\tmthree} \hastype 
			\mtype}
		\end{prooftree}
		\end{equation*}
		verifying that $\typctx = \typctx_1 \mplus \typctx_2 = \typctxtwo_1 \mplus \typctxthree_1 \mplus \typctxtwo_2 
		\mplus \typctxthree_2 = \typctxtwo \mplus \typctxthree$.
		\qedhere
	\end{itemize}	
\end{proof}

\begin{proposition}[Qualitative subjects]
	\label{propappendix:qual-subject}
	\NoteState{prop:qual-subject}
	Let $\tm \tovsub \tm'$.
	\begin{enumerate}
		\item \label{pappendix:qual-subject-reduction}
		\emph{Reduction}: 	if $\namedtyjp{\tderiv}{}{\tm}{\typctx}{\mtype}$ then there is $\namedtyjp{\tderiv'}{}{\tm'}{\typctx}{\mtype}$ such that $\size{\tderiv} \geq \size{\tderiv'}$.
		
		\item \label{pappendix:qual-subject-expansion}
		\emph{Expansion}: if $\namedtyjp{\tderiv'}{}{\tm'}{\typctx}{\mtype}$ then there is a derivation $\namedtyjp{\tderiv}{}{\tm}{\typctx}{\mtype}$.
	\end{enumerate}
\end{proposition}

\begin{proof}
	\begin{enumerate}
		\item By induction on	the evaluation context $\ctx$ in the step $\tm = \ctxp{\tmtwo} \tovsub \ctxp{\tmtwo'} = \tm'$ with $\tmtwo \rtom \tmtwo'$ or $\tmtwo' \rtoe \tmtwo'$. 
		The proof is analogous to the one for open quantitative subject reduction (\Cref{prop:weak-subject-reduction}), except that now there is one more case:
		\begin{itemize}
			\item \emph{Abstraction}, \ie $\ctx = \la{\var}{\ctxtwo}$. 
			Then, $\tm = \ctxp{\tmtwo} = \la{\var}{\ctxtwop{\tmtwo}} \allowbreak\rootRew{a} \la{\var}{\ctxtwop{\tmtwo'}} = \ctxp{\tmtwo'} = \tm'$ with $\tmtwo \rootRew{a} \tmtwo'$ and $a \in \{\msym, \esym\}$.
			Then, the derivation $\tderiv$ is necessarily (for some $n \in \nat$)
			\begin{equation*}
			\tderiv = 
			\begin{prooftree}[separation=1em]
			\hypo{}
			\ellipsis{$\tderivtwo_i$}{\typctx_i, \var \hastype \mtypetwo_1 \vdash \ctxtwop{\tmtwo} \hastype \mtypethree_i}
			\infer1[\footnotesize$\lambda$]{\typctx_i \vdash \la{\var}\ctxtwop{\tmtwo} \hastype \larrow{\mtypetwo_i}{\mtypethree_i}}
			\delims{ \left( }{ \right)_{1 \leq i \leq n} }
			\infer1[\footnotesize$\ruleManyVal$]{\bigmplus_{i=1}^n \typctx_i \vdash \la{\var}\ctxtwop{\tmtwo} \hastype \bigmplus_{i=1}^{n} \mset{\larrow{\mtypetwo_i}{\mtypethree_i}} }
			\end{prooftree}
			\end{equation*}
			By \ih, for all $1 \leq i \leq n$, there is a derivation $\concl{\tderivtwo_i'}{\typctx_1, \var \hastype \mtypetwo_i}{\ctxtwop{\tmtwo'}}{\mtypethree_i}$ with: 
			$\size{\tderivtwo_i'} \leq \size{\tderivtwo_i}$.
			We can then build the derivation 
			\begin{equation*}
			\tderiv' = 
			\begin{prooftree}[separation=1em]
			\hypo{}
			\ellipsis{$\tderivtwo'_i$}{\typctx_i, \var \hastype \mtypetwo_1 \vdash \ctxtwop{\tmtwop} \hastype \mtypethree_i}
			\infer1[\footnotesize$\lambda$]{\typctx_i \vdash \la{\var}\ctxtwop{\tmtwop} \hastype \larrow{\mtypetwo_i}{\mtypethree_i}}
			\delims{ \left( }{ \right)_{1 \leq i \leq n} }
			\infer1[\footnotesize$\ruleManyVal$]{\bigmplus_{i=1}^n \typctx_i \vdash \la{\var}\ctxtwop{\tmtwop} \hastype \bigmplus_{i=1}^{n} \mset{\larrow{\mtypetwo_i}{\mtypethree_i}} }
			\end{prooftree}
			\end{equation*}
			where $\size{\tderiv'} = 1 + \sum_{i=1}^n \size{\tderivtwo_i'} \leq 1 + \sum_{i=1}^n \size{\tderivtwo_i} = \size{\tderiv}$. Note that $\size{\tderiv'} = \size{\tderiv}$ if $n = 0$.
		\end{itemize}
		
		\item By induction on	the evaluation context $\ctx$ in the step $\tm = \ctxp{\tmtwo} \tovsub \ctxp{\tmtwo'} = \tm'$ with $\tmtwo \rtom \tmtwo'$ or $\tmtwo' \rtoe \tmtwo'$. Cases:
		\begin{itemize}
				\item \emph{Hole}, \ie $\ctx = \ctxhole$.
		Then there are two sub-cases:
		\begin{enumerate}
			\item \emph{Multiplicative}, \ie $\tm = \subctxp{\la\var\tmtwo}\tmthree \rtom  \subctxp{\tmtwo \esub{\var}{\tmthree}} = \tm'$.
			Then $\tderiv'$ has necessarily the form:
			\begin{equation*}
			\tderiv' = 
			\begin{prooftree}
			\hypo{}
			\ellipsis{$\tderivtwo$}{\tyjp{}{\tmtwo}{\typctx' ; \var \hastype \mtypetwo}{\mtype}}
			\hypo{}
			\ellipsis{$\tderivthree$}{\tyjp{}{\tmthree}{\typctxtwo}{\mtypetwo}}
			\infer2[\footnotesize$\ruleES$]{\typctx'\uplus\typctxtwo \vdash\tmtwo\esub\var\tmthree \hastype \mtype}
			\hypo{}
			\ellipsis{$\tderiv_1$}{\quad}
			\infer2[\footnotesize$\ruleES$]{}
			\ellipsis{}{\quad}
			\hypo{}
			\ellipsis{$\tderiv_n$}{\quad}
			\infer2[\footnotesize$\Es$]{\typctx\uplus\typctxtwo \vdash \subctxp{\tmtwo\esub\var\tmthree} \hastype \mtype}
			\end{prooftree}
			\end{equation*}
			
			We can build $\tderiv$ as follows:
			\begin{equation*}
			\tderiv = 
			\begin{prooftree}[separation = 1em]
			\hypo{}
			\ellipsis{$\tderivtwo$}{\typctx', \var \hastype \mtypetwo \vdash \tmtwo \hastype \mtype}
			\infer1[\footnotesize$\ruleFun$]{\typctx' \vdash \la\var\tmtwo \hastype \larrow{\mtypetwo}{\mtype}}
			\infer1[\footnotesize$\ruleManyVal$]{\typctx' \vdash \la\var\tmtwo \hastype \mset{\larrow{\mtypetwo}{\mtype}}}
			\hypo{}
			\ellipsis{$\tderiv_1$}{\quad}
			\infer2[\footnotesize$\ruleES$]{}
			\ellipsis{}{\quad}
			\hypo{}
			\ellipsis{$\tderiv_n$}{\quad}
			\infer2[\footnotesize$\ruleES$]{\typctx \vdash \subctxp{\la\var\tmtwo} \hastype \mset{\larrow{\mtypetwo}{\mtype}}}
			\hypo{}
			\ellipsis{$\tderivthree$}{\typctxtwo\vdash\tmthree \hastype \mtypetwo}
			\infer2[\footnotesize$\ruleApp$]{\typctx \uplus \typctxtwo \vdash \subctxp{\la\var\tmtwo}\tmthree \hastype \mtype}
			\end{prooftree}
			\end{equation*}
			
			\item \emph{Exponential}, \ie $\tm = \tmtwo\esub\var{\subctxp{\tval}} \rtoe \subctxp{\tmtwo \isub{\var}{\tval}} = \tmp$.
			Then the derivation $\tderiv$ has necessarily the form:
			\begin{equation*}
			\tderiv' = 
			\begin{prooftree}
			\hypo{}
			\ellipsis{$\tderiv''$}{\typctxtwo \mplus\typctxthree' \vdash \tmtwo\isub\var \tval \hastype \mtype}
			\hypo{}
			\ellipsis{$\tderiv_1$}{\quad}
			\infer2[\footnotesize{$\Es$}]{}
			\ellipsis{}{}
			\hypo{}
			\ellipsis{$\tderiv_n$}{\quad}
			\infer2[\footnotesize$\Es$]{\typctxtwo\uplus\typctxthree \vdash \subctxp{\tmtwo\isub\var \tval} \hastype \mtype}
			\end{prooftree}
			\end{equation*}
			where $\typctx = \typctxtwo \uplus \typctxthree$.
			By the anti-substitution lemma (\reflemma{anti-substitution}), there are  derivations $\namedtyjp{\tderivtwo}{}{\tmtwo}{\typctxtwo, \var\hastype\mtypetwo}{\mtype}$ and $\namedtyjp{\tderivthree}{}{\tval}{\typctxthree'}{\mtypetwo}$.
			We can then build the following derivation $\tderiv$:
			\begin{equation*}
			\tderiv = 
			\begin{prooftree}
			\hypo{}
			\ellipsis{$\tderivtwo$}{\tyjp{}{\tmtwo}{\typctxtwo, \var \hastype \mtypetwo}{\mtype}}
			\hypo{}
			\ellipsis{$\tderivthree$}{\tyjp{}{\tval}{\typctxthree'}{\mtypetwo}}
			\hypo{}
			\ellipsis{$\tderiv_1$}{\quad}
			\infer2[\footnotesize$\Es$]{}
			\ellipsis{}{\quad}
			\hypo{}
			\ellipsis{$\tderiv_n$}{\quad}
			\infer2[\footnotesize$\Es$]{\typctxthree \vdash \subctxp{\tval} \hastype \mtypetwo}
			\infer2[\footnotesize$\Es$]{\typctxtwo \mplus \typctxthree \vdash\tmtwo\esub\var{\subctxp{\tval}}\hastype \mtype}
			\end{prooftree}
			\end{equation*}
		\end{enumerate}
		
		\item \emph{Application left}, \ie $\ctx = \ctxtwo \tmthree$.
		Then, $\tm = \ctxp{\tmtwo} = \ctxtwop{\tmtwo} \tmthree \rootRew{a} \ctxtwop{\tmtwo'} \tmthree = \ctxp{\tmtwo'} = \tm'$ with $\tmtwo \rootRew{a} \tmtwo'$ and $a \in \{\msym, \esym\}$.
		The derivation $\tderiv'$ is necessarily
		\begin{equation*}
		\tderiv' = 
		\begin{prooftree}
		\hypo{}
		\ellipsis{$\tderivtwo'$}{\tyjp{}{\ctxtwop{\tmtwo'}}{\typctxtwo}{\mult{\ty{\mtypetwo}{\mtype}}}}
		\hypo{}
		\ellipsis{$\tderivthree$}{\tyjp{}{\tmthree}{\typctxthree}{\mtypetwo}}
		\infer2[\footnotesize$\ruleAp$]{\tyjp{}{\ctxtwop{\tmtwo'} \tmthree}{\typctxtwo \mplus \typctxthree}{\mtype}}
		\end{prooftree}
		\end{equation*}
		where $\typctx = \typctxtwo \uplus \typctxthree$.
		By \ih, there is a derivation $\namedtyjp{\tderivtwo}{}{\ctxtwop{\tmtwo}}{\typctxtwo}{\mult{\ty{\mtypetwo}{\mtype}}}$.
		We can then build the following derivation: 
		\begin{equation*}
		\tderiv = 
		\begin{prooftree}
		\hypo{}
		\ellipsis{$\tderivtwo$}{\tyjp{}{\ctxtwop{\tmtwo}}{\typctxtwo}{\mult{\ty{\mtypetwo}{\mtype}}}}
		\hypo{}
		\ellipsis{$\tderivthree$}{\tyjp{}{\tmthree}{\typctxthree}{\mtypetwo}}
		\infer2[\footnotesize$\ruleAp$]{\tyjp{}{\ctxtwop{\tmtwo} \tmthree}{\typctxtwo \mplus \typctxthree}{\mtype}}
		\end{prooftree}
		\end{equation*}

		\item \emph{Application right}, \ie $\ctx = \tmthree \ctxtwo$.
		Analogous to the previous case.
		
		\item \emph{Explicit substitution left}, \ie $\ctx = \ctxtwo\esub{\var}{\tmthree}$. 
		Then, $\tm = \ctxp{\tmtwo} = \ctxtwop{\tmtwo} \esub{\var}{\tmthree} \rootRew{a} \ctxtwop{\tmtwo'}\esub{\var}{\tmthree} = \ctxp{\tmtwo'} = \tm'$ with $\tmtwo \rootRew{a} \tmtwo'$ and $a \in \{\msym, \esym\}$.
		The derivation $\tderiv'$ is necessarily
		\begin{equation*}
		\tderiv' = 
		\begin{prooftree}
		\hypo{}
		\ellipsis{$\tderivtwo'$}{\tyjp{}{\ctxtwop{\tmtwo'}}{\typctxtwo ; \var \hastype \mtypetwo}{\mtype}}
		\hypo{}
		\ellipsis{$\tderivthree$}{\tyjp{}{\tmthree}{\typctxthree}{\mtypetwo}}
		\infer2[\footnotesize$\Es$]{\typctxtwo \mplus \typctxthree \vdash \ctxtwop{\tmtwo'} \esub{\var}{\tmthree} \hastype \mtype}
		\end{prooftree}
		\end{equation*}
		where $\typctx = \typctxtwo \mplus \typctxthree$.
		By \ih, there is a derivation $\namedtyjp{\tderivtwo}{}{\ctxtwop{\tmtwo}}{\typctxtwo, \var \hastype \mtypetwo}{\mtype}$.
		We can then build the following derivation: 
		\begin{equation*}
		\tderiv = 
		\begin{prooftree}
		\hypo{}
		\ellipsis{$\tderivtwo$}{\tyjp{}{\ctxtwop{\tmtwo}}{\typctxtwo ; \var \hastype \mtypetwo}{\mtype}}
		\hypo{}
		\ellipsis{$\tderivthree$}{\tyjp{}{\tmthree}{\typctxthree}{\mtypetwo}}
		\infer2[\footnotesize$\Es$]{\tyjp{}{\ctxtwop{\tmtwo} \esub{\var}{\tmthree}}{\typctxtwo \mplus \typctxthree}{\mtype}}
		\end{prooftree}
		\end{equation*}

		\item \emph{Explicit substitution right}, \ie $\ctx = \tmthree \esub{\var}{\ctxtwo}$. 
		Analogous to the previous case.

			\item \emph{Abstraction}, \ie $\ctx = \la{\var}{\ctxtwo}$. 
			Then, $\tm = \ctxp{\tmtwo} = \la{\var}{\ctxtwop{\tmtwo}} \allowbreak\rootRew{a} \la{\var}{\ctxtwop{\tmtwo'}} = \ctxp{\tmtwo'} = \tm'$ with $\tmtwo \rootRew{a} \tmtwo'$ and $a \in \{\msym, \esym\}$.
			Then, the derivation $\tderiv'$ is necessarily (for some $n \in \nat$)
			\begin{equation*}
			\tderiv' = 
			\begin{prooftree}[separation=1em]
			\hypo{}
			\ellipsis{$\tderivtwo'_i$}{\typctx_i, \var \hastype \mtypetwo_1 \vdash \ctxtwop{\tmtwop} \hastype \mtypethree_i}
			\infer1[\footnotesize$\lambda$]{\typctx_i \vdash \la{\var}\ctxtwop{\tmtwop} \hastype \larrow{\mtypetwo_i}{\mtypethree_i}}
			\delims{ \left( }{ \right)_{1 \leq i \leq n} }
			\infer1[\footnotesize$\ruleManyVal$]{\bigmplus_{i=1}^n \typctx_i \vdash \la{\var}\ctxtwop{\tmtwop} \hastype \bigmplus_{i=1}^{n} \mset{\larrow{\mtypetwo_i}{\mtypethree_i}} }
			\end{prooftree}
			\end{equation*}
			By \ih, for all $1 \leq i \leq n$, there is a derivation $\concl{\tderivtwo_i}{\typctx_1, \var \hastype \mtypetwo_i}{\ctxtwop{\tmtwo}}{\mtypethree_i}$.
			We can then build the following derivation:
			\begin{equation*}
			\tderiv = 
			\begin{prooftree}[separation=1em]
			\hypo{}
			\ellipsis{$\tderivtwo_i$}{\typctx_i, \var \hastype \mtypetwo_1 \vdash \ctxtwop{\tmtwo} \hastype \mtypethree_i}
			\infer1[\footnotesize$\lambda$]{\typctx_i \vdash \la{\var}\ctxtwop{\tmtwo} \hastype \larrow{\mtypetwo_i}{\mtypethree_i}}
			\delims{ \left( }{ \right)_{1 \leq i \leq n} }
			\infer1[\footnotesize$\ruleManyVal$]{\bigmplus_{i=1}^n \typctx_i \vdash \la{\var}\ctxtwop{\tmtwo} \hastype \bigmplus_{i=1}^{n} \mset{\larrow{\mtypetwo_i}{\mtypethree_i}} }
			\end{prooftree}
			\end{equation*}
			\qedhere
		\end{itemize} 

	\end{enumerate}
\end{proof}

\section{Proofs of \Cref{sect:open}}

\subsection{Correctness} 

\begin{proposition}[Open quantitative subject reduction]
	\label{propappendix:weak-subject-reduction}
	\NoteState{prop:weak-subject-reduction}
		Let $\namedtyjp{\tderiv}{}{\tm}{\typctx}{\mtype}$ be a derivation. If $\tm \tovsubo \tm'$ then there is 
$\namedtyjp{\tderiv'}{}{\tm'}{\typctx}{\mtype}$ with $\size{\tderiv} > \size{\tderiv'}$.
\end{proposition}

\begin{proof}
	By induction on the open evaluation context $\weakctx$ such that $\tm = \weakctxp{\tmtwo} \tow \weakctxp{\tmtwo'} = \tm'$ with $\tmtwo \rtom \tmtwo'$ or $\tmtwo' \rtoe \tmtwo'$. 
	Cases for $\weakctx$:
	\begin{itemize}
		\item \emph{Hole}, \ie $\weakctx = \ctxhole$.
		Then there are two sub-cases:
		\begin{enumerate}
			\item \emph{Multiplicative}, \ie $\tm = \subctxp{\la\var\tmtwo}\tmthree \rtom  \subctxp{\tmtwo \esub{\var}{\tmthree}} = \tm'$.
			Then $\tderiv$ has necessarily the form:
			\begin{equation*}
			\tderiv = 
			\begin{prooftree}[separation = 1em]
			\hypo{}
			\ellipsis{$\tderivtwo$}{\typctx', \var \hastype \mtypetwo \vdash \tmtwo \hastype \mtype}
			\infer1[\footnotesize$\ruleFun$]{\typctx' \vdash \la\var\tmtwo \hastype \larrow{\mtypetwo}{\mtype}}
			\infer1[\footnotesize$\ruleManyVal$]{\typctx' \vdash \la\var\tmtwo \hastype \mset{\larrow{\mtypetwo}{\mtype}}}
			\hypo{}
			\ellipsis{$\tderiv_1$}{\quad}
			\infer2[\footnotesize$\ruleES$]{}
			\ellipsis{}{\quad}
			\hypo{}
			\ellipsis{$\tderiv_n$}{\quad}
			\infer2[\footnotesize$\ruleES$]{\typctx \vdash \subctxp{\la\var\tmtwo} \hastype \mset{\larrow{\mtypetwo}{\mtype}}}
			\hypo{}
			\ellipsis{$\tderivthree$}{\typctxtwo\vdash\tmthree \hastype \mtypetwo}
			\infer2[\footnotesize$\ruleApp$]{\typctx \uplus \typctxtwo \vdash \subctxp{\la\var\tmtwo}\tmthree \hastype \mtype}
			\end{prooftree}
			\end{equation*}
			
			with $\size{\tderiv} = 2 + n + \size{\tderivtwo} + \size{\tderivthree} + \sum_{i=1}^{n} \size{\tderiv_{i}}$.
			We can then build $\tderiv'$ as follows:
			\begin{equation*}
			\tderiv' = 
			\begin{prooftree}
			\hypo{}
			\ellipsis{$\tderivtwo$}{\tyjp{}{\tmtwo}{\typctx' ; \var \hastype \mtypetwo}{\mtype}}
			\hypo{}
			\ellipsis{$\tderivthree$}{\tyjp{}{\tmthree}{\typctxtwo}{\mtypetwo}}
			\infer2[\footnotesize$\ruleES$]{\typctx'\uplus\typctxtwo \vdash\tmtwo\esub\var\tmthree \hastype \mtype}
			\hypo{}
			\ellipsis{$\tderiv_1$}{\quad}
			\infer2[\footnotesize$\ruleES$]{}
			\ellipsis{}{\quad}
			\hypo{}
			\ellipsis{$\tderiv_n$}{\quad}
			\infer2[\footnotesize$\Es$]{\typctx\uplus\typctxtwo \vdash \subctxp{\tmtwo\esub\var\tmthree} \hastype \mtype}
			\end{prooftree}
			\end{equation*}
			where $\size{\tderiv'} = 1+ n + \size{\tderivtwo} + \size{\tderivthree} + \sum_{i=1}^{n} \size{\tderiv_{i}} = \size{\tderiv} - 1$.
			
			\item \emph{Exponential}, \ie $\tm = \tmtwo\esub\var{\subctxp{\tval}} \rtoe \subctxp{\tmtwo \isub{\var}{\tval}} = \tmp$.
			Then the derivation $\tderiv$ has necessarily the form:
			\begin{equation*}
			\tderiv = 
			\begin{prooftree}
			\hypo{}
			\ellipsis{$\tderivtwo$}{\tyjp{}{\tmtwo}{\typctxtwo, \var \hastype \mtypetwo}{\mtype}}
			\hypo{}
			\ellipsis{$\tderivthree$}{\tyjp{}{\tval}{\typctxthree'}{\mtypetwo}}
			\hypo{}
			\ellipsis{$\tderiv_1$}{\quad}
			\infer2[\footnotesize$\Es$]{}
			\ellipsis{}{\quad}
			\hypo{}
			\ellipsis{$\tderiv_n$}{\quad}
			\infer2[\footnotesize$\Es$]{\typctxthree \vdash \subctxp{\tval} \hastype \mtypetwo}
			\infer2[\footnotesize$\Es$]{\typctxtwo\uplus\typctxthree \vdash\tmtwo\esub\var{\subctxp{\tval}}\hastype \mtype}
			\end{prooftree}
			\end{equation*}
			where $\typctx = \typctxtwo \uplus \typctxthree$ and $\size{\tderiv} = 1 + n + \size{\tderivtwo} + \size{\tderivthree} + \sum_{i=1}^{n} \size{\tderiv_{i}}$.
			By the substitution lemma (\reflemma{substitution}), there is a derivation $\namedtyjp{\tderiv''}{}{\tmtwo\isub{\var}{\tval}}{\typctxtwo \mplus \typctxthree'}{\mtype}$
			such that  $\size{\tderiv''} \leq \size{\tderivtwo} + \size{\tderivthree}$.
			We can then build the following derivation $\tderiv'$:
			\begin{equation*}
			\tderiv' = 
			\begin{prooftree}
			\hypo{}
			\ellipsis{$\tderiv''$}{\typctxtwo \mplus\typctxthree' \vdash \tmtwo\isub\var \tval \hastype \mtype}
			\hypo{}
			\ellipsis{$\tderiv_1$}{\quad}
			\infer2[\footnotesize{$\Es$}]{}
			\ellipsis{}{}
			\hypo{}
			\ellipsis{$\tderiv_n$}{\quad}
			\infer2[\footnotesize$\Es$]{\typctxtwo\mplus\typctxthree \vdash \subctxp{\tmtwo\isub\var \val} \hastype \mtype}
			\end{prooftree}
			\end{equation*}
			where  $\size{\tderiv'} = n + \size{\tderiv''} + \sum_{i=1}^{n} \size{\tderiv_{i}} \leq n + \size{\tderivtwo} + \size{\tderivthree} + \sum_{i=1}^{n} \size{\tderiv_{i}} < 1 + n + \size{\tderivtwo} + \size{\tderivthree} + \sum_{i=1}^{n} \size{\tderiv_{i}} = \size{\tderiv}$ ($\tderiv'$ contains at least one rule $\Es$ less than $\tderiv$).
		\end{enumerate}
		
		\item \emph{Application left}, \ie $\weakctx = \weakctxtwo\tmthree$.
		Then, $\tm = \weakctxp{\tmtwo} = \weakctxtwop{\tmtwo} \tmthree \rootRew{a} \weakctxtwop{\tmtwo'} \tmthree = \weakctxp{\tmtwo'} = \tm'$ with $\tmtwo \rootRew{a} \tmtwo'$ and $a \in \{\msym, \esym\}$.
		The derivation $\tderiv$ is necessarily
		\begin{equation*}
		\tderiv = 
		\begin{prooftree}
		\hypo{}
		\ellipsis{$\tderivtwo$}{\tyjp{}{\weakctxtwop{\tmtwo}}{\typctxtwo}{\mult{\ty{\mtypetwo}{\mtype}}}}
		\hypo{}
		\ellipsis{$\tderivthree$}{\tyjp{}{\tmthree}{\typctxthree}{\mtypetwo}}
		\infer2[\footnotesize$\ruleAp$]{\tyjp{}{\weakctxtwop{\tmtwo} \tmthree}{\typctxtwo \mplus \typctxthree}{\mtype}}
		\end{prooftree}
		\end{equation*}
		where $\typctx = \typctxtwo \uplus \typctxthree$, $\sizem{\tderiv} = 1 + \sizem{\tderivtwo} + \sizem{\tderivthree}$ and $\size{\tderiv} = 1 + \size{\tderivtwo} + \size{\tderivthree}$.
		By \ih, there is a derivation $\namedtyjp{\tderivtwo'}{}{\weakctxtwop{\tmtwo'}}{\typctxtwo}{\mult{\ty{\mtypetwo}{\mtype}}}$ with $\size{\tderivtwo'} < \size{\tderivtwo}$.
		We can then build the derivation 
		\begin{equation*}
		\tderiv' = 
		\begin{prooftree}
		\hypo{}
		\ellipsis{$\tderivtwo'$}{\tyjp{}{\weakctxtwop{\tmtwo'}}{\typctxtwo}{\mult{\ty{\mtypetwo}{\mtype}}}}
		\hypo{}
		\ellipsis{$\tderivthree$}{\tyjp{}{\tmthree}{\typctxthree}{\mtypetwo}}
		\infer2[\footnotesize$\ruleAp$]{\tyjp{}{\weakctxtwop{\tmtwo'} \tmthree}{\typctxtwo \mplus \typctxthree}{\mtype}}
		\end{prooftree}
		\end{equation*}
		noting that $\size{\tderiv'} = 1 + \size{\tderivtwo'} + \size{\tderivthree} <_{\ih} 1 + \size{\tderivtwo} + \size{\tderivthree} = \size{\tderiv}$.	
		\item \emph{Application right}, \ie $\weakctx = \tmthree \weakctxtwo$.
		Analogous to the previous case.
		
		\item \emph{Explicit substitution left}, \ie $\weakctx = \weakctxtwo\esub{\var}{\tmthree}$. 
		Then, $\tm = \weakctxp{\tmtwo} = \weakctxtwop{\tmtwo} \esub{\var}{\tmthree} \rootRew{a} \weakctxtwop{\tmtwo'}\esub{\var}{\tmthree} = \weakctxp{\tmtwo'} = \tm'$ with $\tmtwo \rootRew{a} \tmtwo'$ and $a \in \{\msym, \esym\}$.
		The derivation $\tderiv$ is necessarily
		\begin{equation*}
		\tderiv = 
		\begin{prooftree}
		\hypo{}
		\ellipsis{$\tderivtwo$}{\tyjp{}{\weakctxtwop{\tmtwo}}{\typctxtwo ; \var \hastype \mtypetwo}{\mtype}}
		\hypo{}
		\ellipsis{$\tderivthree$}{\tyjp{}{\tmthree}{\typctxthree}{\mtypetwo}}
		\infer2[\footnotesize$\Es$]{\tyjp{}{\weakctxtwop{\tmtwo} \esub{\var}{\tmthree}}{\typctxtwo \mplus \typctxthree}{\mtype}}
		\end{prooftree}
		\end{equation*}
		where $\typctx = \typctxtwo \uplus \typctxthree$ and $\size{\tderiv} = 1 + \size{\tderivtwo} + \size{\tderivthree}$.
		By \ih, there is a derivation $\namedtyjp{\tderivtwo'}{}{\weakctxtwop{\tmtwo'}}{\typctxtwo, \var \hastype \mtypetwo}{\mtype}$ with $\size{\tderivtwo'} < \size{\tderivtwo}$ .
		We can then build the derivation 
		\begin{equation*}
		\tderiv' = 
		\begin{prooftree}
		\hypo{}
		\ellipsis{$\tderivtwo'$}{\tyjp{}{\weakctxtwop{\tmtwo'}}{\typctxtwo ; \var \hastype \mtypetwo}{\mtype}}
		\hypo{}
		\ellipsis{$\tderivthree$}{\tyjp{}{\tmthree}{\typctxthree}{\mtypetwo}}
		\infer2[\footnotesize$\Es$]{\typctxtwo \uplus \typctxthree \vdash \weakctxtwop{\tmtwo'} \esub{\var}{\tmthree} \hastype \mtype}
		\end{prooftree}
		\end{equation*}
		noting that $\size{\tderiv'} = \size{\tderivtwo'} + \size{\tderivthree} <_{\ih} \size{\tderivtwo} + \size{\tderivthree} = \size{\tderiv}$.
		
		\item \emph{Explicit substitution right}, \ie $\weakctx = \tmthree \esub{\var}{\weakctxtwo}$. 
		Analogous to the previous case.
		\qedhere
	\end{itemize}
\end{proof}

\begin{theorem}[Open correctness]
	\label{thmappendix:open-correctness}
	\NoteState{thm:open-correctness}
		Let $\derive{\tderiv}{\tm}$ be a derivation.
	Then there is a $\osym$-normalizing evaluation $\deriv \colon \tm \tovsubo^* \fire$ with $\size{\deriv} \leq \size{\tderiv}$.
\end{theorem}

\begin{proof}
	Given the derivation $\concl{\tderiv}{\typctx}{\tm}{\mtype}$, the proof is by induction on the size $\size{\tderiv}$ of $\tderiv$.
	If $\tm$ is normal for $\tovsubo$, then $\tm = \fire$ is a fireball.	
	Otherwise, $\tm$ is not normal for $\tovsubo$ and so $\tm \tovsubo \tmtwo$.
	According to open subject reduction (\Cref{prop:weak-subject-reduction}), there is a derivation $\concl{\tderivtwo}{\typctx}{\tmtwo}{\mtype}$ such that $\size{\tderivtwo} < \size{\tderiv}$.
	By \ih, there exists a fireball $\fire$ and a reduction sequence $\deriv' \colon \tmtwo \tovsubo^* \fire$.
	The $\osym$-evaluation $\deriv$ of the statement is obtained by concatenating the first step $\tm \tovsubo \tmtwo$ and $\deriv'$.		\qedhere
\end{proof}

\subsection{Completeness}

\begin{proposition}[Typability of open normal forms]
	\label{propappendix:precise-open-typability-nf}
	\NoteState{prop:precise-open-typability-nf}
	\begin{enumerate}
		\item \emph{Inert:}\label{pappendix:precise-open-typability-nf-inert} 
		for every inert term $\itm$ and multi type $\mtype$,	there exist a type context $\typctx$ 	and a derivation $\concl{\tderiv}{\typctx}{\itm}{\mtype}$.

		\item\label{pappendix:precise-open-typability-nf-fireball}
		\emph{Fireball:}
		for every fireball $\fire$ there exists a type context $\typctx$ 	and a derivation $\concl{\tderiv}{\typctx}{\fire}{\zero}$.
	\end{enumerate}
\end{proposition}

\begin{proof}
	We prove simultaneously \Cref{p:precise-open-typability-nf-fireball,p:precise-open-typability-nf-inert} by 
		mutual induction on the definition of fireball and inert term.
		Note that \Cref{p:precise-open-typability-nf-inert} is stronger than \Cref{p:precise-open-typability-nf-fireball}.
		Cases for inert terms:
		\begin{itemize}
			\item \emph{Variable}, \ie $\tm = \var$, which is an inert term. 
			Let $\mtype$ be a multi type: hence, $\mtype = \mset{\ltype_1, \dots, \ltype_n}$ for some $n \in \nat$ and some $\ltype_1, \dots, \ltype_n$ linear types.
			We can then build  the derivation 
			\begin{equation*}
			\tderiv = 
			\begin{prooftree}[separation = 1em]
			\infer0[\footnotesize$\ruleAx$]{\var \hastype \mset{\ltype_1} \vdash \var \hastype \ltype_1}
			\hypo{\dots}
			\infer0[\footnotesize$\ruleAx$]{\var \hastype \mset{\ltype_n} \vdash \var \hastype \ltype_n}
			\infer3[\footnotesize$\ruleManyVar$]{\var \hastype \mset{\ltype_1, \dots, \ltype_n} \vdash \var \hastype \mset{\ltype_1, \dots, \ltype_n}}
			\end{prooftree}
			\end{equation*}

			\item \emph{Inert application}, \ie $\itm = \itmtwo \fire$ for some inert term $\itmtwo$ and fireball $\fire$.
			Let $\mtype$ be a multi type.
			By \ih for fireballs, there is a derivation $\concl{\tderivthree}{\typctxthree}{\fire}{\emptytype}$  for some type context $\typctxthree$.
			By \ih for inert terms, since $\mset{\larrow{\emptytype}{\mtype}}$ is a multi type, there is a derivation $\concl{\tderivtwo}{\typctxtwo}{\itmtwo}{\mset{\larrow{\emptytype}{\mtype}}}$ for some type context $\typctxtwo$.
			We can then build the derivation 
			\begin{equation*}
			\tderiv  =
			\begin{prooftree}
			\hypo{}
			\ellipsis{$\tderivtwo$}{\typctxtwo \vdash \itmtwo \hastype \mset{\larrow{\emptytype}{\mtype}}}
			\hypo{}
			\ellipsis{$\tderivthree$}{\typctxthree \vdash \fire \hastype \emptytype}
			\infer2[\footnotesize$\ruleApp$]{\typctxtwo \mplus \typctxthree \vdash \itmtwo \fire \hastype \mtype}				
			\end{prooftree}
			\end{equation*}
			
			\item \emph{Explicit substitution on inert}, \ie $\itm = \itmtwo \esub{\var}{\itmthree}$ for some inert terms $\itmtwo$ and $\itmthree$.
			Let $\mtype$ be a multi type.
			By \ih for inert terms applied to $\itmtwo$, there is a derivation $\concl{\tderivtwo}{\typctxtwo, \var \hastype \mtypetwo}{\itmtwo}{\mtype}$ for some multi type $\mtypetwo$ and type context $\typctxtwo$.
			By \ih for inert terms applied to $\itmthree$, there is a derivation $\concl{\tderivthree}{\typctxthree}{\itmthree}{\mtypetwo}$  for some type context $\typctxthree$.
			We can then build the derivation 
			\begin{equation*}
				\tderiv =
				\begin{prooftree}
				\hypo{}
				\ellipsis{$\tderivtwo$}{\typctxtwo, \var \hastype \mtypetwo \vdash \itmtwo \hastype \mtype}
				\hypo{}
				\ellipsis{$\tderivthree$}{\typctxthree \vdash \itmthree \hastype \mtypetwo}
				\infer2[\footnotesize$\ruleES$]{\typctxtwo \mplus \typctxthree \vdash \itmtwo \esub{\var} {\itmthree} \hastype \mtype}				
				\end{prooftree}
			\end{equation*}
		\end{itemize}
	
		Cases for fireballs:
		\begin{itemize}
			\item \emph{Inert term}, \ie $\fire  = \itm$. It follows from the \ih for inert terms with respect to $\mtype \defeq \zero$.
			
			\item \emph{Abstraction}, \ie $\fire  = \la{\var}{\tmtwo}$.
			We can then build the derivation
			\begin{equation*}
				\tderiv = 
				\begin{prooftree}
				\infer0[\footnotesize$\ruleManyVal$]{\vdash \la{\var}{\tmtwo} \hastype \emptytype}
			\end{prooftree}
			\end{equation*}
			
			\item \emph{Explicit substitution on fireball}, \ie $\fire = \firetwo \esub{\var}{\itm}$ for some fireball $\firetwo$ and inert term $\itm$.
			By \ih for fireballs applied to $\firetwo$, there is a derivation $\concl{\tderivtwo}{\typctxtwo, \var \hastype \mtypetwo}{\firetwo}{\emptytype}$ for some multi type $\mtypetwo$ and type context $\typctxtwo$.
			By \ih for inert terms applied to $\itm$, there is a derivation $\concl{\tderivthree}{\typctxthree}{\itm}{\mtypetwo}$  for some type context $\typctxthree$.
			We can then build the derivation 
			\begin{equation*}
			\tderiv = 
			\begin{prooftree}
			\hypo{}
			\ellipsis{$\tderivtwo$}{\typctxtwo, \var \hastype \mtypetwo \vdash \firetwo \hastype \emptytype}
			\hypo{}
			\ellipsis{$\tderivthree$}{\typctxthree \vdash \itm \hastype \mtypetwo}
			\infer2[\footnotesize$\ruleES$]{\typctxtwo \mplus \typctxthree \vdash \firetwo \esub{\var} {\itm} \hastype \emptytype}				
			\end{prooftree}
			\end{equation*}
			\qedhere
		\end{itemize}
\end{proof}

\begin{theorem}[Open completeness]
	\label{thmappendix:open-completeness}
	\NoteState{thm:open-completeness}
	Let $\deriv \colon \tm \tovsubo^* \fire$ be an $\osym$-normalizing evaluation.
	Then there is a derivation $\concl{\tderiv}{\typctx}{\tm}{\emptytype}$.
\end{theorem}

\begin{proof}
	By induction on the length $\size{\deriv}$ of the $\osym$-evaluation $\deriv$.
	
	If $\size{\deriv} = 0$ then $\tm = \fire$ is $\osym$-normal and hence a fireball.
	According to typability of fireballs (\Cref{prop:precise-open-typability-nf}), there is a derivation $\concl{\tderiv}{\typctx}{\tm}{\emptytype}$.
	
	Otherwise, $\size{\deriv} > 0$ and $\deriv$ is the concatenation of a first step $\tm \tovsubo \tmtwo$ and an evaluation $\deriv' \colon \tmtwo \tovsubo^* \fire$.
	By \ih, there is a derivation $\concl{\tderivtwo}{\typctx}{\tmtwo}{\emptytype}$. 
	By subject expansion (\refpropp{qual-subject}{expansion}), there is a derivation 
	$\concl{\tderiv}{\typctx}{\tm}{\emptytype}$.\qedhere 
\end{proof}

\section{Proofs of \Cref{sect:shrinking}}

\Cref{rmk:merge-split-coshrinking} and \Cref{l:spread-shrinking} below are  used to prove both correctness (\Cref{thm:correctness}) and completeness (\Cref{thm:completeness}) 

\begin{remark}[Merging and splitting \leftsh shrinkingess]
	\label{rmk:merge-split-coshrinking}
	Let $\mtype, \mtypetwo, \mtypethree$ be multi types with $\mtype = \mtypetwo \mplus \mtypethree$; then $\mtype$ is \leftsh   iff $\mtypetwo$ and $\mtypethree$ are \leftsh  .
	Similarly for type contexts.
\end{remark}

\begin{lemma}
	[Spreading of \leftsh shrinkingness]
	\label{lappendix:spread-shrinking}
	\NoteState{l:spread-shrinking}
Let $\concl{\tderiv}{\typctx}{\ptm}{\mtype}$ be a derivation and $\ptm$ a \pointed term. 
		If $\typctx$ is \leftsh then $\mtype$ is \leftsh.
	\end{lemma}

\begin{proof}
	By induction on the definition of the \pointed term $\ptm$. 
	Cases:
	\begin{itemize}
		\item \emph{Variable}, \ie $\ptm = \var$. 
		Then necessarily, for some $n \in \nat$, 
		\begin{equation*}
		\tderiv = 
		\begin{prooftree}
		\infer0[\footnotesize$\Ax$]{\var \hastype \mset{\ltype_1} \vdash \var \hastype \ltype_1}
		\hypo{\overset{n \in \nat}{\ldots}}
		\infer0[\footnotesize{$\Ax$}]{\var \hastype \mset{\ltype_n} \vdash \var \hastype \ltype_n}
		\infer3[\footnotesize$\ruleManyVar$]{\typctx \vdash \var \hastype \mtype}
		\end{prooftree}
		\end{equation*}
		with $\mtype = \mset{\ltype_1, \dots, \ltype_n}$ and $\typctx = \var \hastype \mtype$.
		Since $\typctx$ is a \leftsh   type context, $\mtype$ is a \leftsh   multi type.
		
		\item \emph{Application}, \ie $\ptm = \ptmtwo \tm$ where $\ptmtwo$ is a \pointed term.
		The derivation $\tderiv$ is necessarily (with $\typctx = \typctxtwo \mplus \typctxthree$)
		\begin{equation*}
		\tderiv = 
		\begin{prooftree}
		\hypo{}
		\ellipsis{$\tderivtwo$}{\typctxtwo \vdash \ptmtwo \hastype \mset{\larrow{\mtypetwo}{\mtype}}}
		\hypo{}
		\ellipsis{$\tderivthree$}{\typctxthree \vdash \tm \hastype \mtypetwo}
		\infer2[\footnotesize$\ruleAp$]{\typctxtwo \uplus \typctxthree \vdash \ptmtwo \tm \hastype \mtype}
		\end{prooftree}
		\end{equation*}
		where $\typctxtwo$ and $\typctxthree$ are \leftsh   type contexts by 
		\refrmk{merge-split-coshrinking}. 
		By \ih, $\mset{\larrow{\mtypetwo}{\mtype}}$ is a \leftsh   multi type, and hence $\mtype$ is a \leftsh   multi type.
		
		\item \emph{Explicit substitution}, \ie $\ptm = \ptmtwo \esub{\var}{\ptmthree}$ where $\ptmtwo$ and $\ptmthree$ are 
		\pointed terms.
		The derivation $\tderiv$ is necessarily (with $\typctx = \typctxtwo \mplus \typctxthree$)
		\begin{equation*}
		\tderiv = 
		\begin{prooftree}
		\hypo{}
		\ellipsis{$\tderivtwo$}{\typctxtwo, \var \hastype \mtypetwo \vdash \ptmtwo \hastype \mtype}
		\hypo{}
		\ellipsis{$\tderivthree$}{\typctxthree \vdash \ptmthree \hastype \mtypetwo}
		\infer2[\footnotesize$\Es$]{\typctxtwo \uplus \typctxthree \vdash \ptm\esub{\var}{\ptmthree} \hastype \mtype}
		\end{prooftree}
		\end{equation*}
		where $\typctxtwo$ and $\typctxthree$ are \leftsh   contexts by 
		\refrmk{merge-split-coshrinking}. 
		By \ih applied to $\tderivthree$ (as $\ptmthree$ is a \pointed term), $\mtypetwo$ is a \leftsh   multi type, and hence $\typctxtwo, \var \hastype \mtypetwo$ is a \leftsh   type context.
		By \ih applied to $\tderivtwo$ (as $\ptmtwo$ is a \pointed term), $\mtype$ is a \leftsh   multi type.
		\qedhere
	\end{itemize}
\end{proof}

\subsection{Correctness}

We split the proof of shrinking quantitative subject reduction in two cases:
\begin{enumerate}
	\item \emph{Open:} after a $\tovsubo$ step any type derivation strictly shrinks its size 
	(\Cref{prop:weak-subject-reduction});
	\item \emph{Strong:} after a $\tovsubs$ step the derivation strictly shrinks its size if it is shrinking 
	(\Cref{prop:shrinking-subject-reduction}). 
\end{enumerate}

The open case (\Cref{prop:weak-subject-reduction}) represents the base case for the strong case 
(\Cref{prop:shrinking-subject-reduction}).  
To prove the base case, we need a substitution lemma (\Cref{l:substitution}) to handle the exponential step, which in turn relies on \Cref{l:typing-value-splitting}. We also need the spreading of left types (\Cref{l:spread-shrinking}).

\begin{proposition}[Shrinking quantitative subject reduction for $\toess$]
	\label{propappendix:shrinking-subject-reduction}
	\NoteState{prop:shrinking-subject-reduction}
	Let $\typctx$ be a \leftsh   context.
	Suppose that $\concl{\tderiv}{\typctx}{\tm}{\mtype}$ and that if $\tm$ is a \valES then $\mtype$ is \rightsh. If $\tm \toess \tm'$ then there is a derivation $\namedtyjp{\tderiv'}{}{\tm'}{\typctx}{\mtype}$ with $\size{\tderiv} > \size{\tderiv'}$.
\end{proposition}

\begin{proof}
	The proof is by induction on the evaluation strong context $\strongctx$ such that $\tm = \strongctxp{\tmtwo} \toess 
	\strongctxp{\tmtwo'} = \tm'$ with $\tmtwo \tovsubo \tmtwo'$. 
	Cases for $\strongctx$ (excluding the two cases that are given in the body of the paper):
	\begin{itemize}
		\item \emph{Hole}, \ie{} $\strongctx = \ctxhole$ and $\tm \tovsubo  \tm'$.
		According to open subject reduction (\Cref{prop:weak-subject-reduction}), there exists a derivation $\concl{\tderiv'}{\typctx}{\tm'}{\mtype}$ such that $\size{\tderiv'} < \size{\tderiv}$.
		
\item \emph{Abstraction}, \ie $\strongctx = \la{\var}{\strongctxtwo}$. 
		So, $\tm = \strongctxp{\tmtwo} = \la{\var}{\strongctxtwop{\tmtwo}} \toess 
		\la{\var}{\strongctxtwop{\tmtwo'}} = \strongctxp{\tmtwo'} = \tm'$ with $\tmtwo \tovsubo \tmtwo'$.
		Since $\tm$ is an \valES, $\mtype$ is a \rightsh multi type by hypothesis and hence it has the form $\mtype = 
		\mset{\larrow{\mtypethree_1}{\mtypetwo_1}, \dots, \larrow{\mtypethree_n}{\mtypetwo_n}}$ for some $n > 0$, where 
		$\mtypethree_i$ is \leftsh and $\mtypetwo_i$ is \rightsh for all $1 \leq i \leq n$.
		Thus, the derivation $\tderiv$ has the following form:
		\begin{equation*}
		\tderiv = 
		\begin{prooftree}[separation=1em]
		\hypo{}
		\ellipsis{$\tderivtwo_i$}{\typctx_i, \var \hastype \mtypethree_i \vdash \strongctxtwop{\tmtwo} \hastype \mtypetwo_i}
		\infer1[\footnotesize$\lambda$]{\typctx_i \vdash \la{\var}\strongctxtwop{\tmtwo} \hastype 
			\larrow{\mtypethree_i}{\mtypetwo_i}}
		\delims{ \left( }{ \right)_{1 \leq i \leq n} }
		\infer1[\footnotesize$\ruleManyVal$]{ \bigmplus_{i=1}^n \typctx_{i} \vdash \la{\var}\strongctxtwop{\tmtwo} \hastype 
			\bigmplus_{i=1}^n \mset{\larrow{\mtypethree_i}{\mtypetwo_i}}}
		\end{prooftree}
		\end{equation*}
		For all $1 \leq i \leq n$, by \ih (as $\typctx_i, \var \hastype \mtypethree_i$ is a \leftsh type context and $\mtypetwo_i$ is a \rightsh multi type), there is a derivation $\concl{\tderivtwo_i'}{\typctx_i, \var \hastype 
			\mtypethree_i}{\strongctxtwop{\tmtwo'}}{\mtypetwo_i}$ with $\size{\tderivtwo_i'} < \size{\tderivtwo_i}$.
		We can then build the derivation 
		\begin{equation*}
		\tderiv' = 
		\begin{prooftree}[separation=1em]
		\hypo{}
		\ellipsis{$\tderivtwop_i$}{\typctx_i, \var \hastype \mtypethree_i \vdash \strongctxtwop{\tmtwo'} \hastype 
			\mtypetwo_i}
		\infer1[\footnotesize$\lambda$]{\typctx_i \vdash \la{\var}\strongctxtwop{\tmtwo'} \hastype 
			\larrow{\mtypethree_i}{\mtypetwo_i}}
		\delims{ \left( }{ \right)_{1 \leq i \leq n} }
		\infer1[\footnotesize$\ruleManyVal$]{ \bigmplus_{i=1}^n \typctx_{i} \vdash \la{\var}\strongctxtwop{\tmtwop} 
			\hastype  \bigmplus_{i=1}^n \mset{\larrow{\mtypethree_i}{\mtypetwo_i}}}
		\end{prooftree}
		\end{equation*}
		where $\size{\tderiv'} = \sum_{i=1}^n(\size{\tderivtwo_i'} +1) < \sum_{i=1}^n(\size{\tderivtwo_i} + 1) = \size{\tderiv}$.
		
		\item \emph{Explicit substitution of rigid context}, \ie $\strongctx = \tmthree \esub{\var}{\ictx}$. 
		So, $\tm = \strongctxp{\tmtwo} = \tmthree\esub{\var}{\ictxp{\tmtwo}} \toess
		\tmthree\esub{\var}{\ictxp{\tmtwo'}} = \strongctxp{\tmtwo'} = \tm'$ with $\tmtwo \tovsubo \tmtwo'$.
		Then, necessarily:
		\begin{equation*}
		\tderiv = 
		\begin{prooftree}
		\hypo{}
		\ellipsis{$\tderivtwo$}{\typctxtwo, \var \hastype \mtypetwo \vdash \tmthree \hastype \mtype}
		\hypo{}
		\ellipsis{$\tderivthree$}{\typctxthree \vdash \ictxp{\tmtwo} \hastype \mtypetwo}
		\infer2[\footnotesize$\Es$]{\typctxtwo \uplus \typctxthree \vdash \tmthree \esub{\var}{\ictxp{\tmtwo}} \hastype 
			\mtype}
		\end{prooftree}
		\end{equation*}
		with $\typctx = \typctxtwo \uplus \typctxthree$  \leftsh by hypothesis, and then so is $\typctxthree$ by \Cref{rmk:merge-split-coshrinking}.
		By \ih applied to $\tderivthree$ (as $\ictxp{\tmtwo}$ is not an \valES), there is a derivation 
		$\concl{\tderivthree'}{\typctxthree}{\ictxp{\tmtwo'}}{\mtypetwo}$~~with $\size{\tderivthree'} < \size{\tderivthree}$.
		We can then build the following derivation: 
		\begin{equation*}
		\tderiv' = 
		\begin{prooftree}
		\hypo{}
		\ellipsis{$\tderivtwo$}{\typctxtwo, \var \hastype \mtypetwo \vdash \tmthree \hastype \mtype}
		\hypo{}
		\ellipsis{$\tderivthree'$}{\typctxthree \vdash \ictxp{\tmtwo'} \hastype \mtypetwo}
		\infer2[\footnotesize$\Es$]{\typctxtwo \uplus \typctxthree \vdash \tmthree \esub{\var}{\ictxp{\tmtwo'} } \hastype 
			\mtype}
		\end{prooftree}
		\end{equation*}
		where $\typctx = \typctxtwo \uplus \typctxthree$ and $\size{\tderiv'} = \size{\tderivtwo} + \size{\tderivthree'} +1 < \size{\tderivtwo} + \size{\tderivthree} +1 = \size{\tderiv}$.
		
		\item \emph{Strong context with explicit substitution of rigid term}, \ie $\strongctx = 
		\strongctxtwo\esub{\var}{\ptm}$. 
		Then, $\tm = \strongctxp{\tmtwo} = \strongctxtwop{\tmtwo} \esub{\var}{\ptm} \toess 
		\strongctxtwop{\tmtwo'}\esub{\var}{\ptm} = \strongctxp{\tmtwo'} = \tm'$ with $\tmtwo \tovsubo \tmtwo'$.
		The derivation $\tderiv$ has the following shape:
		\begin{equation*}
		\tderiv = 
		\begin{prooftree}
		\hypo{}
		\ellipsis{$\tderivtwo$}{\typctxtwo, \var \hastype \mtypetwo \vdash \strongctxtwop{\tmtwo} \hastype \mtype}
		\hypo{}
		\ellipsis{$\tderivthree$}{\typctxthree \vdash \ptm \hastype \mtypetwo}
		\infer2[\footnotesize$\Es$]{\typctxtwo \uplus \typctxthree \vdash \strongctxtwop{\tmtwo} \esub{\var}{\ptm} \hastype 
			\mtype}
		\end{prooftree}
		\end{equation*}
		where $\typctx = \typctxtwo \uplus \typctxthree$ is   \leftsh by hypothesis, and then so are 	$\typctxtwo$ and $\typctxthree$ by \Cref{rmk:merge-split-coshrinking}.
		According to spreading of \leftsh shrinkingness applied to $\tderivthree$ (\Cref{l:spread-shrinking}, which can be applied because $\ptm$ is a rigid term), $\mtypetwo$ is   \leftsh.
		Note that $\strongctxtwop{\tmtwo}\esub{\var}{\ptm}$ is an \valES if and only if $\strongctxtwop{\tmtwo}$ is an \valES.
		So, the \ih can be applied to $\tderivtwo$ (since $\typctxtwo, \var \hastype \mtypetwo$ is a   \leftsh	context) and hence there is a derivation $\concl{\tderivtwo'}{\typctxtwo, \var \hastype 
			\mtypetwo}{\strongctxtwop{\tmtwo'}}{\mtype}$ with $\size{\tderivtwo'} < \size{\tderivtwo}$.
		We can then build the following derivation:
		\begin{equation*}
		\tderiv' = 
		\begin{prooftree}
		\hypo{}
		\ellipsis{$\tderivtwo'$}{\typctxtwo, \var \hastype \mtypetwo \vdash \strongctxtwop{\tmtwo'} \hastype \mtype}
		\hypo{}
		\ellipsis{$\tderivthree$}{\typctxthree \vdash \ptm \hastype \mtypetwo}
		\infer2[\footnotesize$\Es$]{\typctxtwo \uplus \typctxthree \vdash \strongctxtwop{\tmtwo'} \esub{\var}{\ptm} 
			\hastype \mtype}
		\end{prooftree}
		\end{equation*}
		where $\typctx = \typctxtwo \uplus \typctxthree$ and $\size{\tderiv'} = \size{\tderivtwo'} + \size{\tderivthree} +1 < \size{\tderivtwo} + 
			\size{\tderivthree} +1 = \size{\tderiv}$.

		\item \emph{Rigid context with explicit substitution of rigid term}, \ie $\strongctx = \ictx\esub{\var}{\ptm}$. 
		Then, $\tm = \strongctxp{\tmtwo} = \ictxp{\tmtwo} \esub{\var}{\ptm} \toess
		\ictxp{\tmtwo'}\esub{\var}{\ptm} = \strongctxp{\tmtwo'} = \tm'$ with $\tmtwo \tovsubo \tmtwo'$.
		The derivation $\tderiv$ has the following shape:
		\begin{equation*}
		\tderiv = 
		\begin{prooftree}
		\hypo{}
		\ellipsis{$\tderivtwo$}{\typctxtwo, \var \hastype \mtypetwo \vdash \ictxp{\tmtwo} \hastype \mtype}
		\hypo{}
		\ellipsis{$\tderivthree$}{\typctxthree \vdash \ptm \hastype \mtypetwo}
		\infer2[\footnotesize$\Es$]{\typctxtwo \uplus \typctxthree \vdash \ictxp{\tmtwo} \esub{\var}{\ptm} \hastype \mtype}
		\end{prooftree}
		\end{equation*}
		where $\typctx = \typctxtwo \mplus \typctxthree$ is   \leftsh by hypothesis, and then so are $\typctxtwo$ and $\typctxthree$ by \Cref{rmk:merge-split-coshrinking}.
		According to spreading of \leftsh shrinkingness applied to $\tderivthree$ (\Cref{l:spread-shrinking}, which can be applied because $\ptm$ is a rigid term), $\mtypetwo$ is   \leftsh.
		Thus, the \ih can be applied to $\tderivtwo$ (since $\typctxtwo, \var \hastype \mtypetwo$ is   \leftsh and $\ictxp{\tmtwo}$ is not an \valES) to obtain a derivation $\concl{\tderivtwo'}{\typctxtwo, \var \hastype \mtypetwo}{\ictxp{\tmtwo'}}{\mtype}$ such that $\size{\tderivtwo'} < \size{\tderivtwo}$.
		We can then build the following derivation:
		\begin{equation*}
		\tderiv' = 
		\begin{prooftree}
		\hypo{}
		\ellipsis{$\tderivtwo'$}{\typctxtwo, \var \hastype \mtypetwo \vdash \ictxp{\tmtwo'} \hastype \mtype}
		\hypo{}
		\ellipsis{$\tderivthree$}{\typctxthree \vdash \ptm \hastype \mtypetwo}
		\infer2[\footnotesize$\Es$]{\typctxtwo \uplus \typctxthree \vdash \ictxp{\tmtwo'} \esub{\var}{\ptm} \hastype \mtype}
		\end{prooftree}
		\end{equation*}
		where $\typctx = \typctxtwo \uplus \typctxthree$ and $\size{\tderiv'} = \size{\tderivtwo'} + \size{\tderivthree} < \size{\tderivtwo} + 
			\size{\tderivthree} = \size{\tderiv}$.

		\item \emph{Rigid term with explicit substitution of rigid context}, \ie $\strongctx = \ptm\esub{\var}{\ictx}$. 
		Then, $\tm = \strongctxp{\tmtwo} = \ptm\esub{\var}{\ictxp{\tmtwo}} \toess \ptm\esub{\var}{\ictxp{\tmtwo'}} 
		= \strongctxp{\tmtwo'} = \tm'$ with $\tmtwo \tovsubo \tmtwo'$.
		The derivation $\tderiv$ has the following shape:
		\begin{equation*}
		\tderiv = 
		\begin{prooftree}
		\hypo{}
		\ellipsis{$\tderivtwo$}{\typctxtwo, \var \hastype \mtypetwo \vdash \ptm \hastype \mtype}
		\hypo{}
		\ellipsis{$\tderivthree$}{\typctxthree \vdash \ictxp{\tmtwo} \hastype \mtypetwo}
		\infer2[\footnotesize$\Es$]{\typctxtwo \uplus \typctxthree \vdash \ptm \esub{\var}{\ictxp{\tmtwo}} \hastype \mtype}
		\end{prooftree}
		\end{equation*}
		with $\typctx = \typctxtwo \mplus \typctxthree$   \leftsh by hypothesis, and then so is 	$\typctxthree$ by \Cref{rmk:merge-split-coshrinking}.
		By \ih applied to $\tderivthree$ (as $\ictxp{\tmtwo}$ is not an \valES), there is a derivation 
		$\concl{\tderivthree'}{\typctxthree}{\ictxp{\tmtwo'}}{\mtypetwo}$~with $\size{\tderivthree'} < \size{\tderivthree}$.
		We can then build the following derivation: 
		\begin{equation*}
		\tderiv' = 
		\begin{prooftree}
		\hypo{}
		\ellipsis{$\tderivtwo$}{\typctxtwo, \var \hastype \mtypetwo \vdash \ptm \hastype \mtype}
		\hypo{}
		\ellipsis{$\tderivthree'$}{\typctxthree \vdash \ictxp{\tmtwo'} \hastype \mtypetwo}
		\infer2[\footnotesize$\Es$]{\typctxtwo \uplus \typctxthree \vdash \ptm \esub{\var}{\ictxp{\tmtwo'} } \hastype 
			\mtype}
		\end{prooftree}
		\end{equation*}
		where $\typctx = \typctxtwo \uplus \typctxthree$ and $\size{\tderiv'} = \size{\tderivtwo} + \size{\tderivthree'} +1 < \size{\tderivtwo} + 
			\size{\tderivthree} +1 = \size{\tderiv}$.\qedhere
	\end{itemize}
\end{proof}

\begin{theorem}[Shrinking correctness for $\toess$]
	\label{thmappendix:correctness}
	\NoteState{thm:correctness}
	Let
	$\derive{\tderiv}{\tm}$ be a shrinking derivation.
	Then there is a $\esssym$-normalizing evaluation $\deriv \colon \tm \toess^* \sfire$ with $\size{\deriv} \leq \size{\tderiv}$.
\end{theorem}

\begin{proof}
	Given the derivation $\concl{\tderiv}{\typctx}{\tm}{\mtype}$, the proof is by induction on the size $\size{\tderiv}$ of $\tderiv$.
	If $\tm$ is normal for $\toess$, then $\tm = \sfire$ is a strong fireball.	
	Otherwise, $\tm$ is not normal for $\toess$ and so $\tm \toess \tmtwo$.
	According to shrinking subject reduction (\Cref{prop:shrinking-subject-reduction}), there is a derivation $\concl{\tderivtwo}{\typctx}{\tmtwo}{\mtype}$ such that $\size{\tderivtwo} < \size{\tderiv}$.
	By \ih, there exists a strong fireball $\sfire$ and a reduction sequence $\deriv' \colon \tmtwo \toess^* \sfire$.
	The $\esssym$-evaluation $\deriv$ of the statement is obtained by concatenating the first step $\tm \toess \tmtwo$ and $\deriv'$.		\qedhere
\end{proof}

\subsection{Completeness}

The proof of \Cref{prop:typability-normal} relies on the presence of the ground type $\ground$ that is both unitary \leftsh and unitary \rightsh. 

\begin{lemma}[Shrinking typability of normal forms]\hfill
	\label{propappendix:typability-normal}
	\NoteState{prop:typability-normal}
	\begin{enumerate}
		\item\label{papppendix:typability-normal-inert}
		\emph{Inert:}
		for every \full inert term $\sitm$ and \leftsh multi type $\mtype$,
		there exist a \leftsh type context $\typctx$ 
		and a derivation $\concl{\tderiv}{\typctx}{\sitm}{\mtype}$.

		\item\label{pappendix:typability-normal-fireball}
		\emph{Fireball:}
		for every \full fireball $\sfire$ 
		there exists a shrinking derivation $\derive{\tderiv}{\sfire}$.
	\end{enumerate}
\end{lemma}

\begin{proof}
	Both points are proved by mutual induction on the definition of \full inert terms $\sitm$ and \full fireballs 
	$\sfire$.
	Cases of strong inert terms:
	\begin{itemize}
		\item \emph{Variable}, \ie $\sitm = \var = \sfire$. 
		Let $\mtype = \mset{\ltype_1, \dots, \ltype_n}$ be a \leftsh   multi type for some 	$n \in \nat$.
		We can build the following derivation
		\begin{equation*}
		\tderiv = 
		\begin{prooftree}
		\infer0[\footnotesize$\Ax$]{\var \hastype \mset{\ltype_1} \vdash \var \hastype \ltype_1}
		\hypo{\overset{n \in \nat}{\ldots}}
		\infer0[\footnotesize$\Ax$]{\var \hastype \mset{\ltype_n} \vdash \var \hastype \ltype_n}
		\infer3[\footnotesize$\ruleManyVar$]{\var \hastype \mtype \vdash \var \hastype \mtype}
		\end{prooftree}
		\end{equation*}
		where $\typctx = \var \hastype \mtype$ is a \leftsh   type context.
		
		\item \emph{Inert application}, \ie $\sitm = \sitmtwo \sfire$.
		Let $\mtype$ be a \leftsh   multi type.
		By \ih (point 2), there is a derivation $\concl{\tderivtwo}{\typctxtwo}{\sfire}{\mtypetwo}$ where $\typctxtwo$ is   \leftsh and $\mtypetwo$ is   \rightsh,
		thus $\mset{\larrow{\mtypetwo}{\mtype}}$ is \leftsh  .
		By \ih (point 1), there is a derivation $\concl{\tderivthree}{\typctxthree}{\sitm}{\mset{\larrow{\mtypetwo}{\mtype}}}$ where 
		$\typctxthree$ is \leftsh  .
		We have the following derivation:
		\begin{equation*}
		\tderiv = 
		\begin{prooftree}
		\hypo{}
		\ellipsis{$\tderivthree$}{\typctxthree \vdash \sitm \hastype \mset{\larrow{\mtypetwo}{\mtype}}}
		\hypo{}
		\ellipsis{$\tderivtwo$}{\typctxtwo \vdash \sfire \hastype \mtypetwo}
		\infer2[\footnotesize$\ruleAp$]{\typctx \vdash \sitm\sfire \hastype \mtype}
		\end{prooftree}
		\end{equation*}
		where $\typctx = \typctxtwo \uplus \typctxthree$ is \leftsh   by \Cref{rmk:merge-split-coshrinking}.
		
		\item \emph{Explicit substitution on inert}, \ie $\sitm = \sitmtwo \esub{\var}{\sitmthree}$.
		Let $\mtype$ be a \leftsh   multi type.
		By \ih (point 1), there is a derivation $\concl{\tderivtwo}{\typctxtwo, \var \hastype \mtypetwo}{\sitmtwo}{\mtype}$ where $\typctxtwo, \var \hastype \mtypetwo$ is a \leftsh   type context;  
		in particular, $\mtypetwo$ is a \leftsh   multi type.
		By \ih (point 1), there is a derivation $\concl{\tderivthree}{\typctxthree}{\sitmthree}{\mtypetwo}$ where $\typctxthree$ is a \leftsh   context.
		We have the following derivation:
		\begin{equation*}
		\tderiv = 
		\begin{prooftree}
		\hypo{}
		\ellipsis{$\tderivtwo$}{\typctxtwo, \var \hastype \mtypetwo \vdash \sitmtwo \hastype \mtype}
		\hypo{}
		\ellipsis{$\tderivthree$}{\typctxthree \vdash \sitmthree \hastype \mtypetwo}
		\infer2[\footnotesize$\ruleAp$]{\typctx \vdash \sitmtwo \esub{\var}{\sitmthree} \hastype \mtype}
		\end{prooftree}
		\end{equation*}
		where $\typctx = \typctxtwo \uplus \typctxthree$ is \leftsh (resp.~  \leftsh) by 	\Cref{rmk:merge-split-coshrinking}.
	\end{itemize}

Cases of strong fireballs:
		\begin{itemize}
		\item \emph{Strong inert}, \ie $\sfire = \sitm$. Then it follows by the \ih (point 1) applied with respect to the left multi type $\mset\ground$ which is also a right multi type, thus the \ih does provides a shrinking derivation.
		 
		\item \emph{Abstraction}, \ie $\sfire = \la{\var}{\sfiretwo}$. 
		By \ih (point 2), there exists a derivation $\concl{\tderivtwo}{\typctx, \var \hastype \mtypethree}{\sfire}{\mtypetwo}$ where 
		$\typctx, \var \hastype \mtypethree$ is a   \leftsh type context and $\mtypetwo$ is a   \rightsh multi 		type.
		Hence, $\mtype = \mset{\larrow{\mtypethree}{\mtypetwo}}$ is a   \leftsh multi type.
		We have the following derivation:
		\begin{equation*}
		\tderiv = 
		\begin{prooftree}
		\hypo{}
		\ellipsis{$\tderivtwo$}{\typctx, \var \hastype \mtypethree \vdash \sfire \hastype \mtypetwo}
		\infer1[\footnotesize$\lambda$]{\typctx \vdash \la{\var}\sfire \hastype \larrow{\mtypethree}{\mtypetwo}}
		\infer1[\footnotesize$\ruleMany$]{\typctx \vdash \la{\var}\sfire \hastype \mtype}
		\end{prooftree}
		\end{equation*}
		where $\typctx$ is a   \leftsh type context.
		
		\item \emph{Explicit substitution on fireball}, \ie $\sfire = \sfiretwo \esub{\var}{\sitm}$.
		By \ih (point 2), there is a derivation $\concl{\tderivtwo}{\typctxtwo, \var \hastype \mtypetwo}{\sfiretwo}{\mtype}$ where 
		$\mtype$ is a   \rightsh multi type and $\typctxtwo, \var \hastype \mtypetwo$ is a   \leftsh type context;  
		in particular, $\mtypetwo$ is a   \leftsh multi type.
		By \ih (point 1), there is a derivation $\concl{\tderivthree}{\typctxthree}{\sitm}{\mtypetwo}$ where $\typctxthree$ is a   \leftsh type context and $\mtypetwo$ is a   \leftsh multi type.
		We have the following derivation:
		\begin{equation*}
		\tderiv = 
		\begin{prooftree}
		\hypo{}
		\ellipsis{$\tderivtwo$}{\typctxtwo, \var \hastype \mtypetwo \vdash \sfiretwo \hastype \mtype}
		\hypo{}
		\ellipsis{$\tderivthree$}{\typctxthree \vdash \sitm \hastype \mtypetwo}
		\infer2[\footnotesize$\ruleAp$]{\typctx \vdash \sfiretwo \esub{\var}{\sitm} \hastype \mtype}
		\end{prooftree}
		\end{equation*}
		where $\typctx = \typctxtwo \uplus \typctxthree$ is a   \leftsh by 
		\Cref{rmk:merge-split-coshrinking}.
		\qedhere
	\end{itemize}	
\end{proof}

\begin{theorem}[Shrinking completeness for $\toess$]
	\label{thmappendix:completeness}
	\NoteState{thm:completeness}
	Let $\deriv \colon \tm \toess^* \sfire$ be a $\esssym$-normalizing evaluation.
	Then there is a shrinking derivation $\derive{\tderiv}{\tm}$.
\end{theorem}

\begin{proof}
	By induction on the length $\size{\deriv}$ of the $\esssym$-evaluation $\deriv$.
	
	If $\size{\deriv} = 0$ then $\tm = \sfire$ is $\esssym$-normal and hence a strong fireball.
	According to typability of strong fireballs (\Cref{prop:typability-normal}), there is a shrinking derivation $\concl{\tderiv}{\typctx}{\tm}{\mtype}$.
	
	Otherwise, $\size{\deriv} > 0$ and $\deriv$ is the concatenation of a first step $\tm \toess \tmtwo$ and an evaluation $\deriv' \colon \tmtwo \toess^* \fire$.
	By \ih, there is a shrinking derivation $\concl{\tderivtwo}{\typctx}{\tmtwo}{\mtype}$. 
	By subject expansion (\refpropp{qual-subject}{expansion}), there is a shrinking derivation 
	$\concl{\tderiv}{\typctx}{\tm}{\mtype}$.\qedhere 
\end{proof}

}
\end{document}